\documentclass[manuscript,nonacm]{acmart}
\usepackage[utf8]{inputenc}

\usepackage{amsthm}
\usepackage{amsmath}
\newtheorem{theorem}{Theorem}[section]
\newtheorem{lemma}[theorem]{Lemma}
\newtheorem{corollary}[theorem]{Corollary}

\newtheorem{claim}{Claim}
\newtheorem{definition}{Definition}

\usepackage{mathtools}

\usepackage{tikz}
\usepackage{graphicx}
\graphicspath{ {./images/} }

\tikzset{every loop/.style={}}
\tikzstyle{vertex} = [draw, circle, fill, inner sep=1pt]
\tikzstyle{bndvertex} = [draw, rectangle, minimum width=2pt, minimum height=2pt]
\tikzstyle{essvertex} = [bndvertex, fill=red]
\tikzstyle{essedge} = [line width=1pt, red]
\usetikzlibrary{calc}

\usepackage{subcaption}
\usepackage{enumitem}
\usepackage{xspace}
\usepackage{xcolor}
\usepackage{xparse}
\usepackage{xfrac}
\usepackage{todonotes}
\usepackage{textpos}

\usepackage{hyperref}

\newcommand{\Inline}[1]{#1\xspace}
\newcommand{\ProblemFormat}[1]{\textsc{#1}}
\newcommand{\ProblemName}[1]{\Inline{\ProblemFormat{#1}}}
\newcommand{\probFVS}{\ProblemName{Feedback Vertex Set}}

\newcommand{\probConnectedFVS}{\ProblemName{Connected Feedback Vertex Set}}
\newcommand{\probIndependentFVS}{\ProblemName{Independent Feedback Vertex Set}}
\newcommand{\probTreeDeletion}{\ProblemName{Tree Deletion Set}}

\newcommand{\quo}[2]{\sfrac{#1}{#2}}
\newcommand{\msotwo}{\Inline{\ensuremath{\mathsf{CMSO}_2}}}
\newcommand{\fvs}[1]{\mathrm{fvs}(#1)}
\newcommand{\lemmaA}{\Inline{Contraction Lemma}}
\newcommand{\lemmaB}{\Inline{Downgrade Lemma}}

\newcommand{\floor}[1]{\lfloor #1 \rfloor}

\newcommand{\mpara}[1]{\bigskip\noindent\underline{#1}\hspace{0.2cm}}
\newcommand{\cal}[1]{\mathcal{#1}}

\newcommand{\Oh}[2][]{{\mathcal{O}}_{#1}(#2)}

\newcommand{\A}{\mathbb{A}}
\newcommand{\B}{\mathbb{B}}
\newcommand{\C}{\mathbb{C}}
\newcommand{\D}{\mathbb{D}}
\newcommand{\F}{\mathbb{F}}
\newcommand{\N}{\mathbb{N}}
\newcommand{\X}{\mathbb{X}}
\newcommand{\Y}{\mathbb{Y}}
\newcommand{\Yf}{\Y}

\newcommand{\Cc}{\mathcal{C}}
\newcommand{\Fc}{\mathcal{F}}
\newcommand{\Ff}{\Fc}

\newcommand{\Pc}{\mathcal{P}}
\newcommand{\Tc}{\mathcal{T}}

\newcommand{\Path}{\rightsquigarrow}
\newcommand{\wt}[1]{\widetilde{#1}}
\newcommand{\wh}[1]{\widehat{#1}}

\newcommand{\Contract}{\mathsf{Contract}}
\newcommand{\Downgrade}{\mathsf{Downgrade}}

\def\cqedsymbol{\ifmmode$\lrcorner$\else{\unskip\nobreak\hfil
\penalty50\hskip1em\null\nobreak\hfil$\lrcorner$
\parfillskip=0pt\finalhyphendemerits=0\endgraf}\fi} 
\newcommand{\cqed}{\renewcommand{\qed}{\cqedsymbol}}

\newcommand{\bnd}{\partial}

\newcommand{\Types}{\mathsf{Types}}
\newcommand{\tp}{\mathsf{tp}}
\newcommand{\forget}{\mathsf{forget}}

\renewcommand{\leq}{\leqslant}
\renewcommand{\geq}{\geqslant}

\renewcommand{\le}{\leqslant}

\begin{document}

\title{Maintaining \msotwo Properties on Dynamic Structures With Bounded Feedback Vertex Number}

\thanks{This work is a part of project BOBR that has received funding from the European Research Council (ERC) under the European Union's Horizon 2020 research and innovation programme (grant agreement No. 948057).

An~extended abstract of this article has appeared at the 40th International Symposium on Theoretical Aspects of Computer Science (STACS 2023).}

\author{Konrad Majewski}
\email{k.majewski@mimuw.edu.pl}
\orcid{0000-0002-3922-7953}
\author{Michał Pilipczuk}
\email{michal.pilipczuk@mimuw.edu.pl}
\orcid{0000-0001-7891-1988}
\affiliation{%
	\institution{Institute of Informatics, University of Warsaw}
	\streetaddress{Banacha 2}
	\postcode{02-097}
	\city{Warsaw}
	\country{Poland}
}

\author{Marek Sokołowski}
\email{msokolow@mpi-inf.mpg.de}
\email{marek.sokolowski@mimuw.edu.pl}
\orcid{0000-0001-8309-0141}
\affiliation{%
	\institution{Max Planck Institute for Informatics, Saarland Informatics Campus}
	\streetaddress{Campus, Stuhlsatzenhausweg E1 4}
	\postcode{66123}
	\city{Saarbr\"{u}cken}
	\country{Germany}
}
\authornote{This work was created during this author's work at the Institute of Informatics at the University of Warsaw.}



\begin{abstract}
Let $\varphi$ be a sentence of $\msotwo$ (monadic second-order logic with quantification over edge subsets and counting modular predicates) over the signature of graphs. We present a dynamic data structure that for a given graph $G$ that is updated by edge insertions and edge deletions, maintains whether $\varphi$ is satisfied in $G$.
The data structure is required to correctly report the outcome only when the feedback vertex number of $G$ does not exceed a fixed constant $k$, otherwise it reports that the feedback vertex number is too large. With this assumption, we guarantee amortized update time $\Oh[\varphi,k]{\log n}$.
If we additionally assume that the feedback vertex number of $G$ never exceeds $k$, this update time guarantee is worst-case.

By combining this result with a classic theorem of Erd\H{o}s and P\'osa, we give a fully dynamic data structure that maintains whether a graph contains a packing of $k$ vertex-disjoint cycles with amortized update time $\Oh[k]{\log n}$. Our data structure also works in a larger generality of relational structures over binary signatures. 

\end{abstract}

\maketitle
\begin{textblock}{20}(-1.6, 4.5)
	\includegraphics[width=40px]{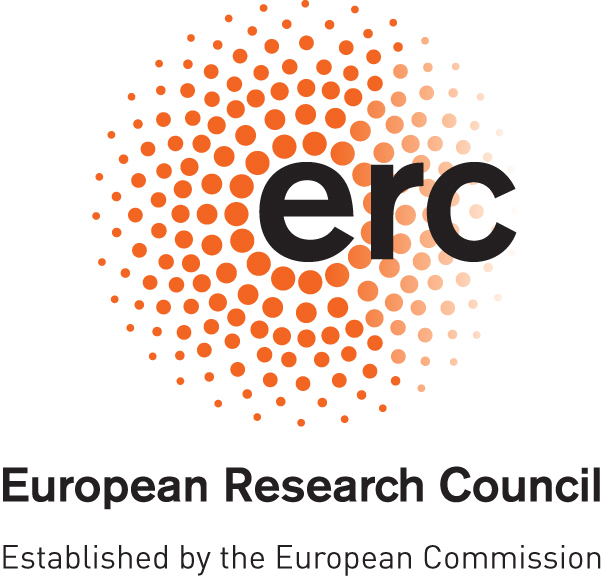}%
\end{textblock}
\begin{textblock}{20}(-2.05, 4.9)
	\includegraphics[width=60px]{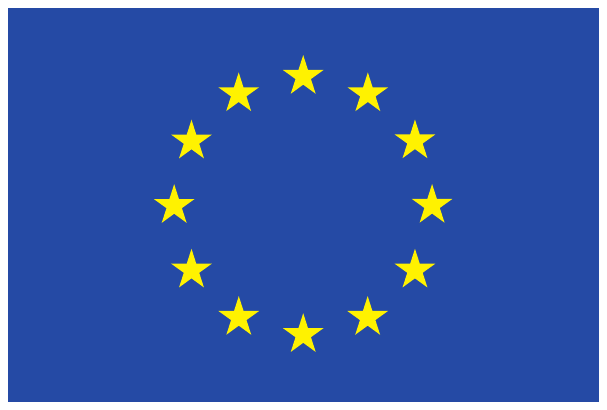}%
\end{textblock}

\section{Introduction}%
\label{sec:introduction}

We consider data structures for graphs in a fully dynamic model, where the considered graph can be updated by the following operations: add an edge, remove an edge, add an isolated vertex, and remove an isolated vertex. Most of the contemporary work on data structures for graphs focuses on problems that in the static setting are polynomial-time solvable, such as connectivity or distance computation. In this work we follow a somewhat different direction and consider {\em{parameterized problems}}. That is, we consider problems that are $\mathsf{NP}$-hard in the classic sense, even in the static setting, and we would like to design efficient dynamic data structures for them. The update time guarantees will typically depend on the size of the graph $n$ and a parameter of interest $k$, and the goal is obtain as good dependence on $n$ as possible while allowing exponential (or worse) dependence on $k$. The idea behind this approach is that the data structure will perform efficiently on instances where the parameter $k$ is small, which is exactly the principle assumed in the field of parameterized complexity.

The systematic investigation of such {\em{parameterized dynamic data structures}} was initiated by Alman et al.~\cite{AlmanMW20}, though a few earlier results of this kind can be found in the literature, e.g.~\cite{DvorakT13,DvorakKT14,IwataO14}. Alman et al. revisited several techniques in parameterized complexity and developed their dynamic counterparts, thus giving suitable parameterized dynamic data structures for a number of classic problems, including {\sc{Vertex Cover}}, {\sc{Hitting Set}}, {\sc{$k$-Path}}, and {\sc{Feedback Vertex Set}}. The last example is important for our motivation. Recall that a {\em{feedback vertex set}} in an (undirected) graph $G$ is a subset of vertices that intersects every cycle in $G$, and the {\em{feedback vertex number}} of $G$ is the smallest size of a feedback vertex set in $G$. The data structure of Alman et al. monitors whether the feedback vertex number of a dynamic graph $G$ is at most $k$ (and reports a suitable witness, if so) with amortized update time $2^{\Oh{k\log k}}\cdot \log n$.

Dvo\v{r}\'ak et al.~\cite{DvorakKT14} and, more recently, Chen et al.~\cite{ChenCDFHNPPSWZ21} studied parameterized dynamic data structures for another graph parameter {\em{treedepth}}. Formally, the treedepth of a graph $G$ is the least possible height of an {\em{elimination forest}} $G$: a rooted forest on the vertex set of $G$ such that every edge of $G$ connects a vertex with its ancestor. Intuitively, that a graph $G$ has treedepth $d$ means that $G$ has a tree decomposition whose {\em{height}} is $d$, rather than width. Chen et al.~\cite{ChenCDFHNPPSWZ21} proved that in a dynamic graph of treedepth at most $d$, an optimum-height elimination forest can be maintained with update time $2^{\Oh{d^2}}$ (worst case, under the promise that the treedepth never exceeds $d$). This improved upon the earlier result of Dvo\v{r}\'ak et al.~\cite{DvorakKT14}, who for the same problem achieved update time $f(d)$ for a non-elementary function $f$.

As already observed by Dvo\v{r}\'ak et al.~\cite{DvorakKT14}, such a data structure can be used not only to the concrete problem of computing the treedepth, but more generally to maintaining satisfiability of any property that can be expressed in the {\em{Monadic Second-Order logic}} $\mathsf{MSO}_2$. This logic extends standard First-Order logic $\mathsf{FO}$ by allowing quantification over subsets of vertices and subsets of edges, so it is able to express through constant-size sentences $\mathsf{NP}$-hard problems such at Hamiltonicity or $3$-colorability. More precisely, the following result was proved by Dvo\v{r}\'ak et al.~\cite{DvorakKT14} (see Chen et al.~\cite{ChenCDFHNPPSWZ21} for lifting the promise of boundedness of treedepth).

\begin{theorem}[\cite{DvorakKT14,ChenCDFHNPPSWZ21}]\label{thm:td-mso}
 Given an $\mathsf{MSO}_2$ sentence $\varphi$ over the signature of graphs and $d\in \N$, one can construct a dynamic data structure that maintains whether a given dynamic graph $G$ satisfies $\varphi$. The data structure is obliged to report a correct answer only when the treedepth of $G$ does not exceed $d$, and otherwise it reports {\em{Treedepth too large}}. The updates work in amortized time $f(\varphi,d)$ for a computable function $f$, under the assumption that one is given access to a dictionary on the edges of $G$ with constant-time operations.
\end{theorem}

The proof of Theorem~\ref{thm:td-mso} is based on the following idea. If a graph $G$ is supplied with an elimination forest of bounded depth, then, by the finite-state properties of $\mathsf{MSO}_2$, whether $\varphi$ is satisfied in $G$ can be decided using a suitable bottom-up dynamic programming algorithm. Then it is shown that when $G$ is updated by edge insertions and removals, one is able to maintain not only an optimum-height elimination forest $F$ of $G$, but also a run of this dynamic programming algorithm on $F$. This blueprint brings the classic work on algorithmic meta-theorems in parameterized complexity to the setting of dynamic data structures, by showing that dynamic maintenance of a suitable decomposition is a first step to maintaining all properties that can be efficiently computed using this decomposition.

Notably, Chen et al.~\cite{ChenCDFHNPPSWZ21} apply this principle to two specific problems of interest: detection of $k$-paths and $(\geq k)$-cycles in undirected graphs. Using known connections between these objects and treedepth, they gave dynamic data structures for the detection problems that have update time $2^{\Oh{k^2}}$ for $k$-paths (assuming a dictionary on edges) and $2^{\Oh{k^4}}\cdot \log n$ for $(\geq k)$-cycles.

One of the main questions left open by the work of Dvo\v{r}\'ak et al.~\cite{DvorakKT14} and by Chen et al.~\cite{ChenCDFHNPPSWZ21} was whether in a dynamic graph of treewidth at most $k$ it is possible to maintain a tree decomposition of width at most $g(k)$ with polylogarithmic update time. Note here that the setting of tree decompositions is the natural context in which $\mathsf{MSO}_2$ on graphs is considered, due to Courcelle's Theorem~\cite{Courcelle90}, while the treedepth of a graph is always an upper bound on its treewidth. Thus, the works of Dvo\v{r}\'ak et al.~\cite{DvorakKT14} and of Chen et al.~\cite{ChenCDFHNPPSWZ21} can be regarded as partial progress towards resolving this question, where a weaker (larger) parameter treedepth is considered.

\paragraph*{Our contribution.} We approach the question presented above from another direction, by considering feedback vertex number --- another parameter that upper-bounds the treewidth. As mentioned, Alman et al.~\cite{AlmanMW20} have shown that there is a dynamic data structure that monitors whether the feedback vertex number is at most $k$ with update time $2^{\Oh{k\log k}}\cdot \log n$. We extend this result by showing that in fact, every $\msotwo$-expressible property can be efficiently maintained in graphs of bounded feedback vertex number. Here is our main result.

\begin{theorem}\label{thm:main-graph}
 Given a sentence $\varphi$ of $\msotwo$ over the signature of graphs and $k\in \N$, one can construct a data structure that maintains whether a given dynamic graph $G$ satisfies $\varphi$. The data structure is obliged to report a correct answer only if the feedback vertex number of $G$ is at most $k$, otherwise it reports {\em{Feedback vertex number too large}}. The graph is initially empty and the amortized update time is $f(\varphi,k)\cdot \log n$, for some computable function $f$.
\end{theorem}

Here, $\msotwo$ is an extension of $\mathsf{MSO}_2$ by modular counting predicates; this extends the generality slightly.

Similarly as noted by Chen et al.~\cite{ChenCDFHNPPSWZ21}, the appearance of the $\log n$ factor in the update time seems necessary: a data structure like the one in Theorem~\ref{thm:main-graph} could be easily used for connectivity queries in dynamic forests, for which there is an $\Omega(\log n)$ lower bound in the cell-probe model~\cite{patrascu_lower_2004}.

With an~additional promise of boundedness of feedback vertex number of the maintained graph, we actually prove the worst-case complexity bounds:
\begin{theorem}\label{thm:main-graph-promise}
 Given a sentence $\varphi$ of $\msotwo$ over the signature of graphs and $k\in \N$, one can construct a data structure that maintains whether a given dynamic graph $G$ satisfies $\varphi$.
 Assuming that at all times, $G$ has feedback vertex number at most $k$, the worst-case update time is $f(\varphi,k)\cdot \log n$, for some computable function $f$.
\end{theorem}

We prove Theorems~\ref{thm:main-graph} and~\ref{thm:main-graph-promise} in a larger generality of relational structures over binary signatures, see Theorem~\ref{thm:main_theorem} for a formal statement. More precisely, we consider relational structures over signatures consisting of relation symbols of arity at most $2$ that can be updated by adding and removing tuples from the relations, and by adding and removing isolated elements of the universe. In this language, graphs correspond to structures over a signature consisting of one binary relation signifying adjacency. As feedback vertex number we consider the feedback vertex number of the Gaifman graph of the structure. Generalization to relational structures is not just a mere extension of Theorem~\ref{thm:main-graph}, it is actually a formulation that appears naturally in the inductive strategy that is employed in the proof.

As for this proof, we heavily rely on the approach used by Alman et al.~\cite{AlmanMW20} for monitoring the feedback vertex number. This approach is based on applying two types of simplifying operations, in alternation and a bounded number of times:
\begin{itemize}[nosep]
 \item contraction of subtrees in the graph; and
 \item removal of high-degree vertices.
\end{itemize}
We prove that in both cases, while performing the simplification it is possible to remember a bounded piece of information about each of the simplified parts, thus effectively enriching the whole data structure with information from which the satisfaction of $\varphi$ can be inferred. Notably, for the contracted subtrees, this piece of information is the $\msotwo$-type of appropriately high rank. To maintain these types in the dynamic setting, we use the top trees data structure of Alstrup et al.~\cite{AlstrupHLT05}.
All in all, while our data structure is based on the same combinatorics of the feedback vertex number, it is by no means a straightforward lift of the work of Alman et al.~\cite{AlmanMW20}: enriching the data structure with information about types requires several new ideas and insights, both on the algorithmic and on the logical side of the reasoning. A~more extensive discussion can be found in Section~\ref{sec:overview}.

\paragraph*{Applications.} Similarly as in the work of Chen et al.~\cite{ChenCDFHNPPSWZ21}, we observe that Theorem~\ref{thm:main-graph} can be used to obtain dynamic data structures for specific parameterized problems through a win/win approach. Consider the {\em{cycle packing number}} of a graph $G$: the maximum number of vertex-disjoint cycles that can be found in $G$. A classic theorem of Erd\H{o}s and P\'osa~\cite{erdosPosa} states that there exists a universal constant $c$ such that if the feedback vertex number of a graph $G$ is larger than $c\cdot p\log p$, then the cycle packing number of $G$ is at least $p$. We can use this result to establish the following.

\begin{theorem}\label{thm:EP}
 For a given $p\in \N$ one can construct a dynamic data structure that for a dynamic graph $G$ (initially empty) maintains whether the cycle packing number of $G$ is at least $p$. The amortized update time is $f(p)\cdot \log n$, for a computable function $f$. 
\end{theorem}
\begin{proof}
 For a given $p$, it is straightforward to write a $\msotwo$ sentence $\varphi_p$ that holds in a graph $G$ if and only if $G$ contains $p$ vertex-disjoint cycles. Then we may use the data structure of Theorem~\ref{thm:main-graph} for $\varphi_p$ and $k=c\cdot p\log p$, where $c$ is the constant given by the theorem of Erd\H{o}s and P\'osa~\cite{erdosPosa}. Note that if this data structure reports that {\em{Feedback vertex number too large}}, then the cycle packing number is at least $p$, so this outcome can be reported. 
\end{proof}

The same principle can be applied to other problems related to constrained variants of feedback vertex sets, e.g. {\sc{Connected Feedback Vertex Set}}, {\sc{Independent Feedback Vertex Set}}, and {\sc{Tree Deletion Set}}. 
We say that a feedback vertex set $S$ is
\begin{itemize}[nosep]
  \item {\em{connected}} if the induced subgraph $G[S]$ is connected;
  \item {\em{independent}} if the induced subgraph $G[S]$ is edgeless; and
  \item a {\em{tree deletion set}} if $G-S$ is connected (i.e., is a tree).
\end{itemize}
The parameterized complexity of corresponding problems \probConnectedFVS, \probIndependentFVS, and \probTreeDeletion was studied in~\cite{GiannopoulouLSS16,MisraConnected,MisraPRS12,RamanSS13,LiP20}.

Similarly as in~Theorem~\ref{thm:EP}, we can focus on the dynamic versions of the above problems, where the graph~$G$ is initially empty, and at each step we either insert an~edge or an~isolated vertex to~$G$, or an edge or an isolated vertex from~$G$.

\begin{theorem}
 For a given $p\in \N$ one can construct a dynamic data structure that for a dynamic graph $G$ (initially empty) maintains whether $G$ contains the following objects:
 \begin{itemize}[nosep]
  \item a connected feedback vertex set of size at most $p$;
  \item an independent feedback vertex set of size at most $p$; and
  \item a tree deletion set of size at most $p$.
 \end{itemize}
 The amortized update time is $f(p)\cdot \log n$, for a computable function $f$. 
\end{theorem}
\begin{proof}
 Observe that for a given $p\in \N$, we can write $\msotwo$ sentences $\varphi_1,\varphi_2,\varphi_3$ over the signature of graphs that respectively express the properties of having a connected feedback vertex set of size at most $p$, having an independent feedback vertex set of size at most $p$, and having a tree deletion set of size at most $p$. Hence, we can use three instances of the data structure of Theorem~\ref{thm:main-graph}, applied to sentences $\varphi_1,\varphi_2,\varphi_3$, respectively, and each with parameter $p$. Note that if these data structures report that {\em{Feedback vertex number too large}}, then there is no feedback vertex set of size at most $p$, so in particular no connected or independent feedback vertex set or a tree deletion set of size at most $p$. So a negative answer to all three problems can be reported then.
\end{proof}

\paragraph*{Follow-up work: dynamic treewidth.}
After the publication of the conference version of this work~\cite{DBLP:conf/stacs/MajewskiPS23}, Korhonen et al.~\cite{DBLP:conf/focs/KorhonenMNP023}, devised an~efficient analog of Theorems~\ref{thm:td-mso} and~\ref{thm:main-graph} for treewidth:
\begin{theorem}[\cite{DBLP:conf/focs/KorhonenMNP023}]\label{thm:tw-mso}
 Given a $\msotwo$ sentence $\varphi$ over the signature of graphs and $k\in \N$, one can construct a dynamic data structure that maintains whether a given dynamic graph $G$ of size $n$ satisfies $\varphi$. The data structure is obliged to report a correct answer only when the treewidth of $G$ does not exceed $k$, and otherwise it reports {\em{Treewidth too large}}. The updates work in amortized time $f(\varphi,k) \cdot n^{o(1)}$ for a computable function $f$.
\end{theorem}

The proof of Theorem~\ref{thm:tw-mso} is structured similarly to its treedepth counterpart: whenever we have an~access to a~tree decomposition of $G$ of near-optimum width, we can verify the satisfaction of $\varphi$ using a~finite-state bottom-up dynamic programming algorithm due to a~seminal work of Courcelle~\cite{Courcelle90}.
The authors then show that such a~decomposition of width at most $6k + 5$, together with a~run of the bottom-up dynamic programming scheme on the decomposition, can indeed be maintained in amortized time $f(\varphi, k) \cdot 2^{\Oh{\sqrt{\log n} \log \log n}} \in f(\varphi,k) \cdot n^{o(1)}$ per update.

Since both treedepth and feedback vertex number are upper bounds for treewidth of a~graph, the result of Theorem~\ref{thm:tw-mso} can be seen as a~generalization of Theorems~\ref{thm:td-mso} and~\ref{thm:main-graph}. However, this generalization comes at an~expense of a~noticeably worse update time, depending heavily on the size $n$ of the graph.
It still remains an~open question whether this dependence can be improved to $\log^{\Oh{1}} n$ or even $\log n$.
Currently, this is only known to be possible for graphs of treewidth at most $2$: Bodlaender showed that a~near-optimum-width tree decomposition of such graphs can be maintained in worst-case logarithmic time~\cite{Bodlaender93a}.
Note however that the result of Bodlaender is incomparable to Theorems~\ref{thm:main-graph}, \ref{thm:main-graph-promise}: there exist graphs of treewidth $2$ and arbitrarily large feedback vertex number, and there already exist graphs of feedback vertex number $2$ whose treewidth is strictly larger than $2$.


\newcommand{\MSO}{\mathsf{MSO}}

\section{Overview}\label{sec:overview}

In this section we present an overview of the proof of Theorem~\ref{thm:main-graph}. We deliberately keep the description high-level in order to convey the main ideas. In particular, we focus on the graph setting and delegate the notation-heavy aspects of relational structures to the full exposition.

Let $G$ be the given dynamic graph. 
We focus on the model where we have a promise that the feedback vertex number of $G$ is at most $k$ at all times. If we are able to construct a data structure in this promise model, then it is easy to lift this to the full model described in Theorem~\ref{thm:main-graph} using the standard technique of {\em{postponing invariant-breaking insertions}}. This technique was also used by Chen et al.~\cite{ChenCDFHNPPSWZ21} and dates back to the work of Eppstein et al.~\cite{EppsteinGIS96}.

\paragraph*{Colored graphs.} We will be working with edge- and vertex-colored graphs. That is, if $\Sigma$ is a finite set of colors (a {\em{palette}}), then a {\em{$\Sigma$-colored graph}} is a graph where every vertex and edge is assigned a color from $\Sigma$. In our case, all the palettes will be of size bounded by functions of $k$ and the given formula $\varphi$, but throughout the reasoning we will use different (and rapidly growing) palettes. For readers familiar with relational structures, in general we work with relational structures over binary signatures (involving symbols of arity $0,1,2$), which are essentially colored graphs supplied with flags.

Thus, we assume that the maintained dynamic graph $G$ is also a $\Sigma$-colored graph for some initial palette $\Sigma$. When $G$ is updated by a vertex or edge insertion, we assume that the color of the new feature is provided with the update.

\paragraph*{Monadic Second-Order Logic.} $\MSO_2$ is the Monadic Second-Order logic with quantification over vertex subsets and edge subsets. This is a standard logic considered in parameterized complexity in connection with treewidth and Courcelle's Theorem. We refer to~\cite[Section~7.4]{platypus} for a thorough introduction, and explain here only the main features. There are four types of variables: {\em{individual}} vertex/edge variables that evaluate to single vertices/edges, and {\em{monadic}} vertex/edge variables that evaluate to vertex/edge subsets. These can be quantified both existentially and universally. One can check equality of vertices/edges, incidence between an edge and a vertex, and membership of a vertex/edge to a vertex/edge subset. In case of colored graphs, one can also check colors of vertices/edges using unary predicates. Negation and all boolean connectives are allowed.

Note that in Theorem~\ref{thm:main-graph} we consider the variant $\msotwo$ of Monadic Second-Order logic, which is an extension of the above by modular counting predicates that can be applied to monadic variables. For simplicity, we ignore this extension for the purpose of this overview.

\paragraph*{Types.}
The key technical ingredient in our reasoning are {\em{types}}, which is a standard tool in model theory. Let $G$ be a $\Sigma$-colored graph and $q$ be a nonnegative integer. With $G$ we can associate its {\em{rank-$q$ type}} $\tp^q(G)$, which is a finite piece of data that contains all information about the satisfaction of $\MSO_2$ sentences of quantifier rank at most $q$ in $G$ (i.e., with quantifier nesting bounded by $q$). More precisely:
\begin{itemize}[nosep]
 \item For every choice of $q$ and $\Sigma$ there is a finite set $\Types^{q,\Sigma}$ containing all possible rank-$q$ types of $\Sigma$-colored graphs. The size of $\Types^{q,\Sigma}$ depends only on $q$ and $\Sigma$.
 \item For every $\MSO_2$ sentence $\psi$ of quantifier rank at most $q$, the type $\tp^q(G)$ uniquely determines whether $\psi$ holds in $G$. 
\end{itemize}

In addition to the above, we also need an understanding that types are {\em{compositional}} under gluing of graphs along small boundaries. For this, we work with the notion of a {\em{boundaried graph}}, which is a graph $G$ together with a specified subset of vertices $\bnd G$, called the {\em{boundary}}. Typically, these boundaries will be of constant size.
We extend the notion of a type to boundaried graphs, where the rank-$q$ type $\tp^q(G)$ of a boundaried graph $G$ contains information not only about all rank-$q$ $\MSO_2$ sentences satisfied in $G$, but also about all such sentences that in addition can use the vertices of $\bnd G$ as parameters (one can also think that vertices of $\bnd G$ are given through free variables). Again, for every finite set $D$, there is a finite set of possible types $\Types^{q,\Sigma}(D)$ of boundaried $\Sigma$-colored graphs with boundary $D$, and the size of $\Types^{q,\Sigma}(D)$ depends only on $q$, $\Sigma$, and $|D|$.

Now, on boundaried graphs there are two natural operations. First, if $G$ is a boundaried graph and $u\in \bnd G$, then one can {\em{forget}} $u$ in $G$. This yields a boundaried graph $\forget(G,u)$ obtained from $G$ by removing $u$ from the boundary (otherwise the graph remains intact). Second, if $G$ and $H$ are two boundaried graphs and $\xi$ is a partial bijection between $\bnd G$ and $\bnd H$, then the {\em{join}} $G\oplus_\xi H$ is the boundaried graph obtained from the disjoint union of $G$ and $H$ by identifying vertices that correspond to each other in $\xi$; the new boundary is the union of the old boundaries (with identification applied). 

With these notions in place, the compositionality of types can be phrased as follows:
\begin{itemize}[nosep]
 \item Given $\tp^q(G)$ and $u\in \bnd G$, one can uniquely determine $\tp^q(\forget(G,u))$.
 \item Given $\tp^q(G)$ and $\tp^q(H)$ and a partial bijection $\xi$ between the boundaries of $G$ and $H$, one can uniquely determine $\tp^q(G\oplus_\xi H)$.
\end{itemize}
The determination described above is effective, that is, can be computed by an algorithm.
 
\paragraph*{Top trees.} We now move to the next key technical ingredient: the {\em{top trees}} data structure of Alstrup et al.~\cite{AlstrupHLT05}. Top trees work over a dynamic forest $F$, which is updated by edge insertions and deletions (subject to the promise that no update breaks acyclicity) and insertions and deletions of isolated vertices. For each connected component $T$ of $F$ one maintains a {\em{top tree}}~$\Delta_T$, which is a hierarchical decomposition of $T$ into {\em{clusters}}. Each cluster $S$ is a subtree of $T$ with at least one edge that is assigned a boundary $\bnd S\subseteq V(S)$ of size at most $2$ with the following property: every vertex of $S$ that has a neighbor outside of $S$ belongs to $\bnd S$. Formally, the top tree $\Delta_T$ is a binary tree whose nodes are assigned clusters in $T$ so that:
\begin{itemize}[nosep]
 \item the root of $\Delta_T$ is assigned the cluster $(T,\bnd T)$, where $\bnd T$ is a choice of at most two vertices in $T$;
 \item the leaves of $\Delta_T$ are assigned single-edge clusters;
 \item for every internal node $x$ of $\Delta_T$, the edge sets of clusters in the children of $x$ form a partition of the edge set of the cluster at $x$.
\end{itemize}
Note that the last property implies that the cluster at $x$, treated as a boundaried graph, can be obtained from the two clusters at the children of $x$ by applying the join operation, possibly followed by forgetting a subset of the boundary. We will then say that the cluster at $x$ is obtained by {\em{joining}} the two clusters at its children.

In~\cite{AlstrupHLT05}, Alstrup et al. showed how to maintain, for a dynamic forest $F$, a forest of top trees $\{\Delta_T\colon T\textrm{ is a component of }F
\}$ so that each tree $\Delta_T$ has depth $\Oh{\log n}$ and every operation is performed in worst-case time $\Oh{\log n}$. Moreover, they showed that the top trees data structure can be robustly enriched with various kinds of auxiliary information about clusters, provided this information can be efficiently composed upon joining clusters. More precisely, suppose that with each cluster $C$ we can associate a piece of information ${\cal I}(C)$ so that
\begin{itemize}[nosep]
 \item ${\cal I}(C)$ can be computed in constant time when $C$ has one edge; and
 \item if $C$ is obtained by joining two clusters $C_1$ and $C_2$, then from ${\cal I}(C_1)$ and ${\cal I}(C_2)$ one can compute ${\cal I}(C)$ in constant time.
\end{itemize}
 Then, as shown in~\cite{AlstrupHLT05}, with each cluster $C$ one can store the corresponding piece of information ${\cal I}(C)$, and still perform updates in time $\Oh{\log n}$.
 
 In our applications, we work with top trees over dynamic $\Sigma$-colored forests, where with each cluster $C$ we store information on its type:
 $${\cal I}(C)=\tp^p(C)$$
 for a suitably chosen $p\in \N$. Here, for technical reasons we need to be careful about the colors: the type $\tp^p(C)$ takes into account the colors of all the edges of $C$ and all the vertices of $C$ {\em{except}} the vertices of $\bnd C$ (formally, we consider the type of $C$ with colors stripped from boundary vertices). The rationale behind this choice is that a single vertex $u$ can participate in the boundary of multiple clusters, hence in the dynamic setting we cannot afford to update the type of each of them upon updating the color of $u$. Rather, every cluster $C$ stores its type with the colors on $\bnd C$ stripped, and if we wish to compute the type of $C$ with these colors included, it suffices to look up those colors and update the stripped type (using compositionality).
 
 Brushing these technical details aside, after choosing the definitions right, the compositionality of types explained before perfectly fits the properties required from an enrichment of top trees. This means that with each cluster $C$ we can store $\tp^p(C)$ while guaranteeing worst-case update time $\Oh[p,\Sigma]{\log n}$. We remark that the combination of top trees and $\MSO_2$ types appears to be a novel contribution of this work; we hope that it can be reused in the~future.
 
 So if $F$ is a dynamic $\Sigma$-colored forest and $p$ is a parameter, then for each tree $T$ in $F$ we can maintain a top tree $\Delta_T$ whose root is supplied with the type $\tp^p(T)$. Knowing the multiset of rank-$p$ types of trees in $F$, we can use standard compositionality and idempotence of types to compute the type $\tp^p(F)$, from which in turn one can infer which rank-$p$ sentences  are satisfied in $F$. By taking $p$ to be the quantifier rank of a given sentence $\varphi$, we~obtain:
 
 \begin{theorem}[folklore] \label{thm:forests}
  Let $\Sigma$ be a finite palette and $\varphi$ be an $\MSO_2$ sentence over $\Sigma$-colored graphs. Then there is a dynamic data structure that for a dynamic $\Sigma$-colored forest $F$ maintains whether $\varphi$ holds in $F$. The worst-case update time is $\Oh[\varphi,\Sigma]{\log n}$.
 \end{theorem}

 Note that the statement of Theorem~\ref{thm:forests} matches (the colored version of) the statement of Theorem~\ref{thm:main-graph-promise} for $k=0$ and should be considered standard, see e.g.~\cite{CohenT97, Frederickson98} for similar results. In fact, a~work of Bodlaender~\cite{Bodlaender93a} implies Theorem~\ref{thm:forests} above and even extends it to the setting of dynamic graphs of treewidth at most $2$.
This work implements the data structure of Theorem~\ref{thm:forests} via a~different dynamic tree data structure, adapted from the parallel tree contraction algorithm by Miller and Reif~\cite{DBLP:journals/acr/MillerR89} and similar in design to topology trees of Frederickson~\cite{DBLP:journals/siamcomp/Frederickson85,DBLP:journals/siamcomp/Frederickson97,DBLP:journals/jal/Frederickson97}.
Bodlaender then amends his data structure to additionally support \emph{cactus graphs} -- graphs with all simple cycles being pairwise edge-disjoint -- and proceeds to show that graphs of treewidth $2$ \emph{interpret} in cactus graphs, and this interpretation can be maintained efficiently under edge updates.
Unfortunately, his approach does not seem to apply to graphs of larger treewidth, or even graphs of sufficiently large feedback vertex number.
 
%
 The problem of maintenance of $\MSO$ queries over dynamic forests has also been considered in the databases literature, see~\cite{Niewerth18} and references therein, however under a different (and somewhat orthogonal) set of allowed updates. 
 
\paragraph*{The data structure of Alman et al.~\cite{AlmanMW20}.}
Our goal now is to lift Theorem~\ref{thm:forests} to the case of $k>0$. For this we rely on the approach of Alman et al.~\cite{AlmanMW20} for monitoring the feedback vertex number, which is based on a sparsity-based strategy that is standard in parameterized complexity, see e.g.~\cite[Section~3.3]{platypus}. 

The approach is based on two lemmas. The first one concerns the situation when the graph contains a vertex $u$ of degree at most $2$. In this case, it is safe to dissolve $u$: either remove it, in case it has degree $0$ or $1$, or replace it with a new edge connecting its neighbors, in case it has degree $2$. 
Note that dissolving a degree-$2$ vertex naturally can create a multigraph. This creates technical issues both in~\cite{AlmanMW20} and in this work, but we shall largely ignore them for the purpose of this overview.
Formally, we have the following.

\begin{lemma}[folklore]\label{lem:diss}
 Dissolving a vertex of degree at most $2$ in a multigraph does not change the feedback vertex number.
\end{lemma}

The second lemma concerns the situation when the graph has minimum degree at least~$3$. Then a sparsity-based argument shows that every feedback vertex set of size at most $k$ intersects the set of $\Oh{k}$ vertices with highest degrees.

\begin{lemma}[Lemma 3.3 in~\cite{platypus}]\label{lem:high}
 Let $G$ be a multigraph with minimum degree $3$ and let $B$ be the set of $3k$ vertices with highest degrees in $G$. Then every feedback vertex set of size at most $k$ in $G$ intersects $B$.
\end{lemma}

Lemmas~\ref{lem:diss} and~\ref{lem:high} can be used to obtain an FPT algorithm for {\sc{Feedback Vertex Set}} with running time $(3k)^k\cdot (n+m)$ (see~\cite[Theorem~3.5]{platypus}): apply the reduction of Lemma~\ref{lem:diss} exhaustively, and then branch on which of the $3k$ vertices with highest degrees should be included in the solution. This results in a recursion tree of total size at most $(3k)^k$.

The data structure of Alman et al.~\cite{AlmanMW20} is based on dynamization of the branching algorithm presented above. There are two main challenges:
\begin{itemize}[nosep]
 \item dynamic maintenance of the sequence of dissolutions given by Lemma~\ref{lem:diss}; and
 \item dynamic maintenance of the set of high degree vertices.
\end{itemize}

For the first issue, it is explanatory to imagine performing the dissolutions not one by one iteratively, but all at once. It is not hard to see that the result of applying Lemma~\ref{lem:diss} exhaustively is that the input multigraph $G$ gets contracted to a multigraph $\Contract(G)$ in the following way: the edge set of $G$ is partitioned into disjoint trees, and each of them either disappears or is contracted into a single edge in $\Contract(G)$; see Figure~\ref{fig:fern-decomposition} for a visualization. (There may be some corner cases connected to loops in $\Contract(G)$ that result from contracting not trees, but unicyclic graphs; we ignore this issue in this overview.)
We call the elements of this partition {\em{ferns}}, and the corresponding decomposition of $G$ into ferns is called the {\em{fern decomposition}} of $G$. Importantly, the order of performing the contractions has no effect on the outcome, yielding always the same fern decomposition of $G$. 

With each fern of $S$ we can associate its boundary $\bnd S$, which is the set of vertices of $S$ incident to edges that lie outside of $S$. It is not hard to see that this boundary will always be of size $0$, $1$, or $2$. The ferns that correspond to edges in $\Contract(G)$ are the ferns with boundary of size $2$ (each such fern gets contracted to an edge connecting the two vertices of the boundary) and non-tree ferns with boundary of size $1$ (each such fern gets contracted to a loop at the unique vertex of the boundary).

\begin{figure}[t]
\centering
{
  \includegraphics[width=0.8\textwidth]{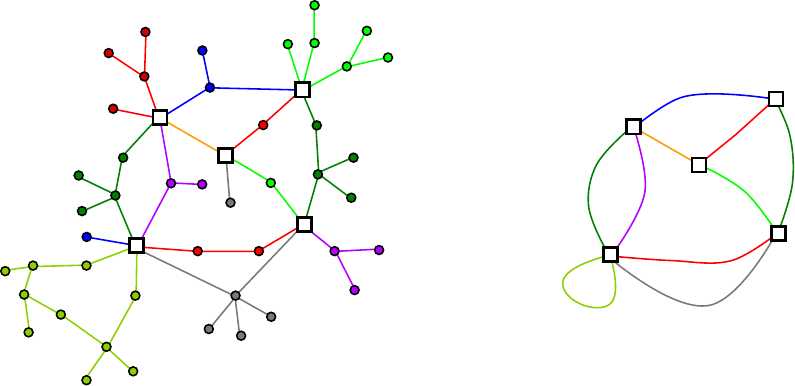}
}
\caption{Left: A graph $G$ together with its fern decomposition. Different ferns are depicted with different colors; these should not be confused with the coloring of edges of $G$ with colors from $\Sigma$. Right: The multigraph $\Contract(G)$ obtained by contracting each fern. Note that in the construction of $\Contract^p(G)$ described in the discussion of the Contraction Lemma, we would not have parallel edges or loops. Instead, each pack of parallel edges would be replaced by a single one, colored with the joint type of the whole pack. Similarly, loops on a vertex would be removed and their joint type would be stored in the color of the vertex.}
\label{fig:fern-decomposition}
\end{figure}

The idea of Alman et al.\ is to maintain the ferns in the fern decomposition using link-cut trees. It is shown that each update in $G$ affects the fern decomposition only slightly, in the sense that it can be updated using a constant number of operations on link-cut trees. In this way, the fern decomposition and the graph $\Contract(G)$ can be maintained with worst-case $\Oh{\log n}$ time per update in $G$. This resolves the first challenge.

For the second challenge, Alman et al.\ observe that if in Lemma~\ref{lem:high} one increases the number of highest degree vertices included in $B$ from $3k$ to $12k$, then the set remains ``valid'' --- in the sense of satisfying the conclusion of the lemma --- even after $\Oh{m/k}$ updates are applied to the graph. Here, $m$ denotes the number of edges of the graph on which Lemma~\ref{lem:high} is applied, which is $\Contract(G)$ in our case.
This means that it remains correct to perform a recomputation of the set $B$ only every $\Theta(m/k)$ updates. Since such a recomputation takes time $\Oh{m}$, the \emph{amortized} update time is $\Oh{k}$.
The work of Alman et al.\ only proves the amortized time complexity guarantee, but the \emph{worst-case} $\Oh{k}$ update time can actually be achieved in a~black-box way through the framework of \emph{global rebuilding} by Overmars and van~Leeuwen~\cite{DBLP:conf/wg/Overmars81,DBLP:journals/ipl/OvermarsL81a}: suppose that $B$ needs to be recomputed every $T \in \Theta(m/k)$ updates.
Then after $\frac12 T$ updates, take the current snapshot of $G$ --- call it $G_{\frac12 T}$ --- and find the set $B_{\frac12 T}$ of $12k$ highest degree vertices in it \emph{in background}.
This process requires $\Oh{m}$ time, but we distribute this computation over the following sequence of $\frac12 T$ queries, taking only $\Oh{m / T} = \Oh{k}$ additional time per update.
After a~total of $T$ queries, we dispose of the set $B$ and replace it with $B_{\frac12 T}$.
Then the new value of $B$ still remains ``valid'' for the following $\frac12 T$ queries, and, ignoring some technical details, every update takes worst-case $\Oh{k}$ time.


Once $\Contract(G)$ and $B\subseteq V(\Contract(G))$ are known, Lemma~\ref{lem:high} asserts that if the feedback vertex number of $G$ is at most $k$, there exists a vertex $b\in B$ whose deletion decreases the feedback vertex number. Therefore, the idea of Alman et al. is to construct a recursive copy of the data structure for each $b\in B$: the copy maintains the graph $\Contract(G)-b$ and uses parameter $k-1$ instead of $k$. Note that when $B$ gets recomputed, all these data structures need to be reset, but the framework of global rebuilding can be used to handle this step with worst-case time guarantees as well.

All in all, once one unravels the recursion, the whole construction is a tree of data structures of depth $k$ and branching $12k$, which is maintained with worst-case time $2^{\Oh{k\log k}}\cdot \log n$. The graph has feedback vertex number at most $k$ if and only if this tree contains at least one leaf with an empty graph.

\paragraph*{Our data structure.}
We now describe the high-level idea of our data structure.

Lemmas~\ref{lem:diss} and~\ref{lem:high} can be used not only to design an FPT algorithm for {\sc{Feedback Vertex Set}}, but also an approximation algorithm. Consider the following procedure: apply the reduction of Lemma~\ref{lem:diss} exhaustively, then greedily take {\em{all}} the $3k$ vertices with highest degrees to the constructed feedback vertex set, and iterate these two steps in alternation until the graph becomes empty. Lemma~\ref{lem:high} guarantees that provided the feedback vertex number was at most $k$ in the first place, the iteration terminates after at most $k$ steps; the $3k^2$ selected vertices form a feedback vertex set. We note that this application of Lemmas~\ref{lem:diss} and~\ref{lem:high} for feedback vertex set approximation is not new, for instance it was recently used by Kammer and Sajenko~\cite{KammerS20} in the context of space-efficient kernelization.

Our data structure follows the design outlined above. That is, instead of a tree of data structures, we maintain a sequence of $2k+2$ data structures, respectively for multigraphs
$$G_0,H_0,G_1,H_1,\ldots,G_k,H_k.$$
These multigraphs essentially satisfy the following:
\begin{itemize}[nosep]
 \item $G_0=G$;
 \item $H_i=\Contract(G_i)$ for $i=0,1,\ldots,k$; and
 \item $G_{i+1}=H_i-B_i$ for $i=0,1,\ldots,k-1$, where $B_i$ is a set that satisfies the conclusion of Lemma~\ref{lem:high} for $G_i$.
\end{itemize}
Note that these invariants imply that provided the feedback vertex number of $G$ is at most~$k$, the feedback vertex number of $G_i$ and of $H_i$ is at most $k-i$ for each $i\in \{0,1,\ldots,k\}$, implying that $G_k$ is a forest and $H_k$ is the empty graph.

The precise definitions of $\Contract(\cdot)$ and of deleting vertices used in the sequence above will be specified later. More precisely,
graphs $G_0,H_0,\ldots,G_k,H_k$ will be colored with palettes $\Sigma_0,\Gamma_0,\ldots,\Sigma_k,\Gamma_k$ in order, where $\Sigma_0=\Sigma$. These palettes will grow (quite rapidly) in sizes, but each will be always of size bounded in terms of $k$, $\Sigma$, and $q$ --- the quantifier rank of the fixed sentence $\varphi$ whose satisfaction we monitor. The idea is that when obtaining $H_i$ from $G_i$ by contracting ferns, we use colors from $\Gamma_i$ to store information about the contracted ferns on edges and vertices of $H_i$. Similarly, when removing vertices of $B_i$ from $H_i$ to obtain $G_{i+1}$, we use colors from $\Sigma_{i+1}$ on vertices of $G_{i+1}$ to store information about the adjacencies of the removed vertices. These steps are encompassed by two key technical statements --- the Contraction Lemma and the Downgrade Lemma --- which we explain below.

\paragraph*{Contraction Lemma.} We explain the Contraction Lemma for the construction of $H\coloneqq H_0$ from $G=G_0$; the construction for $i>0$ is the same. Recall that eventually we are interested in monitoring whether the given sentence $\varphi$ is satisfied in $G$. For this, it is sufficient to monitor the type $\tp^q(G)$, where $q$ is the quantifier rank of $\varphi$. Consider the following construction:
\begin{itemize}[nosep]
 \item Pick some large $p\in \N$.
 \item Consider the fern decomposition $\Ff$ of $G$ and let ${\cal K}\coloneqq \{\bnd S\colon S\in \Ff\}$. For every $D\in {\cal K}$, let $R_D$ be the join of all the ferns with boundary $D$, and with colors stripped from the vertices of $D$. Note that $R_D$ is a boundaried graph with boundary~$D$.
 \item For every $D\in {\cal K}$ with $|D|=2$, contract $R_D$ to a single edge with color $\tp^p(R_D)$ connecting the two vertices of $D$.
 \item For every $D\in {\cal K}$ with $|D|=1$, contract $R_D$ onto the single vertex $d$ of $D$, and make $d$ of color $\tp^p(R_D)$.
 \item Remove $R_\emptyset$, if present, and remember $\tp^p(R_\emptyset)$ through flags\footnote{We assume that a colored graph can be supplied with a bounded number of boolean flags, which thus can store a bounded amount of additional information. In the general setting of relational structures, flags are modeled by nullary predicates (predicates of arity $0$).}.
 \item The obtained colored graph is named $\Contract^p(G)$. Note that $\Contract^p(G)$ is a $\Gamma^p$-colored graph, where $\Gamma^p$ is a palette consisting of all rank-$p$ types of $\Sigma$-colored graphs with a boundary of size at most $2$.
\end{itemize}
Thus, every fern $S$ in $G$ is essentially disposed of, but a finite piece of information (the rank-$p$ type) about $S$ is being remembered in $\Contract^p(G)$ on the boundary of $S$. The intuition is that if $p$ is large enough, these pieces of information are enough to infer the rank-$q$ type of~$G$. This intuition is confirmed by the following Replacement Lemma.

\begin{lemma}[Replacement Lemma, informal statement]\label{lem:replacement-inf}
 For any given $q\in \N$ and $\Sigma$, there exists $p\in \N$ large enough so that for any $\Sigma$-colored graph~$G$, the type $\tp^p(\Contract^p(G))$ uniquely determines the type $\tp^q(G)$.
\end{lemma}

The proof of the Replacement Lemma uses Ehrenfeucht-Fra\"isse games. It is conceptually rather standard, but technically quite involved. We note that the obtained constant $p$ is essentially the number of rank-$q$ types of $\Sigma$-colored graphs, which is approximately a tower of exponentials of height $q$ applied to $|\Sigma|$. Since Replacement Lemma is used $k$ times in the construction, this incurs a huge explosion in the parameter dependence in our data structure.

Replacement Lemma shows that in order to monitor the type $\tp^q(G)$ in the dynamic setting, it suffices to maintain the graph $H\coloneqq \Contract^p(G)$ and the type $\tp^p(H)$. Maintaining $H$ dynamically is the responsibility of the Contraction Lemma.

\begin{lemma}[Contraction Lemma, informal statement]\label{lem:contraction-inf}
 For a given $p\in \N$ and palette~$\Sigma$, there is a dynamic data structure that for a dynamic graph $G$, maintains the graph $\Contract^p(G)$ under updates in $G$. The worst-case update time is $\Oh[p,\Sigma]{\log n}$.
\end{lemma}

The proof of Lemma~\ref{lem:contraction-inf} follows closely the reasoning of Alman et al.~\cite{AlmanMW20}. That is, in the same way as in~\cite{AlmanMW20}, every update in $G$ incurs a constant number of changes in the fern decomposition of $G$, expressed as splitting or merging of individual ferns. Instead of relying on link-cut trees as in~\cite{AlmanMW20}, the ferns are stored using top trees. This is because we enrich the top trees data structure with the information about rank-$p$ types of clusters, as in Theorem~\ref{thm:forests}, so that for each fern $S$ we know its rank-$p$ type. This type is needed to determine the color of the feature (edge/vertex/flag) in $H=\Contract^p(G)$ to which $S$ contributes. 

Executing the plan sketched above requires an extreme care about details. Note for instance that in the construction of $\Contract^p(G)$, when defining $R_D$ we explicitly stripped colors from the boundary vertices. This is for a reason similar to that discussed alongside Theorem~\ref{thm:forests}: including the information on the colors of $D$ in $\tp^p(R_D)$ would mean that a single update to the color of a vertex $d$ would affect the types of all subgraphs $R_D$ with $d\in D$, and there is potentially an unbounded number of such subgraphs. Further, we remark that Alman et al.~\cite{AlmanMW20} relied on an understanding of the fern decomposition through a sequence of dissolutions, which makes some arguments inconvenient for generalization to our setting. We need a firmer grasp on the notion of fern decomposition, hence we introduce a robust graph-theoretic description that is static --- it does not rely on an~iterative dissolution procedure. This robustness helps us greatly in maintaining ferns and their types in the dynamic setting.

Another noteworthy technical detail is that the operator $\Contract^p(\cdot)$, as defined above, does not create parallel edges or loops, and thus we stay within the realm of colored simple graphs (or, in the general setting, of classic relational structures over binary signatures). Unfortunately, this simplification cannot be applied throughout the whole proof, as in Lemma~\ref{lem:high} we need to count the degrees with respect to the multigraph $\Contract(G)$ as defined in Alman et al.~\cite{AlmanMW20}. For this reason, in the full proof we keep trace of two objects at the same time: a relational structure $\mathbb A$ that we are interested in, and a multigraph $H$ which is a supergraph of the Gaifman graph of $\mathbb A$ and that represents the structure of earlier contractions.

\paragraph*{Downgrade Lemma.} Finally, we are left with the Downgrade Lemma, which reduces the graph by removing a bounded number of vertices. Formally, we have a $\Gamma$-colored graph $H$ and a set $B$ of $\Oh{k}$ vertices, and we would like to construct a $\Sigma'$-colored graph $G'=\Downgrade(H,B)$ by removing the vertices of $B$ and remembering information about them on the remaining vertices of $H$. This construction is executed as follows:
\begin{itemize}[nosep]
 \item Enumerate the vertices of $B$ as $b_1,\ldots,b_\ell$, where $\ell=|B|$.
 \item Construct $G'$ by removing vertices of $B$.
 \item For every color $c\in \Gamma$ and $i\in \{1,\ldots,\ell\}$, add to $G'$ a flag signifying whether $b_i$ has color $c$ in $G$.
 \item For every pair $i,j\in \{1,\ldots,\ell\}$, $i<j$, and every color $c\in \Gamma$ add to $G'$ a flag signifying whether $b_i$ and $b_j$ are connected in $G$ by an edge of color $c$.
 \item For every vertex $u\in V(G)\setminus B$, every $i\in \{1,\ldots,\ell\}$, and every color $c\in \Gamma$, refine the color of $u$ in $G'$ by adding the information on whether $u$ and $b_i$ were connected in $G$ by an edge of color $c$. 
 \item The obtained graph is the graph $G'$. Note that $G'$ is $\Sigma'$-colored, where $\Sigma'=\Gamma\times 2^{[\ell]\times \Gamma}$.
\end{itemize}
Thus, the information about vertices of $B$ and edges incident to  $B$ is being stored in flags and colors on vertices of $V(G)\setminus B$. We have the following analogue of the Replacement Lemma.

\begin{lemma}\label{lem:downgrade-logic}
 For any given $p\in \N$, there exists $q' \in \N$ large enough so that for any $\Gamma$-colored graph~$H$ and a~subset~$B$ of $\Oh{k}$~vertices, the type $\tp^{q'}(\Downgrade(H,B))$ uniquely determines $\tp^p(H)$.
\end{lemma}

The proof of Lemma~\ref{lem:downgrade-logic} is actually very simple and boils down to a syntactic modification of formulas. From Lemma~\ref{lem:downgrade-logic} it follows that to maintain the type $\tp^p(H)$, it suffices to maintain a bounded-size set $B$ satisfying the conclusion of Lemma~\ref{lem:high}, the graph $G'=\Downgrade(H,B)$, and its type $\tp^{q'}(G')$. This is the responsibility of the Downgrade Lemma.

\begin{lemma}[Downgrade Lemma, informal statement]
For a given $p\in \N$ and palette $\Gamma$, there is a dynamic data structure that for a dynamic graph $H$ of feedback vertex number at most $k$ and with minimum degree $3$, maintains a set of vertices $B\subseteq V(H)$ with $|B|\leq 12k$ and satisfying the conclusion of Lemma~\ref{lem:high}, and the graph $\Downgrade(H,B)$. The worst-case update time is $\Oh[p,\Gamma,k]{\log n}$.
\end{lemma}

The proof of the Downgrade Lemma is essentially the same as that given for the corresponding step in Alman et al.~\cite{AlmanMW20}. We recompute $B$ from scratch every $\Theta(m/k)$ updates, because the argument of Alman et al. shows that $B$ remains valid for this long. Recomputing $B$ implies recomputing $\Downgrade(H,B)$ in $\Oh[p,\Gamma,k]{m}$ time, so the worst-case complexity is $\Oh[p,\Gamma,k]{1}$ (there are additional logarithmic factors from auxiliary data structures).

\paragraph*{Endgame.} We now have all the pieces to assemble the proof of Theorem~\ref{thm:main-graph}. Let $q_0$ be the quantifier rank of the given sentence $\varphi$ and let $G_0=G$ be the considered dynamic graph. By Replacement Lemma, to monitor $\tp^{q_0}(G_0)$ (from which the satisfaction of $\varphi$ can be inferred), it suffices to monitor $\tp^{p_0}(H_0)$, where $H_0\coloneqq \Contract^{p_0}(G_0)$ and $p_0$ is as provided by the Replacement Lemma. By Contraction Lemma, we can efficiently maintain $H_0$ under updates of $G_0$. By Lemma~\ref{lem:downgrade-logic}, to monitor $\tp^{p_0}(H_0)$ it suffices to monitor $\tp^{q_1}(G_1)$, where $G_1\coloneqq \Downgrade(H_0,B_0)$, and $B_0$ is a set that satisfies the conclusion of Lemma~\ref{lem:high}. By Downgrade Lemma, we can efficiently maintain such a set $B_0$ and the graph $G_1$. We proceed further in this way, alternating the usage of the Contraction Lemma and the Downgrade Lemma. Observe that each application of Downgrade Lemma strictly decrements the feedback vertex number, so after $k$ steps we end up with an empty graph $H_k$. The type of this graph can be directly computed from its flags, and this type can be translated back to infer $\tp^q(G)$ by using Replacement Lemma and Lemma~\ref{lem:downgrade-logic} alternately.

\section{Preliminaries}%
\label{sec:preliminaries}

For a nonnegative integer $p$, we write $[p]\coloneqq \{1,\ldots,p\}$. If $\bar c$ is a tuple of parameters, then the $\Oh[\bar c]{\cdot}$ notation hides multiplicative factors that are bounded by a function of $\bar c$. In this paper it will always be the case that this function is computable.

In this work, we assume the standard word RAM model in which we operate on machine words of length $\Oh{\log n}$.
In particular, one can perform arbitrary arithmetic operations on words and pointers in constant time.
All identifiers (elements of the universe of relational structures, vertices of graphs and multigraphs, etc.) are assumed to fit into a~single machine word, allowing us to operate on them in constant time.

\paragraph*{Multigraphs.}
In this work, we consider undirected multigraphs.
A~\emph{multigraph} $G = (V, E)$ is a~graph that is allowed to contain multiple edges connecting the same pair of vertices, as well as arbitrarily many self-loops (edges connecting a~vertex with itself).
For a~graph $G$, we denote by $V(G)$ the set of vertices of $G$, and by $E(G)$ the set of edges.
We define the \emph{size} of the multigraph as $|G| \coloneqq |V(G)| + |E(G)|$.
The \emph{degree} of a~vertex $v$ is the number of different edges of $G$ incident to $v$, where self-loops on $v$ count twice.

A~subset $S \subseteq V(G)$ of vertices of $G$ is a~\emph{feedback vertex set} if $G - S$ is acyclic.
Here, we naturally assume that self-loops are cycles consisting of one vertex and one edge, and two different edges connecting the same pair of vertices form a~cycle consisting of two vertices and two edges.
Then, the \emph{feedback vertex number} of $G$, denoted $\fvs{G}$, is the minimum size of a~feedback vertex set in $G$.
Note that $\fvs{G} = 0$ if and only if $G$ is a~forest.

\paragraph*{Dynamic sets and dictionaries.}
In our algorithms, we will heavily rely on two standard data structures: \emph{dynamic sets} and \emph{dynamic dictionaries}.

A~dynamic set is a~fully dynamic data structure maintaining a~finite subset $S$ of some linearly ordered universe $(\Omega, \leq)$.
We can add or remove elements in $S$ dynamically, as well as query the existence of a~key in $S$, check the size of $S$, or pick any (say, the smallest) element in $S$.
Provided $\leq$ can be evaluated on any pair of keys in $\Omega$ in worst-case constant time, each of these operations can be performed in worst-case $\Oh{\log |S|}$ time using the standard implementations of balanced binary search trees, such as AVL trees or red-black trees.

More generally, a~dynamic dictionary is a~data structure maintaining a~finite set $M$ of key-value pairs $(k, v_k)$, where all keys are pairwise different and come from $(\Omega, \leq)$.
Again, one can add or remove key-value pairs in $M$, replace the mapping of a~key to a~different value, as well as query the value assigned to some key $k$ in $M$.
Given that $\leq$ can be evaluated in constant time and that the key-value pairs can be manipulated in memory in constant time, each of these operations can be implemented in worst-case $\Oh{\log |M|}$ time by a~standard extension of a~dynamic set.

\subsection{Relational structures and logic}

\label{ssec:preliminaries_logic}

\newcommand{\arity}{\mathrm{ar}}
\newcommand{\Af}{\mathbb{A}}
\newcommand{\Bf}{\mathbb{B}}
\newcommand{\Xf}{\mathbb{X}}
\newcommand{\Gf}{\mathbb{G}}
\newcommand{\Hf}{\mathbb{H}}
\newcommand{\tup}[1]{\bar{#1}}
\newcommand{\Infer}{\mathsf{Infer}}
\newcommand{\Smash}{\mathsf{Smash}}
\newcommand{\Sentences}{\mathsf{Sentences}}
\newcommand{\head}{\mathrm{head}}
\newcommand{\tail}{\mathrm{tail}}
\newcommand{\glue}{\mathsf{glue}}

\newcommand{\Xx}{{\cal X}}
\newcommand{\Yy}{{\cal Y}}

\paragraph*{Relational structures.} For convenience of notation, we shall work over relational structures over signatures of arity at most $2$. A {\em{binary signature}} $\Sigma$ is a set of predicates, where each predicate $R\in \Sigma$ has a prescribed arity $\arity(R)\in \{0,1,2\}$. A {\em{$\Sigma$-structure}} $\Af$ consists of a universe $V$ and, for every predicate $R\in \Sigma$, its {\em{interpretation}} $R^{\Af}\subseteq V^{\arity(R)}$. For a tuple $\tup a\in V^k$ and predicate $R\in \Sigma$ of arity $k$, we say that $R(\tup a)$ {\em{holds in $\Af$}} if $\tup a\in R^{\Af}$. Note that if $R$ is a nullary predicate, i.e. $\arity(R)=0$, then $|V^{\arity(R)}|=1$, hence $R^{\Af}$ is {\em{de facto}} a boolean flag expressing whether $R$ holds in $\Af$ or not. 

The universe of a structure $\Af$ will be denoted by $V(\Af)$, while the elements of this universe will be called {\em{vertices}}. Ordered pairs of vertices are called {\em{arcs}}, where the two components of a pair are called the {\em{tail}} and the {\em{head}}, respectively. This is in line with the graph-theoretic interpretation of structures over binary signatures as of vertex- and arc-colored directed graphs (supplied by boolean flags, aka nullary predicates). The {\em{Gaifman graph}} of $\Af$, denoted $G(\Af)$, is the graph on vertex set $V(\Af)$ where two distinct vertices are adjacent if and only if they together satisfy some predicate in $\Af$. Note that even if this pair of vertices satisfies multiple predicates, the edge is added only once to $G(\Af)$, which makes $G(\Af)$ always a simple and undirected graph.

All signatures and all structures used in this paper will be finite. We also assume that the universes of all the considered structures are subsets of $\N$, which we denote by $\Omega\coloneqq \N$ in this context for clarity.

\paragraph*{Augmented structures.} We say that a~relational structure $\Af$ is \emph{guarded} by an~undirected multigraph $H$ if $V(\Af) = V(H)$, and the Gaifman graph $G(\Af)$ of $\Af$ is a~subgraph of $H$; that is, if any two different elements $u, v \in V(\Af)$ are bound by a~relation in $\Af$, then $uv$ must be an~edge of $H$.
Then, an \emph{augmented structure} is a~pair $(\Af, H)$ consisting of a~structure $\Af$ and a~multigraph $H$ guarding $\Af$.

\paragraph*{Boundaried structures.} A {\em{boundaried structure}} is a structure $\Af$ supplied with a subset $\bnd \Af$ of the universe $V(\Af)$, called the {\em{boundary}} of $\Af$. We consider three natural operations on boundaried structures.

For each $d\in \Omega$ there is an operation $\forget_d(\cdot)$ that takes a boundaried structure $\Af$ with $d\in \bnd \Af$ and returns the structure $\forget_d(\Af)$ obtained from $\Af$ by removing $d$ from the boundary. That is, the structure itself remains intact, but $\bnd \forget_d(\Af)=\bnd \Af\setminus \{d\}$.
Note that this operation is applicable to $\Af$ only if $d\in \bnd \Af$. We will use the following shorthand: for a finite $D\subseteq \Omega$, $\forget_D$ is the composition of $\forget_d$ over all $d\in D$; note that the order does not matter.

Further, there is an operation $\oplus$, called {\em{join}}, which works as follows. Given two boundaried structures $\Af$ and $\Bf$ over the same signature $\Sigma$ such that $V(\Af)\cap V(\Bf)=\bnd \Af\cap \bnd \Bf$, their join $\Af\oplus \Bf$ is defined as the boundaried $\Sigma$-structure where:
\begin{itemize}
 \item $V(\Af\oplus \Bf)=V(\Af)\cup V(\Bf)$;
 \item $\bnd (\Af\oplus \Bf)=\bnd \Af\cup \bnd \Bf$; and
 \item $R^{\Af\oplus \Bf}=R^{\Af}\cup R^{\Bf}$ for each $R\in \Sigma$.
\end{itemize}
Note that the join operation is applicable only if $V(\Af)$ and $V(\Bf)$ intersect only at subsets of their boundaries. However, we allow $\Af$ and $\Bf$ to share vertices in their boundaries, which are effectively ``glued'' during performing the join; this is the key aspect of this definition.

Finally, for all finite $D,D'\subseteq \Omega$ and a surjection $\xi\colon D\to D'$ there is an operation $\glue_\xi(\cdot)$ that takes a boundaried structure $\Af$ with boundary $D$ and such that $D'\cap (V(\Af)\setminus D)=\emptyset$, and returns the structure $\glue_\xi(\Af)$ that is obtained from $\Af$ as follows:
\begin{itemize}
 \item The universe of $\glue_\xi(\Af)$ is $(V(\Af)\setminus D)\cup D'$.
 \item Every relation in $\glue_\xi(\Af)$ is obtained from the corresponding relation in $\Af$ by replacing every occurrence of any $d\in D$ with $\xi(d)$. 
\end{itemize}
Note that we do not require $\xi$ to be injective, in particular it can ``glue'' two different elements $d_1,d_2\in D$ into a single element $\xi(d_1)=\xi(d_2)\in D'$. This will be the primary usage of the operation, hence the name.

\paragraph*{Boundaried multigraphs.} Analogously, a~\emph{boundaried multigraph} is a~multigraph $H$, together with a~subset $\bnd H$ of $V(H)$, called the \emph{boundary} of $H$.
The operations defined for boundaried structures: $\forget_d(\cdot)$, $\oplus$, and $\glue_\xi$ translate naturally to boundaried multigraphs.

\paragraph*{Logic.} Let $\Sigma$ be a binary signature. The \emph{Monadic Second-Order logic over $\Sigma$ with modular counting predicates} ($\msotwo$ over $\Sigma$) is a logic where there are variables for individual vertices, individual arcs, sets of vertices, and sets of arcs. Variables of the first two kinds are called {\em{individual}} and of the latter two kinds are called {\em{monadic}}. Atomic formulas are the following:
\begin{itemize}
 \item Equality for every kind of variables.
 \item Membership checks of the form $x\in X$, where $x$ is an individual variable and $X$ is a monadic variable.
 \item Checks of the form $\head(x,f)$ and $\tail(x,f)$, where $x$ is an individual vertex variable and $f$ is an individual arc variables.
 \item For each $R\in \Sigma$, relation checks for $R$ of the form depending on the arity of $R$: $R$ if $\arity(R)=0$, $R(x)$ where $x$ is an individual vertex variable if $\arity(R)=1$, and $R(f)$ where $f$ is an individual arc variable if $\arity(R)=2$. 
 \item Modular counting checks of the form $|X|\equiv a\bmod p$, where $X$ is a monadic variable and $a,p$ are integers with $p\neq 0$.
\end{itemize}
The semantics of the above are standard. These atomic formulas can be combined into larger formulas using standard boolean connectives, negation, and quantification  over individual vertices, individual arcs, subsets of vertices, and subsets of arcs, each introducing a new variable of the corresponding kind. However, we require that quantification over subsets of arcs is guarded by the union of binary predicates from $\Sigma$. Precisely, if by $\bigvee \Sigma^{(2)}$ we denote the union of all binary predicates in $\Sigma$, then 
\begin{itemize}
 \item quantification over individual arcs takes the form $\exists_{f\in \bigvee \Sigma^{(2)}}$ or $\forall_{f\in \bigvee \Sigma^{(2)}}$; and
 \item quantification over arc subsets takes the form $\exists_{F\subseteq \bigvee \Sigma^{(2)}}$ or $\forall_{F\subseteq \bigvee \Sigma^{(2)}}$.
\end{itemize}
Note that thus, every arc that is quantified or belongs to a quantified set of arcs is present (in its undirected form) in the Gaifman graph of the structure. Again, the semantics of quantification is standard. 

As usual, formulas with no free variables will be called {\em{sentences}}. Satisfaction of a sentence $\varphi$ in a $\Sigma$-structure $\Af$ is defined as usual and denoted $\Af\models \varphi$. This notation is extended to satisfaction of formulas with provided evaluation of free variables in the usual manner.

For a finite set $D\subseteq \Omega$, we define $\msotwo$ formulas over signature $\Sigma$ and {\em{boundary}} $D$ as $\msotwo$ formulas over $\Sigma$ that can additionally use the elements of $D$ as constants, that is, every element $d\in D$ can be freely used in atomic formulas. Such formulas will always be considered over boundaried structures where $D$ is the boundary, hence in particular each $d\in D$ will be always present in the structure.
The set of all $\msotwo$ sentences over signature $\Sigma$ is called $\msotwo[\Sigma]$, and $\msotwo[\Sigma,D]$ if a boundary $D\subseteq \Omega$ is also taken into account.

For a formula $\varphi$, the {\em{rank}} of $\varphi$ is equal to the maximum of the following two quantities:
\begin{itemize}
 \item the maximum nesting depth of quantifiers in $\varphi$; and
 \item the maximum among all the moduli in all the modular counting checks in $\varphi$.
\end{itemize}


\paragraph*{Types.} The following lemma is standard, see e.g.~\cite[Exercise~6.11]{Immerman99}.

\begin{lemma}\label{lem:formulas-finite}
 For a given binary signature $\Sigma$, $D\subseteq \Omega$, and $q\in \N$, there is a finite set $\Sentences^{q,\Sigma}(D)\subseteq \msotwo[\Sigma,D]$ consisting of sentences of rank at most $q$ such the following holds: for every sentence $\varphi\in \msotwo[\Sigma,D]$ of rank at most $q$ there exists $\varphi'\in \Sentences^{q,\Sigma}(D)$ such that
 $$\Af\models \varphi\quad\Leftrightarrow\quad \Af\models \varphi'\qquad\textrm{for every boundaried }\Sigma\textrm{-structure with }\bnd \Af=D.$$
 Moreover, $\Sentences^{q,\Sigma}(D)$ can be computed for given $\Sigma$, $D$, and $q$. Also, given a sentence $\varphi\in \msotwo[\Sigma,D]$ of rank at most $q$, the formula $\varphi'\in \Sentences^{q,\Sigma}(D)$ satisfying the above can be also computed.
\end{lemma}

We will henceforth use the sets $\Sentences^{q,\Sigma}(D)$ provided by Lemma~\ref{lem:formulas-finite} in the notation. As every sentence $\varphi\in \msotwo[\Sigma,D]$ of rank $q$ can be algorithmically translated to an equivalent sentence belonging to $\Sentences^{q,\Sigma}(D)$ in the sequel we may implicitly assume that all considered sentences belong to the corresponding sets $\Sentences^{q,\Sigma}(D)$.

We also define 
$$\Types^{q,\Sigma}(D)\coloneqq 2^{\Sentences^{q,\Sigma}(D)}$$
as the powerset of $\Sentences^{q,\Sigma}(D)$. The next definition is critical.

\begin{definition}
 Let $\Sigma$ be a binary signature, $\Af$ be a $\Sigma$-structure, $D \subseteq \Omega$ and $q\in \N$. Then the {\em{type}} of rank $q$ of $\Af$ is defined as the set of all sentences from $\Sentences^{q,\Sigma}(D)$ satisfied in $\Af$:
 $$\tp^q(\Af)\coloneqq \{\varphi\in \Sentences^{q,\Sigma}(D)~|~\Af\models \varphi\}\in \Types^{q,\Sigma}(D).$$
\end{definition}

The following lemma describes the compositionality of types with respect to the operations on boundaried structures. The proof is a standard application of Ehrenfeucht-Fra\"isse games and is omitted; see e.g.~\cite{GroheK09,Makowsky04}.

\begin{lemma}\label{lem:compositionality}
 Fix a binary signature $\Sigma$ and $q\in \N$.
 \begin{itemize}
  \item For all finite $D\subseteq \Omega$ and $d\in D$, there exists a computable function $\forget^{q,\Sigma}_{d,D}\colon \Types^{q,\Sigma}(D)\to \Types^{q,\Sigma}(D\setminus \{d\})$ such that
  $$\forget^{q,\Sigma}_{d,D}(\tp^q(\Af))=\tp^q(\forget_d(\Af))$$
  for every boundaried $\Sigma$-structure $\Af$ with $\bnd \Af=D$.
  \item For all finite $C,D\subseteq \Omega$, there exists a computable function $\oplus^{q,\Sigma}_{C,D}\colon \Types^{q,\Sigma}(C)\times \Types^{q,\Sigma}(D)\to \Types^{q,\Sigma}(C\cup D)$ such that
  $$\tp^q(\Af)\oplus^{q,\Sigma}_{C,D} \tp^q(\Bf)=\tp^q(\Af\oplus \Bf)$$
  for all boundaried $\Sigma$-structures $\Af$ and $\Bf$ with $\bnd \Af=C$ and $\bnd \Bf=D$.
  \item For all finite $D,D'\subseteq \Omega$ and a surjective function $\xi\colon D\to D'$, there exists a computable function $\glue_\xi^{q,\Sigma}\colon \Types^{q,\Sigma}(D)\to \Types^{q,\Sigma}(D')$ such that
  $$\glue^{q,\Sigma}_{\xi}(\tp^q(\Af))=\tp^q(\glue_\xi(\Af))$$
  for every boundaried $\Sigma$-structure $\Af$ with $\bnd \Af=D$.
 \end{itemize}
\end{lemma}

We note that since the join operation $\oplus$ on boundaried structures is associative and commutative, the join operation $\oplus^{q,\Sigma}_{C,D}$ on types is also associative and commutative whenever $C=D$. 

We will also use the idempotence of the join operation on types, which is encapsulated in the following lemma. Again, the proof is a standard application of Ehrenfeucht-Fra\"isse games and is omitted.

\begin{lemma}
  \label{lem:idempotence_of_types}
 Let $\Sigma$ be a binary signature, $D\subseteq \Omega$, and $q\in \N$. Then there exists $m\in \N$, computable from $\Sigma$, $|D|$, and $q$, such that the following holds: for all $a,b\in \N$ such that $a,b \geq m$ and $a\equiv b\bmod m$, and every type $\alpha\in \Types^q(D)$, we have
 $$\underbrace{\alpha \oplus^{q,\Sigma}_{D,D} \alpha  \oplus^{q,\Sigma}_{D,D} \cdots \oplus^{q,\Sigma}_{D,D} \alpha}_{a\textrm{ times}}=\underbrace{\alpha \oplus^{q,\Sigma}_{D,D} \alpha  \oplus^{q,\Sigma}_{D,D} \cdots \oplus^{q,\Sigma}_{D,D} \alpha}_{b\textrm{ times}}.$$
\end{lemma}

\paragraph*{Canonization of types.} Note that formally the sets $\Sentences^{q,\Sigma}(D)$ are different for different $D\subseteq \Omega$, but whenever $D,D'\subseteq \Omega$ have the same cardinality and $\pi\colon D\to D'$ is a bijection, then $\pi$ also induces also a bijection from $\Sentences^{q,\Sigma}(D)$ to $\Sentences^{q,\Sigma}(D')$ that replaces every occurrence of any $d\in D$ with $\pi(d)$. Recalling that $\Omega=\N$, for every $D\subseteq \Omega$ we let $\iota_D$ be the unique order-preserving bijection from $D$ to $[|D|]$. Thus, $\iota_D$ induces a bijection from $\Sentences^{q,\Sigma}(D)$ to $\Sentences^{q,\Sigma}([|D|])$, which we will also denote by $\iota_D$. The reader may think that if $\varphi\in \Sentences^{q,\Sigma}(D)$, then $\iota_D(\varphi)$ is a ``canonical variant'' of $\varphi$, where the elements of $D$ are reindexed with numbers in $\{1,\ldots,|D|\}$ in an order-preserving way. Note that thus, whenever $|D|=|D'|$, $\iota_{D'}^{-1} \circ \iota_D$ is a bijection from $\Sentences^{q,\Sigma}(D)$ to $\Sentences^{q,\Sigma}(D')$.

As $\iota_D$ acts on the elements of $\Sentences^{q,\Sigma}(D)$, it also naturally acts on their subsets. Hence $\iota_D$ induces a bijection from $\Types^{q,\Sigma}(D)$ to $\Types^{q,\Sigma}(D')$ in the expected way, and we will denote this bijection also as $\iota_D$. Again, for $\alpha\in \Types^{q,\Sigma}(D)$, $\iota_D(\alpha)$ can be regarded as the ``canonical variant'' of $\alpha$.

\paragraph*{Ensembles.} In our reasonings we will often work with decompositions of large structures into smaller, simpler substructures. Such decompositions will be captured by the notion of an {\em{ensemble}}, which we introduce now.

For a binary signature $\Sigma$, a {\em{$\Sigma$-ensemble}} is a finite set $\Xx$ of boundaried $\Sigma$-structures, each with a boundary of size at most $2$. Moreover, we require that the elements of an ensemble $\Xx$ are pairwise joinable, that is, for all $\Gf,\Hf\in \Xx$ we have $V(\Gf)\cap V(\Hf)=\bnd \Gf\cap \bnd \Hf$; equivalently, the sets $V(\Gf)\setminus \bnd \Gf$ for $\Gf\in \Xx$ are pairwise disjoint.
The {\em{smash}} of a $\Sigma$-ensemble $\Xx$ is the $\Sigma$-structure
$$\Smash(\Xx)\coloneqq \forget_{\bigcup_{\Gf\in \Xx} \bnd \Gf}\left(\bigoplus_{\Gf\in \Xx} \Gf\right).$$
Intuitively, $\Smash(\Xx)$ is the structure which is decomposed into the ensemble $\Xx$.

\paragraph*{Replacement Lemma.} We now formulate a logical statement, dubbed {\em{Replacement Lemma}}, that will be crucially used in our data structure. Its intuitive meaning is the following: If we partition a structure $\Af$ into several boundaried structures, each with boundary of size $2$, and we replace each of them with a single arc labeled with its type, then the replacement preserves the type of $\Af$. Here, if we want to preserve the type of rank $q$, the labels of arcs should encode types of rank $p$, where $p$ is sufficiently high depending on $q$.

For $p\in \N$, we define a new signature $\Gamma^p \coloneqq (\Gamma^p)^{(0)} \cup (\Gamma^p)^{(1)} \cup (\Gamma^p)^{(2)}$, where:
\begin{equation*}
\begin{split}
(\Gamma^p)^{(0)} & \coloneqq \Types^{p,\Sigma}(\emptyset), \\
(\Gamma^p)^{(1)} & \coloneqq \Types^{p,\Sigma}(\{1\}), \\
(\Gamma^p)^{(2)} & \coloneqq \Types^{p,\Sigma}(\{1, 2\}).
\end{split}
\end{equation*}
It is apparent that $\Gamma^p$ is finite and computable from $p$ and $\Sigma$.
Now, the {\em{rank-$p$ contraction}} of a $\Sigma$-ensemble $\Xx$ is the $\Gamma^p$-structure $\Contract^p(\Xx)$ defined as follows:
\begin{itemize}
 \item The universe of $\Contract^p(\Xx)$ is $D\coloneqq \bigcup_{\Gf\in \Xx} \bnd \Gf$.
 \item For every $i\in \{0,1,2\}$ and $\alpha\in \Types^{p,\Sigma}([i])$, the interpretation of $\alpha$ in $\Contract^p(\Xx)$ consists of all tuples $\tup a\in D^i$ such that:
 \begin{itemize}
 \item $\tup a$ is ordered by $\leq$ and its elements are pairwise different;
 \item there exists at least one $\Gf\in \Xx$ such that $\bnd \Gf$ is equal to the set of entries of $\tup a$; and
 \item the rank-$p$ type of the join of all the $\Gf\in \Xx$ as above is equal to $\iota_{\tup a}^{-1}(\alpha)$.
\end{itemize}
\end{itemize}
The Replacement Lemma then reads as follows.

\begin{lemma}[Replacement Lemma]
\label{lem:replacement_lemma}
 Let $\Sigma$ be a binary signature and $q\in \N$. Then there exists $p\in \N$ and a function $\Infer\colon \Types^{p,\Gamma^p}\to \Types^{q,\Sigma}$ such that for any $\Sigma$-ensemble $\Xx$,
 $$\tp^q(\Smash(\Xx))=\Infer\left(\tp^p(\Contract^p(\Xx))\right).$$
 Moreover, $p$ and $\Infer$ are computable from $\Sigma$ and $q$.
\end{lemma}

The proof is an elaborate application of Ehrenfeucht-Fra\"isse games. We give it in Appendix~\ref{sec:replacement-proof}.

\subsection{Top trees}
\label{ssec:top_trees}
\newcommand{\Strip}{\mathsf{Strip}}
\newcommand{\AddRel}{\mathsf{addRel}}
\newcommand{\DelRel}{\mathsf{delRel}}

We now focus our attention on simple undirected boundaried graphs, which can be seen as binary relational boundaried structures $G = (V, E)$ equipped with one symmetric binary relation $E$ without self-loops, and no unary or nullary relations.
As above, assume that labels of the vertices are integers; that is, $V(G) \subseteq \Omega = \N$.

Recall that a~graph is a~\emph{forest} if it contains no cycles.
The connected components of forests are called \emph{trees}.
Fix a~tree $T$, and designate a~boundary $\bnd T$ consisting of at most two vertices of $T$.
The elements of $\bnd T$ will be called \emph{external boundary vertices}.
A~boundaried connected graph $(C, \bnd C)$ is a~\emph{cluster} of $(T, \bnd T)$ if:
\begin{itemize}
  \item $C$ is a~connected induced subgraph of $T$ with at least one edge;
  \item $|\bnd C| \leq 2$;
  \item all vertices of $V(C)$ incident to any edge outside of $E(C)$ belong to $\bnd C$; and
  \item $\bnd T \cap V(C) \subseteq \bnd C$; i.e., all external boundary vertices in $C$ are exposed in the boundary of $C$.
\end{itemize}
We remark that $(T, \bnd T)$, as long as it contains at least one edge, is also a~cluster.

Now, given a~boundaried tree $(T, \bnd T)$ with $|\bnd T| \leq 2$, define a~\emph{top tree}~\cite{AlstrupHLT05} over $(T, \bnd T)$ as a~rooted binary tree $\Delta_T$ with a~mapping $\eta$ from the nodes of $\Delta_T$ to clusters of $(T, \bnd T)$, such that:
\begin{itemize}
  \item $\eta(r) = (T, \bnd T)$ where $r$ is the root of $\Delta_T$;
  \item $\eta$ induces a~bijection between the set of leaves of $\Delta_T$ and the set of all clusters built on single edges of $T$; and
  \item each non-leaf node $x$ has two children $x_1$, $x_2$ such that $|\bnd \eta(x_1) \cap \bnd \eta(x_2)| = 1$ and
  \[ \eta(x) = \forget_S\left(\eta(x_1) \oplus \eta(x_2)\right) \]
  for some set $S \subseteq \bnd \eta(x_1) \cup \bnd \eta(x_2)$ of the elements belonging to the boundary of either of the clusters $\eta(x_1)$, $\eta(x_2)$.
  In other words, the cluster mapped by $x$ in $\Delta_T$ is a~join of the two clusters mapped by the children of $x$, followed by a~removal of some (possibly none) elements from the boundary of the resulting structure.
\end{itemize}

If $T$ consists of only one vertex, then the top tree $\Delta_T$ is deemed empty.
This is a~design choice: each cluster is identified by a nonempty subset of edges of $T$, where the root of $\Delta_T$ contains all edges of $T$, and the leaves of $\Delta_T$ contain a~single edge each.

\begin{figure}[ht]
    \centering
      \includegraphics{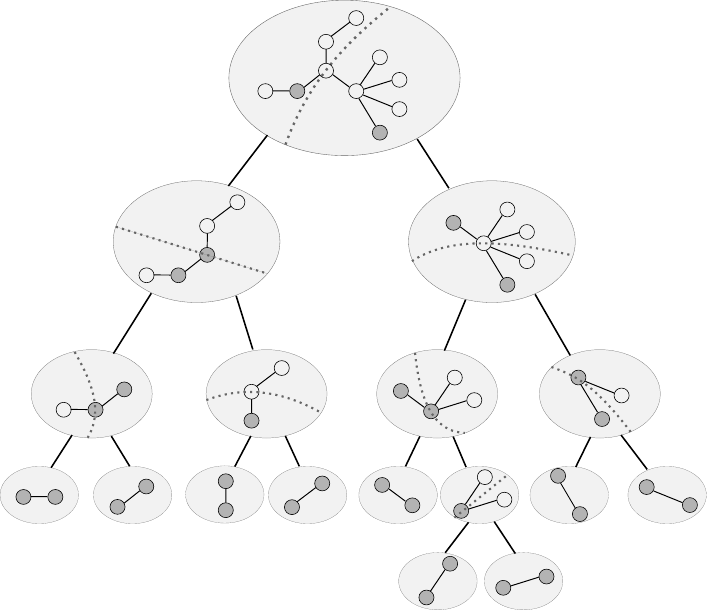}
    \caption{An example top tree $\Delta_T$. Clusters correspond to light gray ovals. Boundary vertices in each cluster are marked dark gray. Note that in this example, $\Delta_T$ has two external boundary vertices. However, it may have fewer (zero or one) such vertices.}
    \Description{An example top tree.}
    \label{fig:top_tree_example}
\end{figure}

Intuitively, a~top tree $\Delta_T$ represents a~recursive decomposition of a~boundaried tree $T$ into smaller and smaller pieces.
In the root of $\Delta_T$, the root cluster $(T, \bnd T)$ is edge-partitioned into two smaller clusters with small boundaries that can be joined along their boundaries to produce $(T, \bnd T)$.
Each of these clusters is again recursively edge-decomposed into simpler pieces, eventually producing clusters consisting of only one edge (Figure~\ref{fig:top_tree_example}).

It turns out that each boundaried tree can be assigned a~shallow top tree:

\begin{theorem}[\cite{AlstrupHLT05}]
  \label{thm:static_top_trees}
  Let $(T, \bnd T)$ be a~boundaried tree with $|\bnd T| \leq 2$, and set $n \coloneqq |V(T)|$.
  Then $(T, \bnd T)$ has a~top tree $\Delta_T$ of depth $\Oh{\log n}$.
\end{theorem}

Given a~forest $F$ of boundaried trees, we define a~forest $\Delta_F$ of top trees of $F$ by assigning each connected component of $F$ a~single top tree.
Here, we assume that one-vertex connected components of $F$ are each given a~separate empty top tree.

While Theorem~\ref{thm:static_top_trees} is fairly straightforward, a~much more interesting result is that a~forest $\Delta_F$ of top trees can be efficiently maintained under the updates of $F$, and that $\Delta_F$ can be used to answer queries about $F$ efficiently.
Namely, consider the following kinds of updates and queries on $F$:

\begin{itemize}
  \item $\mathsf{link}(u, v)$: connects by an~edge two vertices $u$ and $v$, previously in different trees of $F$.
    \item $\mathsf{cut}(u, v)$: disconnects vertices $u$ and $v$ connected by an~edge.
    \item $\mathsf{expose}(S)$: if $S$ is the set of at most two vertices of the same tree $T$ of $F$, assigns the set of external boundary vertices $\bnd T$ to $S$, and returns: a~reference $\Delta_T$ to the root cluster of the top tree of $T$, and the previous boundary $\bnd T$.
    \item $\mathsf{clearBoundary}(\Delta_T)$: given a~reference $\Delta_T$ to the root cluster of the top tree of a~tree $T$, clears the external boundary vertices of $T$, i.e., sets $\bnd T \gets \emptyset$.
    \item $\mathsf{add}(v)$ / $\mathsf{del}(v)$: adds or removes vertex $v$ from $F$.
    If $v$ is removed, it is required to be an~isolated vertex of $F$;
    \item $\mathsf{get}(v)$: given a~vertex $v$ of $F$, returns the reference to the root cluster of the top tree containing $v$.
  \item $\mathsf{jump}(u, v, d)$: if $u$ and $v$ are in the same connected component of $F$, returns the vertex $w$ on the unique simple path between $u$ and $v$ at distance $d$ from $u$, if it exists; and
  \item $\mathsf{meet}(u, v, w)$: if $u$, $v$ and $w$ are in the same connected component of $F$, returns the (unique) vertex $m$ which lies in the intersection of three simple paths in $F$: $uv$, $vw$ and $wu$.
\end{itemize}

It is assumed that the queries: $\mathsf{jump}$, $\mathsf{meet}$, $\mathsf{get}$ do not modify $\Delta_F$.
Moreover, no updates may modify any top trees representing components unrelated to the query.

We note here that the methods $\mathsf{clearBoundary}$ and $\mathsf{get}$ are here mostly for technical reasons related to $\mathsf{expose}$.
The existence of $\mathsf{clearBoundary}$ is a consequence of the fact that $\mathsf{expose}(\emptyset)$ has no reasonable interpretation: given the empty set as the only argument, $\mathsf{expose}$ cannot determine the top tree to be stripped from the boundary vertices.
Then, $\mathsf{get}(v)$ can be implemented solely in terms of two $\mathsf{expose}$ calls (firstly, $\mathsf{expose}(\{v\})$, setting $v$ as an~external boundary vertex of some tree $\Delta_T$, and returning $\Delta_T$ as a~result, and then reverting the old boundary of $\Delta_T$ by another call to $\mathsf{expose}$).
This is, however, unwieldy, and may possibly alter the contents of $\Delta_T$.
Hence, the user of $\Delta_F$ may use $\mathsf{get}$ instead as a~clean, immutable replacement of the calls to $\mathsf{expose}$.

Moreover, assume a~restricted model of computation where the forest of top trees may only be modified by the following operations:

\begin{itemize}
  \item $\mathsf{create}(u, v)$: adds to $\Delta_F$ a~one-vertex top tree corresponding to the one-edge subgraph $C$ of $F$ with $V(C) = \{u, v\}$ and $\bnd C = \{u, v\}$;
  \item $\mathsf{destroy}(T)$: removes from $\Delta_F$ a~one-vertex top tree $T$;
  \item $\mathsf{join}(T_1, T_2, S)$: takes two top trees $T_1$, $T_2$, with roots $r_1$ and $r_2$, respectively, and combines them into a~single top tree $T$ by spawning a~new root node $r$ with children $r_1$ and $r_2$.
  The root $r$ is assigned the cluster $\forget_S\left(\eta(r_1) \oplus \eta(r_2)\right)$.
  \item $\mathsf{split}(T)$: given a~tree $T$ with more than one vertex, splits $T$ into two top trees $T_1$ and $T_2$ by removing the root vertex of $T$.
\end{itemize}

As shown by Alstrup et al., it turns out that even in this restricted model, the updates and queries can be processed efficiently:

\begin{theorem}[\cite{AlstrupHLT05}]
  \label{thm:top_trees}
  There exists a~data structure that, given a~dynamic forest $F$, implements a~dynamic forest $\Delta_F$ of top trees.
  At any point, if $F$ has exactly $n$ vertices, then each tree of $\Delta_F$ has height $\Oh{\log n}$, and each of the queries: $\mathsf{link}$, $\mathsf{cut}$, $\mathsf{expose}$, $\mathsf{clearBoundary}$, $\mathsf{add}$, $\mathsf{del}$, $\mathsf{get}$, $\mathsf{jump}$, and $\mathsf{meet}$ can be executed in $\Oh{\log n}$ worst-case time complexity.

  Additionally, in order to update $\Delta_F$, each query requires at most $\Oh{1}$ calls to $\mathsf{create}$ and $\mathsf{destroy}$, and at most $\Oh{\log n}$ calls to $\mathsf{join}$ and $\mathsf{split}$.
\end{theorem}

We remark here that the most basic form of top trees shown in~\cite{AlstrupHLT05} provides only the implementations of $\mathsf{link}$, $\mathsf{cut}$, and $\mathsf{expose}$.
However, we note that $\mathsf{clearBoundary}$, $\mathsf{add}$, $\mathsf{del}$, and $\mathsf{get}$ are trivial to implement, and $\mathsf{jump}$ and $\mathsf{meet}$ are the extensions of the interface of the data structure presented in the same work~\cite{AlstrupHLT05}.

\paragraph{Auxiliary information.}
In top trees, one can assign auxiliary information to vertices and edges of the underlying forest.
This can be conveniently formalized using relational structures.
Namely, assume that $\Delta_F$ is a~top trees data structure maintaining a~forest of top trees for a~dynamic forest $F$, where $V(F) \subseteq \Omega$.
Consider now an~arbitrary relational structure $\Af$ over a~finite binary signature $\Sigma$ that is guarded by $F$.
$\Af$ is also dynamic: one can add or remove arbitrary tuples from the interpretations of predicates in $\Af$, as long as after each update, $G(\Af)$ is a~subgraph of $F$.
Formally, the interface $\Delta_F$ is extended by the following methods:
\begin{itemize}
  \item $\AddRel(R, \bar{a})$: if for some $i \in \{0, 1, 2\}$, we have that $R \in \Sigma^{(i)}$ and $|\bar{a}| = i$, then adds $\bar{a}$ to the interpretation of $R$ in $\Af$;
  \item $\DelRel(R, \bar{a})$: as above, but removes $\bar{a}$ from the interpretation of $R$ in $\Af$.
\end{itemize}
Under these updates, the set of external boundary vertices in any top tree should not change.

\paragraph{Substructures of $\Af$.} Given a~boundaried graph $G$ which is an~induced subgraph of $F$, we define the substructure $\Af[G]$ of $\Af$ \emph{induced} by $G$ in a~usual way.
That is, we set $V(\Af[G]) = V(G)$ and $\bnd\Af[G] = \bnd G$, and for each $i \in \{0, 1, 2\}$, we define the interpretation of each predicate $R \in \Sigma^{(i)}$ in $\Af[G]$ as $R^{\Af[G]} \coloneqq R^{\Af} \cap V(G)^i$.

However, this definition is not robust enough for our considerations: in our work, we will often need to consider the set of all substructures $\Af[C]$ induced by the clusters $C$ of $\Delta_F$.
In this setup, some information about $\Af$ is shared between multiple induced substructures.
For instance, if an~element $v$ belongs to the boundary of multiple clusters, then each such cluster contains the information about the exact set of unary predicates $R_1$ whose interpretations contain $v$, and by the same token the exact set of binary predicates $R_2$ whose interpretations contain $(v, v)$ (i.e., self-loops on $v$).
Then, a~single update to any such predicate may alter as many as $\Omega(n)$ different substructures of $\Af$ induced by the clusters of $\Delta_F$.
Even worse, the current state of the flag of $\Af$ is stored in all induced substructures, and its change under $\AddRel$ or $\DelRel$ would cause the update of all considered induced substructures.

In order to alleviate this problem, we will consider an~operation stripping boundaried structures from the information on the satisfied flags, unary predicates on boundary vertices, and binary predicates on self-loops on boundary vertices.
Namely, given a~boundaried structure $\Bf$ over $\Sigma$ with boundary $\bnd \Bf$, a~\emph{stripped version} of $\Bf$ is a~boundaried structure $\Bf' \coloneqq \Strip(\Bf)$ over $\Sigma$ defined as follows:
\begin{itemize}
  \item $V(\Bf') = V(\Bf)$ and $\bnd \Bf' = \bnd \Bf$.
  \item $\Bf'$ inherits no flags from $\Bf$; that is, for every $R \in \Sigma^{(0)}$, the interpretation of $R$ in $\Bf'$ is empty.
  \item $\Bf'$ inherits unary relations from the non-boundary elements of $\Bf$; that is, for every $R \in \Sigma^{(1)}$, the interpretation of $R$ in $\Bf'$ is $R^{\Bf'} \coloneqq R^{\Bf} \setminus \bnd \Bf$.
  \item $\Bf'$ inherits all binary relations from $\Bf$, apart from any self-loops on the boundary of $\Bf$; that is, for every $R \in \Sigma^{(2)}$, the interpretation of $R$ in $\Bf'$ is
  $R^{\Bf'} \coloneqq R^{\Bf} \setminus \{(v, v)\,\mid\, v \in \bnd \Bf\}$.
\end{itemize}
Such structures $\Bf'$ will be called \emph{stripped boundaried structures}.
Formally, a~boundaried structure $\Bf'$ is a~stripped boundaried structure if the interpretations of unary and binary relations in $\Bf'$ do not contain tuples of the form $v$ or $(v, v)$ for $v \in \bnd\Bf'$, and $\Bf'$ has no flags.

Observe that we do not need to remove the pairs of the form $(u, v)$ for $u, v \in \bnd \Bf$, $u \neq v$ from the interpretations of binary relations in $\Strip(\B)$: there are at most $\Oh{\log n}$ clusters $C$ of $\Delta_F$ with $u, v \in C$.
Hence an~operation of the form $\AddRel$ or $\DelRel$ involving the pair $(u, v)$ will only modify the information stored in these clusters -- and we will process these modifications efficiently on each such operation.
In fact, removing such pairs $(u, v)$ from the stripped substructures would \emph{complicate} the implementation details of the data structure.
Hence we choose not to remove such pairs from $\Strip(\B)$.

Then, with $\Af$ defined as~above, and $G$ which is an~induced subgraph of $F$, we define an~\emph{almost induced substructure} $\Af\{G\} \coloneqq \Strip(\Af[G])$ as the stripped version of the substructure induced by $G$.
Note that $G$ guards $\Af\{G\}$.

We can now lift operations $\oplus$ and $\forget$ to stripped boundaried structures:
\begin{itemize}
  \item Join $\oplus$ of two stripped boundaried structures is defined in the same way as for ordinary boundaried structures.

  \item Given $S \subseteq \Omega$, the function $\forget_S$ accepts two arguments: a~stripped boundaried structure $\Bf$ with $S \subseteq \bnd \Bf$, and a~mapping $P \in \left(2^S\right)^{\Sigma^{(1)} \cup \Sigma^{(2)}}$, assigning to each unary and binary predicate of $\Sigma$ a~subset of $S$.
  Then $\forget_S(\Bf, P)$ is constructed from $\Bf$ by removing $S$ from its boundary, replacing the evaluation of unary predicates from $\Sigma^{(1)}$ on $S$ with $P|_{\Sigma^{(1)}}$, and replacing the evaluation of binary predicates from $\Sigma^{(2)}$ on self-loops on $S$ with $P|_{\Sigma^{(2)}}$.
  Formally, if $\Bf' = \forget_S(\Bf, P)$, then:
  \[
    \begin{split}
    R^{\Bf'} & \coloneqq R^{\Bf} \cup P(R) \qquad \text{for }R \in \Sigma^{(1)}, \\
    R^{\Bf'} & \coloneqq R^{\Bf} \cup \left\{(v, v)\,\mid\, v \in P(R)\right\} \qquad \text{for }R \in \Sigma^{(2)}.
    \end{split}
  \]
  Intuitively, when elements of $S$ are removed from the boundary of $\Bf$, we need to restore the information about the satisfaction of unary predicates on $S$, and the satisfaction of binary predicates on self-loops on $S$.
  This information is supplied to $\forget_S$ by $P$.
\end{itemize}

Naturally, $\Strip$ commutes with $\oplus$ and $\forget$.
The following is immediate:
\begin{lemma}
  \label{lem:strip_compositional}
  Fix a~binary signature $\Sigma$.
  \begin{itemize}
    \item For any pair of two joinable boundaried structures $\Bf_1$, $\Bf_2$ over $\Sigma$, we have that
      \[ \Strip(\Bf_1 \oplus \Bf_2) = \Strip(\Bf_1) \oplus \Strip(\Bf_2). \]
    \item For any boundaried structure $\Bf$ and any set $S \subseteq \bnd \Bf$, let $P \in \left(2^S\right)^{\Sigma^{(1)} \cup \Sigma^{(2)}}$ be the evaluation of unary predicates from $\Sigma^{(1)}$ on $S$ and binary predicates from $\Sigma^{(2)}$ on self-loops on $S$.
    Then,
    \[ \Strip(\forget_S(\Bf)) = \forget_S\left(\Strip(\Bf), P \right). \]
  \end{itemize}
\end{lemma}

\paragraph{Deducing information on almost induced substructures.}
Finally, additional information can be stored about the substructures almost induced by the clusters present in $\Delta_F$, as long as the information is compositional under joining clusters and removing vertices from the boundary of a~cluster, and this information is isomorphism-invariant.

Formally, for every finite set $D \subseteq \Omega$, consider a~function $\mu_D$ mapping stripped boundaried structures $\Bf$ with $\bnd \Bf = D$ to some universe $I_D$ of possible pieces of information.
Then, $\mu_D$ shall satisfy the following properties:
\begin{itemize}
  \item \emph{Compositionality under joins.} For every finite $D_1, D_2 \subseteq \Omega$, there must exist a~function $\oplus_{D_1, D_2}\,:\, I_{D_1} \times I_{D_2} \to I_{D_1 \cup D_2}$ such that for every pair $\Bf_1$, $\Bf_2$ of stripped boundaried structures with $\bnd \Bf_1 = D_1$, $\bnd \Bf_2 = D_2$, we have:
\[ \mu_{D_1 \cup D_2}\left(\Bf_1 \oplus \Bf_2\right) = \mu_{D_1}(\Bf_1)\, \oplus_{D_1, D_2} \, \mu_{D_2}(\Bf_2). \]
  \item \emph{Compositionality under $\forget$s.} For every finite $D \subseteq \Omega$, and $S \subseteq D$, there must exist a~function $\forget_{S,D}\,:\, I_D \times \left(2^S\right)^{\Sigma^{(1)} \cup \Sigma^{(2)}} \to I_{D \setminus S}$ so that for every stripped boundaried structure $\Bf$ with $\bnd \Bf = D$ and $P \in \left(2^S\right)^{\Sigma^{(1)} \cup \Sigma^{(2)}}$, we have:
\[ \mu_{D \setminus S}\left(\forget_S(\Bf, P)\right) = \forget_{S,D}\left( \mu_D(\Bf), P \right). \]
  \item \emph{Isomorphism invariance.} For every finite $D_1, D_2 \subseteq \Omega$ of equal cardinality, and for every bijection $\phi\,:\,D_1 \to D_2$, there must exist a~function $\iota_{\phi}\,:\, I_{D_1} \to I_{D_2}$ such that for every pair $\Bf_1$, $\Bf_2$ of \emph{isomorphic} boundaried structures with $\bnd \Bf_1 = D_1$, $\bnd \Bf_2 = D_2$, with an~isomorphism $\hat{\phi}\,:\,V(\Bf_1) \to V(\Bf_2)$ extending $\phi$, we have:
  \[ \mu_{D_2}(\Bf_2) = \iota_{\phi}\left( \mu_{D_1}(\Bf_1) \right). \]
\end{itemize}

Then, define the \emph{$\mu$-augmented top trees} data structure as a~variant of top trees in which each node $x$ is augmented with the information $\mu_{\bnd C}(\Af\{C\})$, where $\eta(x) = (C, \bnd C)$.
Thus, given a~reference to a~component $\Delta_T$ of $\Delta_F$ (e.g., obtained from a~call to $\mathsf{expose}$), one can read the information associated with the root cluster of $\Delta_T$.
We stress that this definition of $\mu$-augmented top trees guarantees that the information associated with each cluster $(C, \bnd C)$ of $\Delta_F$ is invariant on the interpretations of unary and binary relations in $\Af$ on the boundary elements of $C$.

We remark that the notation used in the description above is deliberately similar to that defined in Subsection~\ref{ssec:preliminaries_logic}.
In our work, the information $\mu(C)$ stored alongside each cluster $(C, \bnd C)$ in $\mu$-augmented top trees will be precisely the \msotwo{}-type of some rank of the boundaried structure spanned by $G$.
Therefore, thanks to the compositionality of the types of \msotwo{}, the types of forest-like relational structures can be computed by top trees.

We now propose the following lemma, asserting the good asymptotic time complexity of any operation on top trees when the data structure is augmented with the information $\mu$, as long as $\mu_D$, $\oplus_{D_1,D_2}$, and $\forget_{S,D}$ can be computed efficiently:

\begin{lemma}
  \label{lem:top_trees_additional_information}
  Fix $\mu_D$, $\oplus_{D_1, D_2}$, and $\forget_{S,D}$ as above, and consider the $\mu$-augmented top trees data structure $\Delta_F$.
  Suppose the following:
  \begin{itemize}
    \item For each $D \subseteq \Omega$, $|D| \leq 2$, the mapping $\mu_D$ can be computed in $\Oh{1}$ worst-case time from any~stripped boundaried structure with at most $2$ vertices.
    \item For each $D_1, D_2 \subseteq \Omega$, $|D_1|, |D_2| \leq 2$, $|D_1 \cap D_2| = 1$, the function $\oplus_{D_1,D_2}$ can be evaluated on any pair of arguments in worst-case $\Oh{1}$ time.
    \item For each $D \subseteq \Omega$, $|D| \leq 3$, and $S \subseteq D$, the function $\forget_{S,D}$ can be evaluated on any pair of arguments in worst-case $\Oh{1}$ time.
  \end{itemize}
  Then each update and query on $\Delta_F$ can be performed in worst-case $\Oh[\Sigma]{\log n}$ time, where $n = |V(F)|$.
  \begin{proof}
    Firstly, we shall describe how the relational structure $\Af$ is stored in memory.
    Let
    \[ X \coloneqq \{\varepsilon\}\, \cup\, V(F) \,\cup\, \{(v, v)\,\mid\,v \in V(F)\} \,\cup\, \{(u, v)\,\mid\,u \neq v, uv \in E(F)\} \]
    be the set of $0$-, $1$-, and $2$-tuples that may appear in an~interpretation of a~predicate in $\Af$.
    Here, $\varepsilon$ is considered an~empty tuple.
    We remark that $|X| = 1 + 2|V(F)| + 2|E(F)| < 4|V(F)| = 4n$, where $n$ is the number of nodes in $F$.
    Also, there exists a~natural lexicographic ordering $\leq_{\mathrm{lex}}$ of $X$, in which tuples of $X$ can be compared with each other in constant time.

    Then, for each tuple $\bar{a} \in X$, we create a~mutable list $L_{\bar{a}} \subseteq \Sigma^{|\bar{a}|}$ of all $|\bar{a}|$-ary predicates $R$ for which $\bar{a} \in R^{\Af}$.
    Note that for every $\bar{a} \in X$, we have $|L_{\bar{a}}| = \Oh[\Sigma]{1}$.
    Each such list shall be referenced by a~pointer, so that we can update any list at any moment without modifying the pointer referencing the list.
    Naturally, given all lists $L_{\bar{a}}$ for $\bar{a} \in X$, one can uniquely reconstruct $\Af$.

    Moreover, we keep a~dynamic dictionary $M$ such that for every $\bar{a} \in X$, $M(\bar{a})$ stores the pointer to $L_{\bar{a}}$.
    Then, each update on $M$ (insertion or removal from $M$) and each query on $M$ (querying the value of $M$ on a~single key) takes $\Oh{\log |X|} = \Oh{\log n}$ time.
    
    Next, in the forest $\Delta_F$ of top trees, alongside each cluster $(C, \bnd C)$, we store:
    \begin{itemize}
      \item the information $\mu_{\bnd C}(\Af\{C\})$ associated with the cluster; and
      \item for each boundary element $d \in \bnd C$, pointers: to the list $L_{d}$ of unary predicates $R$ for which $d \in R^{\Af}$, and to the list $L_{(d, d)}$ of binary predicates $R$ for which $(d, d) \in R^{\Af}$.
    \end{itemize}
    
    Finally, for our convenience, we keep a~dynamic dictionary $\mathsf{leaves}$, mapping each edge $e \in E(F)$ to the pointer $\mathsf{leaves}(e)$ to the unique leaf node of $\Delta_F$ corresponding to a~one-edge cluster containing $e$ as the only edge.
    Again, $\mathsf{leaves}$ can be updated and queried in $\Oh{\log n}$ time.

    Consider now any update to $F$: $\mathsf{link}$, $\mathsf{cut}$, $\mathsf{expose}$, $\mathsf{clearBoundary}$, $\mathsf{add}$, and $\mathsf{del}$.
    We remark that under each of these updates, the information $\mu$ must only be recomputed for the nodes of $\Delta_F$ created during the update.
    Recall that only two operations on top trees add new nodes to $\Delta_F$: $\mathsf{create}$, spawning a~new one-vertex top tree from a~single-edge subgraph, and $\mathsf{join}$, connecting two rooted top trees mapping to clusters $C_1$, $C_2$ into a~single top tree mapping to a~cluster $C$.

    The $\mathsf{create}$ operation is guaranteed to be called a~constant number of times per query by Theorem~\ref{thm:top_trees}.
    When a~new two-vertex, one-edge cluster $(C, \bnd C)$ is spawned, where $C = \{u, v\}$, we need to compute the information $\mu_{\bnd C}(\Af\{C\})$.
    First, we query the contents of the lists $L_{(u,v)}$ and $L_{(v,u)}$, which requires a~constant number of calls to $M$.
    Given these lists, $\Af\{C\}$ can be reconstructed in constant time.
    Then, constant time is taken to compute the mapping $\mu_{\bnd C}$ on $\Af\{C\}$.
    This information is stored, together with the pointers to the lists $L_d$ and $L_{(d, d)}$ for each $d \in \bnd C$, alongside the constructed cluster.
    Hence, the total time spent in $\mathsf{create}$ is bounded by $\Oh[\Sigma]{\log n}$.
    
    The $\mathsf{join}$ operation is called at most $\Oh{\log n}$ times.
    Recall that in $\mathsf{join}$, the cluster $(C, \bnd C)$ is defined as $C = \forget_S \left(C_1 \oplus C_2\right)$ for two child clusters $C_1$, $C_2$, and some set $S \subseteq \bnd C_1 \cup \bnd C_2$ of elements removed from the boundary of $C$.
    In order to compute $\mu_{\bnd C}(\Af\{C\})$, we need a~few ingredients:
    \begin{itemize}
      \item information $\mu_{\bnd C}(\Af\{C_1\})$ and $\mu_{\bnd C}(\Af\{C_2\})$ about the stripped boundaries structures referenced by children of $C$; and
      \item the mapping $P \in \left(2^S\right)^{\Sigma^{(1)} \cup \Sigma^{(2)}}$ denoting the evaluation of unary predicates from $\Sigma^{(1)}$ on $S$, and of binary predicates from $\Sigma^{(2)}$ on self-loops on $S$.
    \end{itemize}
    
    Note that $\mu_{\bnd C}(\Af\{C_1\})$ and $\mu_{\bnd C}(\Af\{C_2\})$ can be read from the information stored together with the clusters $C_1$ and $C_2$.
    Observe also that we can access lists $L_v$ and $L_{(v, v)}$ for all $v \in S$ in constant time: either $v \in \bnd C_1$, and the pointers to $L_v$ and $L_{(v, v)}$ are stored together with $C_1$, or $v \in \bnd C_2$, and the corresponding pointers are stored together with $C_2$.
    Thus, the sought evaluation $P$ can be constructed from those lists in constant time.
    Now, notice that
    \[ \Af[C] = \forget_S\left(\Af[C_1] \oplus \Af[C_2]\right). \]
    Thus, by Lemma~\ref{lem:strip_compositional}:
    \[ \Af\{C\} = \forget_S\left(\Af\{C_1\} \oplus \Af\{C_2\}, P\right). \]
    Therefore, $\mu_{\bnd C}(\Af\{C\})$ can be computed efficiently from $\mu_{\bnd C_1}(\Af\{C_1\})$ and $\mu_{\bnd C_2}(\Af\{C_2\})$ by exploiting the compositionality of $\mu$ under joins and forgets:
    \begin{equation}
    \label{eq:computation_of_information}
    \begin{split}
    \mu_{\bnd C}(\Af\{C\}) &=
      \mu_{\bnd C}\left[ \forget_S \left( \Af\{C_1\} \oplus \Af\{C_2\}, P\right) \right] = \\
      &= \forget_{S, \bnd C_1 \cup \bnd C_2}\left[\mu_{ \bnd C_1 \cup \bnd C_2}\left(\Af\{C_1\} \oplus \Af\{C_2\}\right), P\right] = \\
      &= \forget_{S, \bnd C_1 \cup \bnd C_2}\left[\mu_{\bnd C_1}\left(\Af\{C_1\}\right)\, \oplus_{\bnd C_1, \bnd C_2}\, \mu_{\bnd C_2}\left(\Af\{C_2\}\right), P\right].
    \end{split}
    \end{equation}
    Note that $|\bnd C_1|, |\bnd C_2| \leq 2$ and $|\bnd C_1 \cap \bnd C_2| = 1$, so $|\bnd C_1 \cup \bnd C_2| \leq 3$.
    Thus, in order to compute the information about $\Af\{C\}$, we need to evaluate $\oplus_{\bnd C_1, \bnd C_2}$ once, followed by one evaluation of $\forget_{S, \bnd C_1 \cup \bnd C_2}$.
    By our assumptions, each of these evaluations take worst-case constant time, and so the computed information $\mu_{\bnd C}(\Af\{C\})$ can be computed in constant time and stored, together with the pointers to the lists $L_v$ and $L_{(v, v)}$ for $v \in \bnd C$, alongside the cluster $(C, \bnd C)$.
    This results in a~worst-case $\Oh[\Sigma]{\log n}$ time bound across all $\mathsf{join}$s per update.
    
    For $\mathsf{get}$, $\mathsf{jump}$ and $\mathsf{meet}$, observe that these are queries on $F$ that do not require any updates to the top trees data structure nor are they related to $\Af$.
    Hence, the implementations of these methods remain unchanged, and so each call to each method concludes in worst-case $\Oh{\log n}$ time.
    
    Finally, we consider $\AddRel(R, \bar{a})$ and $\DelRel(R, \bar{a})$.
    The implementations of these methods depend on the contents of $\bar{a}$:
    \begin{itemize}
      \item If $\bar{a} = \varepsilon$, we only update the dictionary $M$ accordingly.
        Since stripped boundaried structures do not maintain any information on the flags of $\Af$, no information stored in any cluster changes.
        Hence, the entire update can be done in $\Oh{\log n}$ time.

      \item If $\bar{a} = v$ or $\bar{a} = (v, v)$ for $v \in \Omega$, then the interpretation of some predicate $R \in \Sigma^{(1)} \cup \Sigma^{(2)}$ is updated: the element $v \in \Omega$ is either added to or removed from $R^{\Af}$.

      We resolve the update by first calling $\mathsf{expose}(\{v\})$, causing $v$ to become an~external boundary vertex of the unique top tree $\Delta_T$ containing $v$ as a~vertex; let also $\bnd_{\mathrm{old}}$ be the previous set of external vertices of $\Delta_T$.
      After this call, every cluster of $\Delta_T$ containing $v$ as a~vertex necessarily has $v$ in its boundary; hence, no substructure of $\Af$ almost induced by a~cluster of $\Delta_T$ depends on the set of unary predicates satisfied by $v$, or the set of binary predicates satisfied by $(v, v)$.
      Thanks to this fact, we can update the dictionary $M$ according to the query, without any need to update the information stored in the clusters of $\Delta_T$.
      Finally, we revert the set of external boundary vertices of $\Delta_T$ to $\bnd_{\mathrm{old}}$ by another call to $\mathsf{expose}$.
      Naturally, this entire process can be performed in $\Oh[\Sigma]{\log n}$ worst-case time.
      
      \item If $\bar{a} = (u, v)$ with $u \neq v$, then a~pair $(u, v)$ is either added or removed from $R^{\Af}$ for some $R \in \Sigma^{(2)}$.
      We first update the dictionary $M$ accordingly; and let $e = uv \in E(F)$ be the edge of $F$.
      This, however, causes the information stored in some clusters of $\Delta_F$ to become obsolete; namely, the stripped boundaried structures corresponding to the clusters $(C, \bnd C)$ containing $e$ as an~edge change, so the information related to these clusters needs to be refreshed.
      To this end, observe that the set of all such clusters $(C, \bnd C)$ forms a~rooted path from the root of some top tree $\Delta_T$ to the leaf $\mathsf{leaves}(e)$ corresponding to the one-edge cluster containing $e$ as an~edge.
      Hence, the information can be updated by following the tree bottom-up from $\mathsf{leaves}(e)$ all the way to the root of $\Delta_T$, recomputing information $\mu$ about the stripped boundaried structures on the way using (\ref{eq:computation_of_information}).
      As Theorem~\ref{thm:top_trees} asserts that the depth of $\Delta_T$ is logarithmic with respect to $n$, this case is again resolved in worst-case $\Oh[\Sigma]{\log n}$ time.
    \end{itemize}
  Summing up, each update and query: $\mathsf{link}$, $\mathsf{cut}$, $\mathsf{expose}$, $\mathsf{clearBoundary}$, $\mathsf{add}$, $\mathsf{del}$, $\mathsf{get}$, $\mathsf{jump}$, $\mathsf{meet}$, $\AddRel$, and $\DelRel$ can be performed in $\Oh[\Sigma]{\log n}$ time.
  \end{proof}
\end{lemma}

\section{Statement of the main result and proof strategy}%
\label{sec:strategy}

With all the definitions in place, we may state the main result of this work.

\begin{theorem}
  \label{thm:main_theorem}
   Given a sentence $\varphi$ of $\msotwo$ over a~binary relational signature $\Sigma$ and $k\in \N$, one can construct a data structure that maintains whether a given dynamic relational structure $\Af$ over $\Sigma$ satisfies $\varphi$. $\Af$ is initially empty and may be modified by adding or removing elements of the universe, as well as adding or removing tuples from the interpretations of relations in~$\Af$.
   Here, a~vertex $v$ may be removed from $V(\Af)$ only if $v$ participates in no relations of $\Af$.
   
   The data structure is obliged to report a correct answer only if the feedback vertex number of the Gaifman graph $G(\Af)$ of $\Af$ is at most $k$, otherwise it reports {\em{Feedback vertex number too large}}. The amortized update time is $f(\varphi,k)\cdot \log n$, for some computable function $f$.
   
   Moreover, if the feedback vertex number of $G$ is at most $k$ at all times, we can ensure the worst-case update time $f(\varphi, k) \cdot \log n$.
\end{theorem}

Unfortunately, 
the setting of plain relational structures 
comes short in a~couple of combinatorial aspects that will be important:
\begin{itemize}
  \item Our work will contain involved graph-theoretic constructions and proofs, which are cumbersome to analyze in the terminology of relational structures.
  \item In the proof of the efficiency of the proposed data structure, we will rely on the fact that there may exist multiple parallel edges  between a pair of vertices. This feature cannot be modeled easily within the plain setting of relational structures.
\end{itemize}
These issues will be circumvented by assigning $\Af$ a~multigraph $H$ that guards $\Af$. In other words, we shall work with augmented structures $(\Af, H)$.
Then, graph-theoretic properties and constructions will first be stated in terms of $H$, and only later they will be transferred to~$\Af$.

In the language of augmented structures, we propose the following notion of an~efficient data structure dynamically monitoring the satisfaction of $\varphi$:

\begin{definition}
  \label{def:efficient_dynamic_structure}
  For a~class of multigraphs $\Cc$, a~relational signature $\Sigma$, and a~sentence $\varphi \in \msotwo[\Sigma]$, an~\emph{efficient dynamic $(\Cc, \Sigma, \varphi)$-structure} is a~dynamic data structure $\D$ maintaining an~augmented $\Sigma$-structure $(\A, H)$.
  One can perform the following updates on~$(\A, H)$:
  \begin{itemize}
    \item $\mathsf{initialize}(\A, H)$: initializes the data structure with an~augmented $\Sigma$-structure $(\A, H)$ such that $H \in \Cc$.
    \item $\mathsf{addVertex}(v)$: adds an~isolated vertex $v$ to the universe of $\A$ and to the set of vertices of $H$.
    \item $\mathsf{delVertex}(v)$: removes an~isolated vertex $v$ from $\A$ and $H$. It is assumed that no relation in $\A$ and no edge of $H$ contains $v$ as an~element.
    \item $\mathsf{addEdge}(u, v)$: adds an~undirected edge between $u$ and $v$ in $H$.
    \item $\mathsf{delEdge}(u, v)$: removes one of the edges between $u$ and $v$ in $H$.
    \item $\mathsf{addRelation}(R, \bar{a})$: adds a~tuple $\bar{a}$ to the relation $R$ of matching arity in $\A$. Each element of $\bar{a}$ must belong to $V(\A)$ at the time of query.
    \item $\mathsf{delRelation}(R, \bar{a})$: removes $\bar{a}$ from the relation $R$ of matching arity in $\A$.
  \end{itemize}
  $\D$ accepts the updates in constant-sized batches---sequences of operations to be performed one after another.
  The data structure assumes that after each batch of operations, $H \in \Cc$ and $H$ guards $\A$.
  After each batch of operations, $\D$ reports whether $\varphi$ is satisfied in~$\A$.
  The initialization of the data structure is performed in time $\Oh[\Cc, \varphi]{|H| \log |H|}$, while each subsequent update is performed in worst-case time $\Oh[\Cc, \varphi]{\log{|H|}}$, where $|H| = |V(H)| + |E(H)|$.
  
  Additionally, an~efficient dynamic $(C, \Sigma, \varphi)$-structure $\D$ is \emph{weak} if it is only guaranteed that, upon initialization, $\D$ processes correctly the first $\Omega(|H|)$ updates in worst-case time $\Oh[\Cc, \varphi]{\log |H|}$.
\end{definition}

Then, an~analog of Theorem~\ref{thm:main_theorem} for augmented structures reads as follows:

\begin{theorem}
  \label{thm:efficient_augmented_structure}
  For every integer $k \in \N$, let $\Cc_k$ be the class of multigraphs with feedback vertex number at most $k$.
  Then, given $k \in \N$, a~relational signature $\Sigma$, and a~sentence $\varphi \in \msotwo[\Sigma]$, one can construct an~efficient dynamic $(\Cc_k, \Sigma, \varphi)$-structure.
\end{theorem}

In this section, we will present a~proof strategy for Theorem~\ref{thm:efficient_augmented_structure}, as well as offer a~reduction from Theorem~\ref{thm:main_theorem} to Theorem~\ref{thm:efficient_augmented_structure}: that is, given an~efficient dynamic $(\Cc_k, \Sigma, \varphi)$-structure, we will show how the data structure for Theorem~\ref{thm:main_theorem} is produced.
To this end, we should first understand the key differences between Theorem~\ref{thm:main_theorem} and Theorem~\ref{thm:efficient_augmented_structure}.

\begin{itemize}
  \item The definition of an~efficient dynamic structure accepts classes of multigraphs different than $\Cc_k$.
  Indeed, the proof of Theorem~\ref{thm:efficient_augmented_structure} will require us to define classes $\Cc^\star_k$ of graphs constructed from $\Cc_k$ by filtering out all graphs that contain vertices of degree $0$, $1$ or $2$.
  Then, efficient dynamic structures will be presented both for $\Cc_k$ and for $\Cc^\star_k$.
  \item In Theorem~\ref{thm:efficient_augmented_structure}, we assert that after each batch of updates $H$ has low feedback vertex number, which means that the data structure may break down if $\fvs{H}$ becomes too large.
  However, Theorem~\ref{thm:main_theorem} requires us to correctly detect that the invariant is not satisfied, return {\em{Feedback vertex number too large}}, and stand by until the feedback vertex number decreases below the prescribed threshold.
  To this end, we shall use the technique of \emph{postponing invariant-breaking insertions} proposed by Eppstein et al.~\cite{EppsteinGIS96}.
  Unfortunately, data structures exploiting this framework inherently have amortized update time complexities, so we cannot hope for a~worst-case update time bound in the general setting of Theorem~\ref{thm:main_theorem} using this technique.
  \item In Theorem~\ref{thm:efficient_augmented_structure}, the update time is logarithmic with respect to the size of $H$ (i.e., the total number of vertices and edges in $H$), and not in the size of the universe.
  The difference could cause problems as multigraphs with a~bounded number of vertices may potentially contain an~unbounded number of edges.
  However, in the presented reduction, $H$ will actually be the Gaifman graph of $\Af$; thus, $|E(H)|$ will always be bounded in terms of $n$.
  \item The data structure in Theorem~\ref{thm:efficient_augmented_structure} accepts queries in constant-sized batches; this technical design decision will turn out necessary in some parts of the proof of Theorem~\ref{thm:efficient_augmented_structure}.
  However, the reduction from Theorem~\ref{thm:main_theorem} will essentially ignore this difference by never grouping the queries into batches.
\end{itemize}

\paragraph{Proof strategy for Theorem~\ref{thm:efficient_augmented_structure}.}
Recall from the statement of Theorem~\ref{thm:efficient_augmented_structure} the definition of $\Cc_k$ as the class of multigraphs with feedback vertex number at most $k$.
We now define a~restriction of $\Cc_k$ to multigraphs with no vertices of small degree:
\[ \Cc^\star_k \coloneqq \{G \in \Cc_k\,\mid\,\text{each vertex of }G\text{ has degree at least }3\}. \]
Here, the degree of a~vertex $v$ is the number of edges incident to $v$, where each self-loop on $v$ counts as two incidences.

Let us discuss a couple of corner cases in the definition of the multigraph classes:
$\Cc_0$ is the class of all undirected forests, while $\Cc^\star_0$ is the class containing only one graph---the null graph (that is, the graph without any edges or vertices).

The proof will be an~implementation of the following inductive strategy, which was already discussed semi-formally in Section~\ref{sec:overview}.
\begin{itemize}
  \item \emph{(Base case.)} There exists a~simple efficient dynamic $(\Cc^\star_0, \Sigma, \varphi)$-structure, exploiting the fact that such a~dynamic structure is guaranteed to be given a~dynamic augmented structure with empty universe as its input.
  \item \emph{(Contraction step.)} For $k \in \N$, we can construct an~efficient dynamic $(\Cc_k, \Sigma, \varphi)$-structure $\D$ by:
    \begin{itemize}
      \item constructing a~new signature $\Gamma$ and a~new formula $\psi$ from $k$, $\Sigma$, and $\varphi$;
      \item creating an~instance $\D^\star$ of an~efficient dynamic $(\Cc^\star_k, \Gamma, \psi)$-structure;
      \item relaying each batch of queries from $\D$ to $\D^\star$ in a~smart way, so that the correct answer for $\D$ can be deduced from the answers given by $\D^\star$.
    \end{itemize}
  \item \emph{(Downgrade step.)} For $k \in \N$, $k>0$, we can construct a~\emph{weak} efficient dynamic $(\Cc^\star_k, \Sigma, \varphi)$-structure $\D$ by:
    \begin{itemize}
      \item constructing a~new signature $\Gamma$ and a~new formula $\psi$ from $k$, $\Sigma$, and $\varphi$;
      \item creating an~instance $\wt{\D}$ of an~efficient dynamic $(\Cc_{k-1}, \Gamma, \psi)$-structure;
      \item relaying each batch of queries from $\D$ to $\widetilde{\D}$ in a~way allowing us to infer the correct answer for $\D$ from the answers given by $\widetilde{\D}$.
    \end{itemize}
    The produced data structure will be weak: when initialized with an~augmented $\Sigma$-structure $(\A, H)$, it will only be able to process the first $\Omega(|H|)$ updates in worst-case $\Oh[\Cc^\star_k, \varphi]{\log |H|}$ time each.
    Then we will use the technique of global rebuilding by Overmars and van~Leeuwen~\cite{DBLP:conf/wg/Overmars81,DBLP:journals/ipl/OvermarsL81a} to make $\D$ non-weak.
\end{itemize}

The base case is trivial.
The contraction step is formalized by the following lemma:

\begin{lemma}[\lemmaA]
  \label{lem:lemma_a}
  Given an~integer $k \in \N$, a~binary relational signature $\Sigma$,
    and a~sentence $\varphi \in \msotwo[\Sigma]$, there exist:
  \begin{itemize}[nosep]
    \item a~binary signature $\Gamma$;
    \item a~mapping $\Contract$ from augmented $\Sigma$-structures to augmented $\Gamma$-structures; and
    \item a~sentence $\psi \in \msotwo[\Gamma]$,
  \end{itemize}
  all computable from $k$, $\Sigma$, and $\varphi$, such that for every augmented $\Sigma$-structure $(\Af, H)$, if $(\Af^\star,H^\star)=\Contract(\Af,H)$, then:
  \begin{itemize}[nosep]
    \item $H \in \Cc_k$ implies $H^\star \in \Cc^\star_k$;
    \item $\A \models \varphi$ if and only if $\A^\star \models \psi$, and
    \item $|H^\star| \le |H|$.
  \end{itemize}
  Moreover, given an~efficient dynamic $(\Cc^\star_k, \Gamma, \psi)$-structure $\D^\star$, we can construct an~efficient dynamic $(\Cc_k, \Sigma, \varphi)$-structure $\D$.
\end{lemma}

The proof of Lemma~\ref{lem:lemma_a} is presented in Section~\ref{sec:lemmaA}.
The downgrade step is stated formally as follows:

\begin{lemma}[\lemmaB]
  \label{lem:lemma_b}
  Given an~integer $k \in \N$, $k>0$, a~binary relational signature $\Sigma$,
    and a~sentence $\varphi \in \msotwo[\Sigma]$, there exist:
  \begin{itemize}[nosep]
    \item a~binary relational signature $\Gamma$;
    \item a~mapping $\Downgrade$ from augmented $\Sigma$-structures to augmented $\Gamma$-structures; and
    \item a~sentence $\psi \in \msotwo[\Gamma]$,
  \end{itemize}
  all computable from $k$, $\Sigma$, and $\varphi$, such that for every augmented $\Sigma$-structure, if $(\wt{\Af},\wt{H})=\Downgrade(\Af,H)$, then:
  \begin{itemize}[nosep]
    \item $H \in \Cc^\star_k$ implies $\wt{H} \in \Cc_{k-1}$;
    \item $\Af \models \varphi$ if and only if $\wt{\Af} \models \psi$; and
    \item $|\wt{H}| \le |H|$.
  \end{itemize}
  Moreover, given an~efficient dynamic $(\Cc_{k-1}, \Gamma, \psi)$-structure $\wt{\D}$, we can construct
    a~\emph{weak} efficient dynamic $(\Cc^\star_k, \Sigma, \varphi)$-structure $\D$.
\end{lemma}

The proof of Lemma~\ref{lem:lemma_b} is presented in Section~\ref{sec:lemmaB}.
We follow with the global rebuilding technique that will be used by us to make the produced efficient dynamic $(\Cc^\star_k, \Sigma, \varphi)$-structure non-weak.
Here we adapt the statements from~\cite[Chapter V]{DBLP:books/sp/Overmars83} and \cite{DBLP:conf/cocoon/KosarajuP98}:

\begin{theorem}
    \label{thm:global_rebuilding}
    Consider a~dynamic data structure problem where the task is to maintain an~instance of a~problem dynamically under updates and answer queries regarding the current state of the instance.
    Assume each update changes the size of an~instance by at most a~constant.
    
    Suppose we are given a~data structure for the problem that:
    \begin{itemize}[nosep]
       \item can be initialized on an~instance of a~problem of size $n$ in time $T_{\mathrm{init}}(n)$;
       \item can process any sequence of $\Omega(n)$ updates in worst-case time $T_{\mathrm{update}}(n)$ each; and
       \item can answer any query in worst-case time $T_{\mathrm{query}}(n)$.
    \end{itemize}
    
    Then there exists a~data structure for the same dynamic problem that:
    \begin{itemize}[nosep]
        \item can be initialized on an~instance of a~problem of size $n$ in time $\Oh{T_{\mathrm{init}}(n)}$;
        \item can process any update in worst-case time $\Oh{T_{\mathrm{update}}(n) + T_{\mathrm{init}}(n) / n}$; and
        \item can answer any query in worst-case time $\Oh{T_{\mathrm{query}}(n)}$.
    \end{itemize}
\end{theorem}

With all necessary lemmas stated, we can give a proof of Theorem~\ref{thm:efficient_augmented_structure}.

\begin{proof}[Proof of Theorem~\ref{thm:efficient_augmented_structure}]
  We prove the following two families of properties by induction on~$k$:

  \medskip

  \emph{$\Pc_k$: for every $\varphi \in \msotwo[\Sigma]$, there exists an efficient dynamic
    $(\Cc_k, \Sigma, \varphi)$-structure.}
  
  \emph{$\Pc^\star_k$: for every $\varphi \in \msotwo[\Sigma]$, there exists an efficient dynamic
    $(\Cc^\star_k, \Sigma, \varphi)$-structure.}

  \paragraph*{Proof of $\Pc^\star_0$.}
    The only graph in $\Cc^\star_0$ is the null graph.
    Hence, after each batch of updates, the~relational structure $\A$ maintained by the structure must have an~empty universe, and may only contain flags.
    Thus, the postulated efficient dynamic structure only maintains the set of flags $c \subseteq \Sigma^0$, and after each batch of queries, checks whether $\varphi$ is satisfied for this set of flags.
    Each of these can be easily done in worst-case constant time per update.

  \paragraph*{$\Pc^\star_k$ implies $\Pc_k$ for every $k \in \N$.}
  Given a signature $\Sigma$ and a formula $\varphi$, we invoke \lemmaA (Lemma~\ref{lem:lemma_a}), and we compute $\Gamma$ and $\psi$ as in the statement of the lemma.
    Since $\Pc^\star_k$ holds, we take $\D^\star$ to be an~efficient dynamic $(\Cc^\star_k, \Gamma, \psi)$-structure.
    We then invoke Lemma~\ref{lem:lemma_a} again and conclude that there exists an~efficient dynamic $(\Cc_k, \Sigma, \varphi)$-structure $\D$.

  \paragraph*{$\Pc_k$ implies $\Pc^\star_{k+1}$ for every $k \in \N$.}
  We begin as previously, invoking \lemmaB (Lemma~\ref{lem:lemma_b})  instead of Lemma~\ref{lem:lemma_a}.
  The resulting weak efficient dynamic $(\C^\star_{k+1}, \Sigma, \varphi)$-structure can be easily turned into a~non-weak counterpart using Theorem~\ref{thm:global_rebuilding}.

   \medskip
    
  The three propositions above easily allow us to prove~$\Pc_k$ inductively.
  Therefore, the proof of the theorem is complete.
\end{proof}

\paragraph{Reduction from Theorem~\ref{thm:main_theorem} to Theorem~\ref{thm:efficient_augmented_structure}.}
Having established auxiliary Theorem~\ref{thm:efficient_augmented_structure}, we now present the proof of the main result of this work: Theorem~\ref{thm:main_theorem}.

Recall that in the announced reduction, we will use the technique of \emph{postponing invariant-breaking insertions} proposed by Eppstein et al.~\cite{EppsteinGIS96}.
Now, we state it formally.
In our description, we follow the notation of Chen et al.~\cite{ChenCDFHNPPSWZ21}.

Suppose $U$ is a~universe.
We say that a~family $\Fc \subseteq 2^U$ is \emph{downward closed} if $\emptyset \in \Fc$ and for every $S \in \Fc$, every subset of $S$ is also in $\Fc$.
Consider a~data structure $\F$ maintaining an~initially empty set $S \subseteq 2^U$ dynamically, under insertions and removals of single elements.
We say that $\F$:
\begin{itemize}
  \item \emph{strongly supports $\Fc$ membership} if $\F$ additionally offers a~query $\mathsf{member}()$ which verifies whether $S \in \Fc$; and
  \item \emph{weakly supports $\Fc$ membership} if $\F$ maintains $S$ dynamically under the invariant that $S \in \Fc$; however, if an~insertion of an~element into $S$ would violate the invariant, $\F$ must detect this fact and reject the query.
\end{itemize}

Then, Chen et al.~prove the following:

\begin{lemma}[{\cite[Lemma 11.1]{ChenCDFHNPPSWZ21}}]
  \label{lem:delaying_queries}
  Suppose $U$ is a~universe and let $M$ be a~dynamic dictionary over $U$. Let $\Fc \subseteq 2^U$ be downward closed and assume that there is a~data structure $\F$ that weakly supports $\Fc$ membership.
  
  Then, there exists a~data structure $\F'$ that strongly supports $\Fc$ membership, where each $\mathsf{member}$ query takes $\Oh{1}$ time, and each update takes amortized $\Oh{1}$ time and amortized $\Oh{1}$ calls to $M$ and $\F$.
  Moreover, $\F'$ maintains an instance of the data structure $\F$ and whenever $\mathsf{member}() = \mathsf{true}$, then it holds that $\F$ stores the same set $S$ as $\F'$.
\end{lemma}

We remark that the last assertion was not stated formally in~\cite{ChenCDFHNPPSWZ21}, but follows readily from the proof.
With the necessary notions in place, we proceed to the proof of Theorem~\ref{thm:main_theorem}.

\begin{proof}[Proof of Theorem~\ref{thm:main_theorem}]
  Let $k \in \N$, $\Sigma$ and $\varphi \in \msotwo[\Sigma]$ be as in the statement of the theorem.
  We define the following universe for the postponing invariant-breaking insertions technique:
  \[ U \coloneqq \Omega\ \cup\ \Sigma^{(0)}\ \cup\ \left(\Sigma^{(1)} \times \Omega\right)\ \cup\ \left(\Sigma^{(2)} \times \Omega \times \Omega\right), \]
  where $\Omega = \N$ is the space over which relational $\Sigma$-structures are defined.
  Given a~finite set $S \subseteq U$, we define the relational $\Sigma$-structure $\Af(S)$ described by $S$ by:
  \begin{itemize}
    \item defining $V(\Af(S))$ as the set of elements $v \in \Omega$ for which either $v \in S$, or $v$ participates in some tuple in $S$; and
    \item for $i \in \{0, 1, 2\}$, setting the interpretation of every relation $R \in \Sigma^{(i)}$ in $\Af(S)$ as the set of tuples $(a_1, \dots, a_i)$ such that $(R, a_1, \dots, a_i) \in S$.
  \end{itemize}
  Somewhat unusually, we say that $v$ belongs to $\Af(S)$ even when $v \notin S$, but $v$ participates in some tuple in $S$.
  The rationale behind this choice is that this will ensure the downward closure of the family that we will construct shortly.
  On the other hand, we cannot remove $\Omega$ from the definition of $U$; otherwise, elements of $\Af(S)$ not participating in any relations would not be tracked by $S$, but the satisfaction of $\varphi$ in $\Af(S)$ may depend on these elements.

  Let $\Fc_k \subseteq 2^U$ be the family of finite subsets of $U$ such that $S \in \Fc_k$ if the Gaifman graph of $\Af(S)$ has feedback vertex number at most $k$.
  Clearly, $\Fc_k$ is downward closed.
  We also have:
  
  \begin{claim}
  \label{cl:weak_fvs_support}
  There exists a~dynamic data structure $\F$ that maintains an~initially empty dynamic set $S \subseteq U$ and weakly supports $\Fc_k$ membership.
  Moreover, $\F$ is obliged to report whether $\Af(S) \models \varphi$ under the invariant that $S \in \Fc_k$.
  The worst-case update time is $f(\varphi, k) \cdot \log |S|$ for a computable function $f$.
  \end{claim}
  \begin{proof}
  Given $k$, we construct a sentence $\psi_k \in \msotwo[\Sigma]$ testing whether the Gaifman graph of the examined $\Sigma$-structure has feedback vertex number at most $k$.
  Using Theorem~\ref{thm:efficient_augmented_structure}, we set up two auxiliary efficient dynamic structures:
  \begin{itemize}
    \item $\F_\varphi$: an~efficient dynamic $(\Cc_k, \Sigma, \varphi)$-structure; and
    \item $\F_\psi$: an~efficient dynamic $(\Cc_{k+1}, \Sigma, \psi_k)$-structure.
  \end{itemize}
  Recall that $\F_\varphi$ and $\F_\psi$ operate on augmented $\Sigma$-structures, but our aim is to construct a~data structure $\F$ maintaining an~ordinary $\Sigma$-structure $\Af$.
  
  We proceed to the description of $\F$.
  Note that $\F$ should accept all queries which result in the Gaifman graph of $\Af(S)$ having feedback vertex number at most $k$, and reject all other queries.
  We keep an~invariant: if $\F$ currently maintains some finite set $S \subseteq \Omega$, then $S \in \Fc_k$, and both $\F_\varphi$ and $\F_\psi$ maintain the same augmented $\Sigma$-structure $(\Af(S), G(\Af(S)))$.
  
  We now show how to process changes of $\Af(S)$ under the changes to $S$ according to the invariant.
  Each such change may result in: an~addition of an~isolated vertex $v$ to $\Af(S)$, an~addition or removal of a~single relation in $\Af(S)$, and a~removal of an~isolated vertex $v$ from $\Af(S)$, in this order.
  Then:
  \begin{itemize}
    \item Each vertex addition and removal is forwarded verbatim to $\F_\varphi$ and $\F_\psi$.
    \item Each removal of a tuple from a relation is forwarded to $\F_\varphi$.
    If the removal of a~pair $(u, v)$ from the interpretation of some relation $R$ in $\Af(S)$ causes a~removal of an~edge $(u, v)$, $u < v$, from the edge set of $G(\Af(S))$, we follow by issuing the prescribed $\mathsf{delRelation}$ call, as well as $\mathsf{delEdge}(u, v)$ on both $\F_\psi$ and $\F_\varphi$.
    \item We consider additions of tuples to relations.
      Let $\Af = \Af(S)$, and let $\Af' = \Af(S')$ be the  structure after the update.
      If $E(G(\Af)) = E(G(\Af'))$, then the query may be simply relayed to $\F_\varphi$ and $\F_\psi$ since the feedback vertex number of $G(\Af)$ remains unchanged.

      Otherwise, $E(G(\Af))$ expands by some edge $uv$.
      We call $\mathsf{addEdge}(u, v)$ and $\mathsf{addRelation}$ with appropriate arguments in $\F_\psi$.
      The addition of a~single edge may increase the feedback vertex number of $G(\Af)$ by at most $1$; hence, $\fvs{G(\Af')} \leq k + 1$ and thus $G(\Af') \in \Cc_{k + 1}$.
      Therefore, $\F_\psi$ allows us to verify whether $G(\Af') \models \psi_k$; or equivalently, whether $S' \in \Fc_k$.
      
      If the condition $\fvs{G(\Af')} \leq k$ is satisfied, then we accept the query and forward the relation addition query to $\F_\varphi$.
      Otherwise, the relation addition is rejected; then, we roll back the update from $\F_\psi$ by calling $\mathsf{delRelation}$ and $\mathsf{delEdge}(u, v)$.
      In both cases, the invariants are maintained.
  \end{itemize}
  Note that $|G(\Af)| = \Oh{|S|}$.
  Thus, each update to $\F$ is translated to a~constant number of queries to $\F_\varphi$ and $\F_\psi$, hence it requires worst-case $\Oh[\varphi, k]{\log |G(\Af)|} = \Oh[\varphi, k]{\log |S|}$ time.
  The verification whether $\Af \models \varphi$ can be done by directly querying $\F_\varphi$, which can be done in constant time.
  \cqed\end{proof}
  
  Now, by applying Lemma~\ref{lem:delaying_queries}, we get the following:
  \begin{claim}
  \label{cl:strong_fvs_support}
  There exists a~dynamic data structure $\F'$ that maintains an~initially empty dynamic set $S \subseteq U$ and strongly supports $\Fc_k$ membership, where each $\mathsf{member}$ query takes $\Oh{1}$ time, and each update to $S$ takes amortized $f(\varphi, k) \cdot \log |S|$ time for some computable function $f$.
  Additionally, if $S \in \Fc_k$, $\F'$ is obliged to report whether $\Af(S) \models \varphi$.
  \begin{proof}
  We apply Lemma~\ref{lem:delaying_queries} and Claim~\ref{cl:weak_fvs_support}.
  Since the worst-case query time to $M$ is $\Oh{\log |S|}$, and the worst-case (hence also amortized) update time in $\F$ is $\Oh[\varphi,k]{\log |S|}$, the claimed amortized bound on the update time of $\F'$ is immediate.
  Moreover, if $S \in \Fc_k$, then Lemma~\ref{lem:delaying_queries} guarantees that $\F$ contains the same set $S$ as $\F'$.
  Thus, if $\fvs{G(\Af(S))} \leq k$, then $\F$ can be queried in constant time whether $\Af(S) \models \varphi$.
  \cqed\end{proof}
  \end{claim}
  
  From Claim~\ref{cl:strong_fvs_support}, the proof of the theorem is straightforward: we set up $\F'$ as in Claim~\ref{cl:strong_fvs_support}.
  Initially, the $\Sigma$-structure $\Af$ maintained by us is empty, hence we initialize $\F'$ with $S = \emptyset$.

  Each addition or removal of a~single vertex or a~single tuple in the maintained $\Sigma$-structure can be easily translated to a~constant number of element additions or removals in $S$ and forwarded to $\F'$.
  Here, we rely on the fact that a~vertex can be removed from $\Af$ only if it does not participate in any relations in $\Af$; otherwise, the removal of the vertex would require non-constant number of updates to $S$.
  Thus, after each query, we have that $\Af = \Af(S)$.

  Then, after each update concludes, if $\mathsf{member}() = \mathsf{false}$, then $\fvs{G(\Af)} > k$, so we respond \emph{Feedback vertex number too large}.
  Otherwise, we check in $\F'$ whether $\Af \models \varphi$, and return the result of this check.
  
  In order to verify the time complexity of the resulting data structure, we observe that at each point of time, we have $|S| \leq |\Sigma^{(0)}| + (|\Sigma^{(1)}| + 1) \cdot n + |\Sigma^{(2)}| \cdot n^2$, where $n = |V(\Af)|$.
  Thus, each update to the data structure takes amortized $\Oh[\varphi, k]{\log n}$ time.
  This concludes the proof.
\end{proof}

\section{\lemmaA}%
\label{sec:lemmaA}

We move on to the proof of the \lemmaA (Lemma~\ref{lem:lemma_a}).
The proof is comprised of multiple parts.
In Subsections~\ref{ssec:lemmaA_ferns} and~\ref{ssec:lemmaA_static}, we will prove the static variant: given $k \in \N$, a~binary relational signature $\Sigma$, and $\varphi \in \msotwo[\Sigma]$, we can (computably) produce a~new binary signature $\Gamma$, a~new formula $\psi \in \msotwo[\Gamma]$, and a~mapping $\Contract$ from augmented $\Sigma$-structures to augmented $\Gamma$-structures, with the properties prescribed by the statement of the lemma.
Then, in Subsections~\ref{ssec:lemmaA_dynamic_fern}, \ref{ssec:lemmaA_dynamic_contractions}, and \ref{ssec:lemmaA_conclusion}, we will lift the static variant to the full version of the lemma by showing that given an~efficient dynamic $(\Cc^\star_k, \Gamma, \psi)$-structure monitoring the satisfaction of $\psi$ in $\Contract((\Af, H))$, we can produce an~efficient dynamic $(\Cc_k, \Sigma, \varphi)$-structure monitoring the satisfaction of $\varphi$ in $(\Af, H)$.

In Section~\ref{ssec:lemmaA_ferns}, we consider a~plain graph-theoretic problem: given a~multigraph $H$, we define the fern decomposition $\Fc$ of $H$, as well as the {\em{quotient graph}} $\quo{H}{\Fc}$ obtained from $H$ by dissolving vertices of degree $0$, $1$, and $2$, or equivalently by contracting each fern. We refer to the Overview (Section~\ref{sec:overview}) for an~intuitive explanation of fern decompositions and contractions.
Here, we will solve the problem in a~more robust way than that presented in~\cite{AlmanMW20}: we will define an~equivalence relation $\sim$ on the edges of $H$ so that each element of $\Fc$ corresponds to exactly one equivalence class of $\sim$.
Then, in a~series of claims, we will prove that $\Fc$ has strong structural properties which will be used throughout the proof of the \lemmaA{}.
We stress that the notion of contracting the multigraph by dissolving vertices of degree at most $2$ is not novel; though, the definition of $\Fc$ through $\sim$ seems to be new.

In Section~\ref{ssec:lemmaA_static}, we lift the construction of the fern decomposition relational structures.
Given an~augmented $\Sigma$-structure $(\Af, H)$ and the fern decomposition $\Fc$ of $H$, we show how to create a~$\Sigma$-ensemble $\Xx$ such that $\Smash(\Xx) = \Af$, and so that every fern $S \in \Fc$ corresponds to a~unique ensemble element $\Af_S \in \Xx$.
The construction shall achieve two goals: on the one hand, $\Xx$ must be crafted in a~way allowing us to maintain it efficiently under the updates of $H$ and $\Af$.
In particular, we must ensure that no information on the interpretation of relations in $\Af$ is shared between multiple elements of $\Xx$, for otherwise, a~maliciously crafted update to $\Af$ could cause the need to recompute a~huge number of elements of $\Xx$.

On the other hand, the definition of $\Xx$ should allow us to reason about $\Contract^p(\Xx)$ for suitably chosen $p \in \N$.
Indeed, the construction of $\Xx$ is the crucial part in the proof of the static variant of the Contraction Lemma.
For our choice of $\Xx$, dependent on $\Af$ and $H$, the mapping $\Contract(\Af, H)$ claimed in the statement of Lemma~\ref{lem:lemma_a} will be exactly equal to $\Contract^p(\Xx)$ for some $p$ large enough.
Moreover, the Replacement Lemma (Lemma \ref{lem:replacement_lemma}) will allow us to (computably) find a~signature $\Gamma$ and a~formula $\psi \in \msotwo[\Gamma]$ so that $\Af \models \varphi$ if and only if $\Contract^p(\Xx) \models \psi$.
This will conclude the proof of the static variant of Lemma~\ref{lem:lemma_a}.

In Section~\ref{ssec:lemmaA_dynamic_fern}, we present a~dynamic version of the graph-theoretic problem solved in Section~\ref{ssec:lemmaA_ferns}: given a~dynamic multigraph $H$, which changes by additions and removals of edges and isolated vertices, maintain $\Fc$ (the fern decomposition) and $\quo{H}{\Fc}$ (the contraction) dynamically.
Each change to $H$ should be processed in worst-case $\Oh{\log n}$ time, causing each time a~constant number of changes to $\Fc$ and $\quo{H}{\Fc}$.
This is not a~new concept: an~essentially equivalent data structure has been presented by Alman et al.~\cite{AlmanMW20}.
However, it is slightly different in two different ways.
First, the strict definition of $\Fc$ requires us to perform a~more thorough case study; in particular, the data structure of Alman et al.~sometimes produced (few) vertices of degree $2$ in the contraction, which is unfortunately impermissible for us.
Second, we use top trees instead of link-cut trees; this change will be crucial in the next step of the proof.
Sections~\ref{ssec:lemmaA_static} and~\ref{ssec:lemmaA_dynamic_fern} can be read independently of each other.

In Section~\ref{ssec:lemmaA_dynamic_contractions}, we combine the findings of Sections~\ref{ssec:lemmaA_static} and~\ref{ssec:lemmaA_dynamic_fern}.
Namely, we show how the top trees data structure representing the (graph-theoretic) fern decomposition of $H$ can also be used to track the ensemble $\Xx$ constructed from $\Af$ and $H$, and to compute the types of the fern elements of the ensemble.
This, together with the vital properties of type calculus, such as compositionality under joins and idempotence, can be used to maintain $\Contract^p(\Xx)$ dynamically, under the changes to $H$ and $\Af$.
Each change will be processed in worst-case $\Oh[p, \Sigma]{\log n}$ time, producing at most $\Oh[p, \Sigma]{1}$ changes to the rank-$p$ contraction of $\Xx$.

In Section~\ref{ssec:lemmaA_conclusion}, we conclude by presenting an~efficient dynamic $(\Cc_k, \Sigma, \varphi)$-structure $\D$.
Namely, we instantiate three data structures: the data structures presented in Sections~\ref{ssec:lemmaA_dynamic_fern} and~\ref{ssec:lemmaA_dynamic_contractions}, maintaining the graph-theoretic contraction of $H$ and the rank-$p$ contraction of $\Xx$, respectively; and an~efficient dynamic $(\Cc^\star_k, \Gamma, \psi)$-structure $\D^\star$, whose existence is assumed by Lemma~\ref{lem:lemma_a}.
Then, each query to $\D$ is immediately forwarded to data structures from Sections~\ref{ssec:lemmaA_dynamic_fern} and~\ref{ssec:lemmaA_dynamic_contractions}, producing constant-size sequences of changes to $\quo{H}{\Fc}$ and $\Contract^p(\Xx)$.
From these, we produce a~batch of changes to $\D^\star$ of constant size that ensures that the vertices and edges of the multigraph maintained by $\D^\star$ are given by $\quo{H}{\Fc}$, and the relations of the relational structure are given by $\Contract^p(\Xx)$.
It will be then proved that after $\D^\star$ finishes processing the batch, resulting in an~augmented structure $(\Af^\star, H^\star)$, we will have $\Af^\star \models \psi$ if and only if $\Af \models \varphi$.
This will establish the proof of the correctness of $\D$ and conclude the proof of the \lemmaA{}.

\subsection{Fern decomposition}
\label{ssec:lemmaA_ferns}
We start with describing a form of a decomposition of a multigraph that will be maintained by the data structure, which we call a {\em{fern decomposition}}. This decomposition was implicit in countless earlier works on parameterized algorithms for the {\sc{Feedback Vertex Set}} problem, as it is roughly the result of exhaustively dissolving vertices of degree at most $2$ in a graph. In particular, it is also present in the work of Alman et al.~\cite{AlmanMW20}, where the main idea, borrowed here, is to maintain this decomposition dynamically. The difference in the layer of presentation is that the earlier works mostly introduced the decomposition through the aforementioned dissolution procedure, which makes it more cumbersome to analyze. Also, following this approach makes it not obvious (though actually true) that the final outcome is independent of the order of dissolutions. Here, we prefer to introduce the fern decomposition in a more robust way, which will help us later when we will be working with $\msotwo$ types of its components.

In this section we work with multigraphs, where we allow multiple edges with the same endpoints and self-loops at vertices. An {\em{incidence}} is a pair $(u,e)$, where $u$ is a vertex and $e$ is an edge incident to $u$. By slightly abusing the notation, we assume that if $e$ is a self-loop at $u$, then $e$ creates two different incidences $(u,e)$ with $u$. The {\em{degree}} of a vertex is the number of incidences in which it participates. Note that thus, every self-loop is counted twice when computing the degree.

Similarly as for relational structures, a {\em{boundaried multigraph}} is a multigraph $H$ supplied with a subset of its vertices $\bnd H$, called the {\em{boundary}}. A {\em{fern}} is a boundaried multigraph $H$ satisfying the following conditions:
\begin{itemize}
 \item $|\bnd H|\leq 2$.
 \item If $|\bnd H|=2$, then $H$ is a tree in which both vertices of $\bnd H$ are leaves.
 \item If $|\bnd H|=1$, then $H$ is either a tree or a unicyclic graph. In the latter case, the unique boundary vertex of $H$ has degree $2$ and lies on the unique cycle of $H$.
 \item If $|\bnd H|=0$, then $H$ is a tree or a unicyclic graph.
\end{itemize}
Here, a {\em{unicyclic graph}} is a connected graph that has exactly one cycle; equivalently, it is a~connected graph where the number of edges matches the number of vertices. Note that by definition, every fern is connected. If $H$ is a~tree, we say that $H$ is a~{\em{tree fern}}, and if $H$ is a~unicyclic graph, we say that $H$ is a~{\em{cyclic fern}}.

If $F$ is a subset of edges of a multigraph $H$, then $F$ {\em{induces}} a boundaried multigraph $H[F]$ consisting of all edges of $F$ and vertices incident to them. The boundary of $H[F]$ consists of all vertices of $H[F]$ that in $H$ are also incident to edges outside of $F$. For a multigraph $H$ and partition $\Ff$ of the edge set of $H$, we define
$$H[\Ff]\coloneqq \{H[F]\colon F\in \Ff\}\cup \{((u,\emptyset),\emptyset)\colon u\textrm{ is isolated in }H\}.$$
and call it the {\em{decomposition}} of $H$ induced by $\Ff$. Note that to this decomposition we explicitly add a single-vertex graph (with empty boundary) for every isolated vertex of $H$, so that every vertex of $H$ belongs to at least one element of the decomposition.

Consider a multigraph $H$.
We are now going to define a partition $\Ff$ of the edge set of $H$ so that every element of $H[\Ff]$ is a fern and some additional properties are satisfied; these will be summarized in Lemma~\ref{lem:fern-decomposition}.

An edge $e$ in $H$ shall be called {\em{essential}} if it satisfies one of the following conditions:
\begin{itemize}
 \item $e$ lies on a cycle in $H$; or
 \item $e$ is a bridge and removing $e$ from $H$ creates two new components, each of which contains a cycle.
\end{itemize}
A vertex $u$ is {\em{essential}} if it participates in at least three incidences with essential edges, where every self-loop at $u$ is counted twice. Vertices that are not essential are called {\em{non-essential}}. An incidence $(u,e)$ is {\em{critical}} if both $u$ and $e$ are essential.

Define the following relation $\sim$ on the edge set of $H$: $e\sim f$ if and only if there exists a~walk
$$u_0\stackrel{e_1}{-}u_1\stackrel{e_2}{-}u_2\stackrel{e_3}{-}\ldots\stackrel{e_{\ell-1}}{-}u_{\ell-1}\stackrel{e_\ell}{-}u_\ell,$$
where $e_1=e$, $e_\ell=f$, and $e_i$ has endpoints $u_{i-1}$ and $u_i$ for all $i\in [\ell]$, such that for each $i\in [\ell-1]$, the incidences $(u_i,e_i)$ and $(u_i,e_{i+1})$ are not critical. Note that we allow the incidences $(u_0,e_1)$ and $(u_\ell,e_\ell)$ to be critical. A walk $W$ satisfying the condition stated above will be called {\em{safe}}.

We have the following observations.

\begin{lemma}
 For every multigraph $H$, $\sim$ is an equivalence relation on the edge set of $H$.
\end{lemma}
\begin{proof}
 The only non-trivial check is transitivity. Suppose then that $e,f,g$ are pairwise different edges such that $e\sim f$ and $f\sim g$, hence there are safe walks $W_{ef}$ and $W_{fg}$ such that $W_{ef}$ starts with $e$ and ends with $f$, while $W_{fg}$ starts with $f$ and ends with $g$.
 We consider two cases, depending on whether $f$ is traversed by $W_{ef}$ and $W_{fg}$ in the same or in opposite directions.
 
 If $f$ is traversed by $W_{ef}$ and $W_{fg}$ in the same direction, then construct $W_{eg}$ by concatenating $W_{ef}$ with $W_{fg}$ with the first edge removed. Then $W_{eg}$ is a walk that starts with $e$, ends with $g$, and it is easy to see that it is safe. Therefore $e\sim g$.
 
 If $f$ is traversed by $W_{ef}$ and $W_{fg}$ in opposite directions, then construct $W_{eg}$ by concatenating $W_{ef}$ with the last edge removed with $W_{fg}$ with the first edge removed. Again,  $W_{eg}$ is a safe walk that starts with $e$ and ends with $g$, so $e\sim g$.
\end{proof}

\begin{lemma}\label{lem:connection-essential}
Let $e,f$ be essential edges of a~multigraph $H$, and let $W$ be any walk in $H$ that starts with $e$, ends with~$f$, and traverses every edge at most once. Then every edge traversed by $W$ is essential.
\end{lemma}
\begin{proof}
For contradiction, suppose some edge $g$ of $W$ is non-essential. Clearly $g\notin \{e,f\}$. By definition, $g$ is a bridge in $H$ and the removal of $g$ from $H$ creates two new connected components, say $C_1$ and $C_2$, out of which at least one, say $C_1$, is a tree. Since $g$ is traversed only once by $W$, it follows that $e$ and $f$ do not belong to the same component among $C_1,C_2$; by symmetry suppose $e\in E(C_1)$ and $f\in E(C_2)$. Since $C_1$ is a tree and $g$ is a bridge, $e$ cannot be contained in any cycle in $H$; in other words, $e$ is a bridge as well. Moreover, if one removes $e$ from $H$, then one of the resulting new components is a subtree of $C_1$, and hence is acyclic. This means that $e$ is non-essential, a contradiction.  
\end{proof}

\begin{lemma}\label{lem:vertex-connection-essential}
Let $u,v$ be two different essential vertices of a~multigraph $H$, and let $P$ be any (simple) path in $H$ that starts with $u$ and ends with $v$. Then every edge traversed by $P$ is essential.
\end{lemma}
\begin{proof}
The path $P$ traverses one edge incident to $u$ and one edge incident to $v$. Since each of $u$ and $v$ participates in three different critical incidences, we may find essential edges $e$ and $f$, incident to $u$ and $v$, respectively, such that $e\neq f$ and neither $e$ nor $f$ is traversed by $P$. Then adding $e$ and $f$ at the front and at the end of $P$, respectively, yields a walk $W$ that starts with $e$, ends with $f$, and traverses every edge at most once. It remains to apply Lemma~\ref{lem:connection-essential} to~$W$.
\end{proof}

With the above observations in place, we may formulate the main result of this section.

\begin{lemma}[Fern Decomposition Lemma]\label{lem:fern-decomposition}
Let $H$ be a multigraph and let $\Ff$ be the partition of the edge set of $H$ into the equivalence classes of the equivalence relation $\sim$ defined above. Then each element of $H[\Ff]$ is a fern, every non-essential vertex of $H$ belongs to exactly one element of $H[\Ff]$, and $\bigcup_{S\in H[\Ff]} \bnd S$ comprises exactly the essential vertices of $H$.

Moreover, define the quotient multigraph $\quo{H}{\Ff}$ on the vertex set $\bigcup_{S\in H[\Ff]} \bnd S$ by adding:
\begin{itemize}
 \item one edge $uv$ for each tree fern $S\in H[\Ff]$ with $\bnd S=\{u,v\}$, $u\neq v$; and
 \item one loop at $u$ for each cyclic fern $S\in H[\Ff]$ with $\bnd S=\{u\}$.
\end{itemize}
Then $\fvs{\quo{H}{\Ff}}\leq \fvs{H}$ and every vertex of $\quo{H}{\Ff}$ has degree at least $3$ in $\quo{H}{\Ff}$.
\end{lemma}
\begin{proof}
 We develop consecutive properties of $H[\Ff]$ in a series of claims. Whenever we talk about essentiality or criticality, we mean it in the graph $H$.

 \begin{claim}\label{cl:critical-inc-bound}
  For each $F\in \Ff$, $H[F]$ contains at most $2$ critical incidences.
 \end{claim}
 \begin{proof}
 For contradiction, suppose there are three different critical incidences $(u_1,e_1)$, $(u_2,e_2)$, $(u_3,e_3)$ such that $e_1,e_2,e_3\in F$. We may assume that $e_1,e_2,e_3$ are pairwise different, for otherwise, if say $e_1=e_2$, then $(u_1,e_1)$ and $(u_2,e_2)$ are the two incidences of $e_1$ and both of them are critical, implying that $F=\{e_1\}$ and $H[F]$ contains two incidences in total.
 
 Since $e_1\sim e_2$ and $e_1\sim e_3$, there are safe walks $W_{12}$ and $W_{13}$ that both start with $e_1$ and end with $e_2$ and $e_3$, respectively.  By shortcutting $W_{12}$ and $W_{13}$ if necessary we may assume that each of them traverses every edge at most once. Hence, by Lemma~\ref{lem:connection-essential}, every edge traversed by $W_{12}$ or $W_{13}$ is essential.
 Note that since $(u_1,e_1)$ is critical, both $W_{12}$ and $W_{13}$ must start at $u_1$, and hence they traverse $e_1$ in the same direction.
 
 Let $R$ be the maximal common prefix of $W_{12}$ and $W_{13}$, and let $v$ the the second endpoint of $R$. Noting that $R$ is neither empty nor equal to $W_{12}$ or $W_{13}$, we find that $v$ is incident to three different essential edges: one in the prefix $R$ and two on the suffixes of $W_{12}$ and $W_{13}$ after $R$, respectively. It follows that $v$ is essential, and consequently the incidences between $v$ and the incident edges on $W_{12}$ and $W_{13}$ are critical. This contradicts the assumption that $W_{12}$ and $W_{13}$ are safe.
 \cqed\end{proof}

 \begin{claim}\label{cl:bnd-bound}
  For each $F\in \Ff$, $|\bnd H[F]|\leq 2$.
 \end{claim}
 \begin{proof}
 Since edges incident to a non-essential vertex are always pairwise $\sim$-equivalent, it follows that every vertex of $\bnd H[F]$ is essential. Since $H[F]$ is connected by definition, for every pair of different vertices $u,v\in \bnd H[F]$ there exists a path in $H[F]$ connecting $u$ and $v$. By Lemma~\ref{lem:vertex-connection-essential}, each edge of this path is essential. Therefore, if $|\bnd H[F]| > 1$, then every vertex $u\in \bnd H[F]$ participates in a critical incidence in $H[F]$. As by Claim~\ref{cl:critical-inc-bound} there can be at most $2$ critical incidences in $H[F]$, we conclude that $|\bnd H[F]|\leq 2$.
 \cqed\end{proof}

 \begin{claim}\label{cl:tree-or-unicyclic}
  For each $F\in \Ff$, $H[F]$ is either a tree or a unicyclic graph.
 \end{claim}
 \begin{proof}
 By definition $H[F]$ is connected. It therefore suffices to show that it cannot be the case that $H[F]$ contains two different cycles. For contradiction, suppose there are such cycles, say $C$ and $D$. We consider two cases: either $C$ and $D$ share a vertex or are vertex-disjoint.
 
 Assume first that $C$ and $D$ share a vertex. Since $C$ and $D$ are different, there must exist a vertex $u$ that participates in three different incidences with edges of $E(C)\cup E(D)$. By definition, every edge of $C$ and every edge of $D$ is essential. Therefore, $u$ is essential and involved in three different critical incidences in $H[F]$. This is a contradiction with Claim~\ref{cl:critical-inc-bound}.
 
 Assume then that $C$ and $D$ are vertex-disjoint. Since $H[F]$ is connected, we can find a path $P$ in $H[F]$ whose one endpoint $u$ belongs to $V(C)$, the other endpoint $v$ belongs to $V(D)$, while all the internal vertices of $P$ do not belong to $V(C)\cup V(D)$. Observe that every edge traversed by $P$ is essential, for it cannot be a bridge whose removal leaves one of the resulting components a tree. Therefore, $u$ participates in three different incidences with essential edges  in $H[F]$: two with edges of $C$ and one with the first edge of $P$. Again, we find that $u$ is essential and involved in three different critical incidences in $H[F]$, a contradiction with Claim~\ref{cl:critical-inc-bound}.
 \cqed\end{proof}

 \begin{claim}\label{cl:boundary-on-cycle}
  If $H[F]$ contains a cycle $C$, then every vertex of $\bnd H[F]$ belongs to $V(C)$.
 \end{claim}
 \begin{proof}
  Suppose there is a vertex $u\in \bnd H[F]$ that does not lie on $C$. Since $H[F]$ is connected, there is a path $P$ in $H[F]$ from $u$ to a vertex $v\in V(C)$ that is vertex-disjoint with $C$ except for $v$. 
  
  Since $u\in \bnd H[F]$, $u$ is essential, hence we can find an essential edge $e$ incident to $u$ that is not traversed by $P$. Let $f$ be any edge of $C$ that is incident to $v$. Then adding $e$ and $f$ at the front and at the end of $P$, respectively, yields a walk in $H$ that starts with $e$, ends with $f$, and passes through every edge at most once. By Lemma~\ref{lem:connection-essential} we infer that every edge of $P$ is essential.
  
  By definition, every edge of $C$ is also essential.
  Similarly as before, we find that $v$ participates in three different incidences with essential edges: two with edges from $C$ and one with the last edge of $P$. Hence $v$ is essential and creates three critical incidences in $H[P]$, a~contradiction with Claim~\ref{cl:critical-inc-bound}.
 \cqed\end{proof}

 \begin{claim}\label{cl:bnd2}
  If $|\bnd H[F]|=2$ for some $F\in \Ff$, then $H[F]$ is a tree and both vertices of $\bnd H[F]$ are leaves of this tree.
 \end{claim}
 \begin{proof}
 Let $\bnd H[F]=\{u,v\}$.
 Suppose that $H[F]$ contains a cycle $C$. By Claim~\ref{cl:boundary-on-cycle} we have $u,v\in V(C)$, in particular $C$ is not a self-loop. Let $e,f$ be the two edges of $C$ that are incident to $u$. By definition of $\sim$, there is a safe walk $W$ in $H$ that starts with $e$ and ends with $f$, and we may assume that $W$ passes through every edge at most once. By Lemma~\ref{lem:connection-essential}, all edges of $W$ are essential. Note that every edge of $W$ is $\sim$-equivalent with both $e$ and $f$, hence $W$ is contained in $H[F]$.  
 
 As $u,v\in \bnd H[F]$, both $u$ and $v$ are essential, and hence they cannot be internal vertices of the safe walk $W$, as they would create critical incidences with neighboring edges of $W$. It follows that $W$ is actually a closed walk in $H[F]$ that does not pass through $v$, hence it contains a cycle that is different from $C$. This is a contradiction with the unicyclicity of $H[F]$, following from Claim~\ref{cl:tree-or-unicyclic}.
 
 Therefore $H[F]$ is indeed a tree. Suppose now that one of the vertices of $\bnd H[F]$, say $u$, is incident on two different edges of $F$, say $e$ and $f$. Since $e\sim f$, there is a safe walk $W$ in $H$ that starts with $e$ and ends in $f$, and we may assume that $W$ passes through every edge at most once. Again, $u\in \bnd H[F]$ implies that $u$ is essential, hence $u$ cannot be an internal vertex of $W$. So $W$ forms a non-empty closed walk in $H[F]$ that passes through every edge at most once, a contradiction with the fact that $H[F]$ is a tree.
 \cqed\end{proof}

 \begin{claim}\label{cl:bnd1}
  If $|\bnd H[F]|=1$ for some $F\in \Ff$, then $H[F]$ is either a tree or a unicyclic graph. In the latter case, the unique boundary vertex of $H[F]$ has degree $2$ in $H[F]$ and lies on the unique cycle of $H[F]$.
 \end{claim}
 \begin{proof}
 Let $u$ be the unique vertex of $\bnd H[F]$. 
 That $H[F]$ is a tree or a unicyclic graph is implied by Claim~\ref{cl:tree-or-unicyclic}. It remains to prove that if $H[F]$ is unicyclic, then $u$ has degree $2$ in $H[F]$ and it lies on the unique cycle $C$ of $H[F]$. That $u$ lies on $C$ follows by Claim~\ref{cl:boundary-on-cycle}. Note that $u\in \bnd H[F]$ implies that $u$ is essential.
 
 So assume, for the sake of contradiction, that $u$ is incident on some edge $f$ that does not belong to $C$. Since $H[F]$ is unicyclic, $C$ is the only cycle in $H[F]$, and hence $f$ is a bridge whose removal splits $H[F]$ into two connected components. One component of $H[F]-f$ contains $C$, while the other must be a tree, for $C$ is the only cycle in $H[F]$. Let $e$ be any edge of $C$ incident to $u$. Clearly $e$ is essential, hence the incidence $(u,e)$ is critical. Since $e\sim f$, there is a safe walk $W$ in $H$ that starts with $e$ and ends with $f$; note that as before, $W$ is entirely contained in $H[F]$. Again, we may assume that $W$ passes through every edge at most once. Note that since $W$ is safe, it needs to start at the vertex $u$, as otherwise the critical incidence $(u,e)$ is not among the two terminal incidences on $W$. Since the last edge of $W$ is $f$, which is a bridge, the penultimate vertex traversed by $W$ must be $u$ again. Then $W$ with the last edge removed forms a closed walk in $H[F]$ that passes through every edge at most once, which means that every edge on $W$ except for $f$ must belong to some cycle in $H[F]$, and hence is essential. In particular, the edge $e'$ traversed by $W$ just before the last visit of $u$ is essential as well. Now the incidence $(u,e')$ is critical, appears on $W$, and is not among the two terminal incidences on $W$. This is a contradiction with the safeness of~$W$.
 \cqed\end{proof}

 \begin{claim}\label{cl:bnd0}
  If $|\bnd H[F]|=0$ for some $F\in \Ff$, then $H[F]$ is either a tree or a unicyclic graph.
 \end{claim}
 \begin{proof}
 Follows immediately from Claim~\ref{cl:tree-or-unicyclic}.
 \cqed\end{proof}
 
 From Claims~\ref{cl:bnd-bound},~\ref{cl:bnd2},~\ref{cl:bnd1}, and~\ref{cl:bnd0} it follows that for every $F\in \Ff$, $H[F]$ is a fern. Moreover, since all edges incident to a non-essential vertex are pairwise $\sim$-equivalent, it follows that every non-essential vertex of $H$ belongs to exactly one element of $H[\Ff]$. Also, we observe the following.
 
 \begin{claim}\label{cl:bnd-ess}
  The set $\bigcup_{S\in H[\Ff]} \bnd S$ comprises exactly the essential vertices of $H$.
 \end{claim}
 \begin{proof}
  If $u$ is non-essential, then all edges incident to $u$ are pairwise $\sim$-equivalent and $u$ does not participate in any boundary of an element of $\Ff$.
  On the other hand, if $u$ is essential, then it participates in at least three different critical incidences. By Claim~\ref{cl:critical-inc-bound}, they cannot all belong to the same multigraph $H[F]$ for any $F\in \Ff$, hence $u$ is incident to edges belonging to at least two different elements of $\Ff$, say $F$ and $F'$. It follows that $u\in \bnd H[F]\cap \bnd H[F']$. 
 \cqed\end{proof}
 
 We are left with verifying the asserted properties of the quotient graph $\quo{H}{\Ff}$. Let $U=\bigcup_{S\in H[\Ff]} \bnd S=V(\quo{H}{\Ff})$ be the set of essential vertices of $H$.
 
 \begin{claim}\label{cl:fvs-smaller}
  $\fvs{\quo{H}{\Ff}}\leq \fvs{H}$.
 \end{claim}
 \begin{proof}
 Let $X$ be a feedback vertex set of $H$ of size $\fvs{H}$. We construct a set of vertices $X'\subseteq U$ as follows:
 \begin{itemize}
  \item For each $x\in X\cap U$, add $x$ to $X'$.
  \item For each $x\in X\setminus U$, let $F_x$ be the unique element of $\Ff$ such that $x\in V(H[F_x])$. Then, provided $\bnd H[F_x]$ is nonempty, add an arbitrary element of $\bnd H[F_x]$ to $X'$.
 \end{itemize}
 Clearly $|X'|\leq |X|=\fvs{H}$. Therefore, it suffices to argue that $X'$ is a feedback vertex set in $\quo{H}{\Ff}$.

 Consider any cycle $C'$ in $\quo{H}{\Ff}$. Construct a cycle $C$ in $H$ from $C'$ as follows:
 \begin{itemize}
  \item If $C'$ consists of a self-loop at vertex $u$, then this self-loop corresponds to a cyclic fern $H[F]$ for some $F\in \Ff$, and $u$ lies on the unique cycle of $H[F]$. Then we let $C$ be this unique cycle.
  \item Otherwise, each edge $e=uv$ traversed by $C'$ corresponds to a tree fern $H[F_e]$ for some $F_e\in \Ff$, where $\bnd F_e=\{u,v\}$. Then replace $e$ with the (unique) path in $H[F_e]$ connecting $u$ and $v$, and do this for every edge of $C'$. It is easy to see that this yields a cycle in $H$, which is $C$. 
 \end{itemize}
 As $X$ is a feedback vertex set of $H$, there is some $x\in X$ that lies on $C$. It is then easy to see that the vertex added for $x$ to $X'$ lies on $C'$. Since $C'$ was chosen arbitrarily, we conclude that $X'$ is a feedback vertex set of $\quo{H}{\Ff}$. 
 \cqed\end{proof}

 \begin{claim}\label{cl:mindeg}
  In $\quo{H}{\Ff}$, every vertex has degree at least $3$.
 \end{claim}
 \begin{proof}
  By Claim~\ref{cl:bnd-ess}, every element of $U=V(\quo{H}{\Ff})$ is an essential vertex of $H$. Since every essential vertex participates in at least $3$ critical incidences, it suffices to show that the degree of a vertex $u\in U$ in $\quo{H}{\Ff}$ matches the number of critical incidences that $u$ participates in. This is easy to see: every tree fern $S\in H[\Ff]$ with $u\in \bnd S$ contributes $1$ to the degree of $u$ in $\quo{H}{\Ff}$ and contains one critical incidence in which $u$ participates, while each cyclic fern $S$ with $u\in \bnd S$ contributes $2$ to the degree of $u$ in $\quo{H}{\Ff}$ and contains two critical incidences in which $u$ participates.
 \cqed\end{proof}

 Claims~\ref{cl:fvs-smaller} and~\ref{cl:mindeg} finish the proof of Lemma~\ref{lem:fern-decomposition}.
\end{proof}

The decomposition $H[\Ff]$ provided by Lemma~\ref{lem:fern-decomposition} will be called the {\em{fern decomposition}} of~$H$.

\subsection{Static variant of the \lemmaA{}}
\label{ssec:lemmaA_static}
In this section we use the notion of a fern decomposition, introduced in the previous section, to prove the static variant of the \lemmaA{}. Intuitively, given an augmented structure $(\Af,H)$, we make a fern decomposition $H[\Ff]$ of $H$ and split $\Af$ into an ensemble $\Xx$ accordingly. The augmented structure $(\Af^\star,H^\star)=\Contract(\Af,H)$ is defined as follows: $H^\star=\quo{H}{\Ff}$ and $\Af^\star=\Contract^p(\Xx)$ for $p$ sufficiently large so that the Replacement Lemma can be applied to conclude that $\Af^\star$ contains enough information to infer the rank-$q$ type of $\Af$; here, $q$ is the quantifier rank of the given sentence $\varphi$. However, the split of $\Af$ into $\Xx$ needs to be done very carefully so that we will be able to maintain it in a dynamic data structure.

We proceed to a formal description.
First, let $\Ff$ be the partition of the edges of $H$ into equivalence classes of the relation $\sim$ defined in Section~\ref{ssec:lemmaA_ferns}. Then $H[\Ff]$ is the fern decomposition of $H$, with properties described by the Fern Decomposition Lemma (Lemma~\ref{lem:fern-decomposition}). 

Now, our goal is to carefully partition $\Af$ into an ensemble $\Xx$ ``along'' the fern decomposition $H[\Ff]$. Recall that $H$ is a supergraph of the Gaifman graph of $\Af$, which means that if vertices $v,w$ are bound by some relation in $\Af$, then at least one edge $vw$ is present in $H$. Let $U$ be the set of essential vertices of $H$. The ensemble $\Xx$ is defined as follows.
\begin{itemize}
 \item For every fern $S\in H[\Ff]$, create a boundaried $\Sigma$-structure $\Af_S$, where the universe of $\Af_S$ is $V(S)$ and the boundary is $\bnd S$. So far make all relations in $\Af_S$ empty.
 \item For every $R\in \Sigma^{(1)}\cup \Sigma^{(2)}$ and every $w\in V(\Af)\setminus U$ such that $R(w)$ (or $R(w,w)$ in case $R$ is binary) holds in $\Af$, find the unique fern $S\in H[\Ff]$ that contains $w$ ($S$ exists and is unique by Lemma~\ref{lem:fern-decomposition}). Then make $R(w)$ (resp. $R(w,w)$) hold in $\Af_S$, that is, add $w$ (resp. $(w,w)$) to $R^{\Af_S}$.
 \item For every $R\in \Sigma^{(2)}$ and every pair of distinct vertices $v,w\in V(\Af)$ with $(v,w)\in R^{\Af}$, choose any fern $S\in H[\Ff]$ such that $vw$ is an edge in $S$ (such a fern exists by the assumption that $H$ is a supergraph of the Gaifman graph of $\Af)$. Then make $R(v,w)$ hold in $\Af_S$, that is, add $(v,w)$ to $R^{\Af_S}$.
 \item For every $u\in U$, create a boundaried $\Sigma$-structure $\Af_u$ with both the universe and the boundary consisting only of $u$. The structure $\Af_u$ retains the interpretation of all unary and binary relations on $u$ from $\Af$: For each $R\in \Sigma^{(1)}$ we have $R^{\Af_u}=R^{\Af}\cap \{u\}$ and for each $R\in \Sigma^{(2)}$ we have $R^{\Af_u}=R^{\Af}\cap \{(u,u)\}$. Note that the values of nullary predicates are {\em{not}} retained from $\Af$.
 \item Finally, we create a boundaried $\Sigma$-structure $\Af_\emptyset$ with empty universe and boundary that retains the interpretation of all nullary predicates from $\Af$: for each $R\in \Sigma^{(0)}$, we have $R^{\Af_\emptyset}=R^{\Af}$.
\end{itemize}
The ensemble $\Xx$ comprises all $\Sigma$-structures described above, that is, structures $\Af_S$ for $S\in H[\Ff]$, $\Af_u$ for $u\in U$, and $\Af_\emptyset$. These elements of $\Xx$ will be respectively called {\em{fern elements}}, {\em{singleton elements}}, and the {\em{flag element}}.

From the construction we immediately obtain the following.

\begin{lemma}\label{lem:fern-smash}
 For each $R\in \Sigma$, $\{R^\Xf\colon \Xf\in \Xx\}$ is a partition of $R^\Af$. Moreover,
 $$\Af=\Smash(\Xx).$$
 Also, for every fern element $\Af_S$ of $\Xx$, $S$ is a supergraph of the Gaifman graph of $\Af_S$.
\end{lemma}

Let us discuss the intuition. The information about $\Af$ is effectively partitioned among the elements of ensemble $\Xx$.
Fern elements store all information about binary relations between distinct vertices and unary and binary relations on non-essential vertices; they effectively are induced substructures of $\Af$, except that the relations on boundaries are cleared. Note that in some corner cases, a single tuple $(u,v)\in R^{\Af}$, where $u\neq v$, may be stored in several different fern elements, and hence the fern to store it is chosen non-deterministically in the construction. For instance, if in $H$ there are multiple parallel edges connecting $u$ and $v$, then they are all in different ferns, and the tuple $(u,v)$ can be stored in any single of them.
Singleton elements store information concerning single essential vertices. The idea is that when this information is updated, we only need to update a single singleton element corresponding to an essential vertex, rather than all fern elements containing this essential vertex on respective boundaries.
Similarly, the flag element stores the information on flags in $\Af$, so that it can be quickly updated without updating all other elements of $\Xx$.

We proceed to the proof of the static variant of \lemmaA{}. Recall that we work with a given sentence $\varphi\in \msotwo[\Sigma]$. Let $q$ be the rank of $\varphi$.
By the Replacement Lemma (Lemma~\ref{lem:replacement_lemma}) and Lemma~\ref{lem:fern-smash}, we may compute a number $p$, a signature $\Gamma^p$, and a mapping $\Infer\colon \Types^{p,\Gamma^p}\to \Types^{q,\Sigma}$ such that
$$\tp^q(\Af)=\Infer\left(\tp^p(\Contract^p(\Xx))\right).$$
Hence, for \lemmaA{} it suffices to set
$$\Gamma\coloneqq \Gamma^p,\qquad \Contract(\Af,H)\coloneqq (\Contract^p(\Xx),\quo{H}{\Ff}),$$
where $\quo{H}{\Ff}$ is the quotient graph defined in the statement of the Fern Decomposition Lemma, and
$$\psi\coloneqq\bigvee \left\{ \bigwedge \alpha\colon \alpha\in \Types^{p,\Gamma^p}\textrm{ such that } \varphi\in \Infer(\alpha)\right\},$$
where $\bigwedge \alpha$ denotes the conjunction of all sentences contained in type $\alpha$. That these objects satisfy the conclusion of \lemmaA{} follows directly from the Replacement Lemma and the Fern Decomposition Lemma.

\subsection{Dynamic maintenance of fern decomposition}
\label{ssec:lemmaA_dynamic_fern}
\newcommand{\dcrit}{\textsf{critical\_degree}}
\newcommand{\addquo}{\textsf{addQuotient}}
\newcommand{\delquo}{\textsf{delQuotient}}
\newcommand{\attach}{\textsf{attachTree}}
\newcommand{\detach}{\textsf{detachTree}}
\newcommand{\edgerep}{\textsf{edge\_representatives}}
\newcommand{\name}{\textsf{name}}
\newcommand{\makeess}{\textsf{makeEssential}}
\newcommand{\makenon}{\textsf{makeNonEssential}}
\newcommand{\Add}{\mathrm{add}}
\newcommand{\Del}{\mathrm{del}}
\newcommand{\attachment}{\textsf{attachment}}
\newcommand{\Trees}{\textsf{Trees}}
\newcommand{\joint}{\textsf{joint\_edge}}

Let $H$ be a~multigraph, and let $\Fc$ be a~partition of its edges into equivalence classes of the relation~$\sim$ defined in Section~\ref{ssec:lemmaA_ferns}.
In Section~\ref{ssec:lemmaA_static} we saw that the quotient multigraph $\quo{H}{\Fc}$ can be used in the proof of the static variant of the~\lemmaA.
Now, we show how to efficiently maintain such a~multigraph together with the~fern decomposition~$H[\Fc]$ assuming that the~input multigraph~$H$ is dynamically modified.
The idea is to maintain a~forest~$\Delta_\Tc$ of top trees such that each tree $T \in \Tc$ corresponds to a~different fern $S \in H[\Fc]$.
It turns out that a~single update of $H$ causes only a~constant number of changes to such a~representation.
The technique presented here is not new: it was previously used by Alman et al.~\cite{AlmanMW20} to maintain a~feedback vertex set of size~$k$ in dynamic graphs.
However, due to the fact that our goal is to monitor any \msotwo-definable property, we require stronger invariants to hold,
and consequently, our data structure needs to be more careful in the process of updating its inner state.


Let us clarify how we will represent the fern decomposition of a~multigraph~$H$.
For a~graph~$G$ and a~surjection $\name : V(G) \to V'$, denote by $\name(G)$ a~multigraph on the vertex set $V'$ with edges of the form $\{ \name(u)\name(v) \mid uv \in E(G) \}$.
A~boundaried tree $(T, \bnd T)$ together with a~mapping $\name: V(T) \to V(S)$ \emph{represents a~fern} $(S, \bnd S)$ if $\name(T) = S$, and additionally:
\begin{itemize}
  \item if $S$ is a~tree, then $\name$ is a~bijection, and $\bnd T = \name^{-1}(\bnd S)$;
  \item if $S$ is a~unicyclic graph, then $\bnd T = \{x, x'\}$, where $x'$ is a~leaf of~$T$, $\name(x) = \name(x') = u$ where $u$ is some vertex of $S$, and $\name|_{V(T) \setminus \{ x' \}}$ is a~bijection.
    Moreover,  if $|\bnd S| = 1$ then $u \in \bnd S$, and otherwise $u$ is any~vertex on the unique cycle in $S$. Intuitively, $T$ is obtained from $S$ by \emph{splitting} one of the vertices on the unique cycle in $S$ (see Figure~\ref{fig:unicyclic_rep}).
\end{itemize}
If the mapping $\name$ is clear from the context, we may say that $T$ represents $S$.

\begin{figure}[ht]
\captionsetup[subfigure]{labelformat=empty}
\centering
\begin{subfigure}{0.45\textwidth}
\centering
{
  \newcommand{\Ang}{90}
\newcommand{\AngSmall}{8}
\newcommand{\distN}{2}
\newcommand{\distn}{3.3}
\begin{tikzpicture}[scale=0.6]
  \node[vertex] (1) at (\Ang:\distN) [label=right:$u_1$]{};
  \node[vertex] (2) at (\Ang+72:\distN) [label=right:$u_2$]{};
  \node[vertex] (3) at (\Ang+2*72:\distN) [label=left:$u_3$]{};
  \node[vertex] (4) at (\Ang+3*72:\distN) [label=right:$u_4$]{};
  \node[vertex] (5) at (\Ang+4*72:\distN) [label=left:$u_5$]{};

  \coordinate (6) at (\Ang+\AngSmall:\distn) {};
  \coordinate (7) at (\Ang-\AngSmall:\distn) {};
  \coordinate (8) at (\Ang+\AngSmall+72:\distn) {};
  \coordinate (9) at (\Ang-\AngSmall+72:\distn) {};
  \coordinate (10) at (\Ang+\AngSmall+2*72:\distn) {};
  \coordinate (11) at (\Ang-\AngSmall+2*72:\distn) {};
  \coordinate (12) at (\Ang+\AngSmall+3*72:\distn) {};
  \coordinate (13) at (\Ang-\AngSmall+3*72:\distn) {};
  \coordinate (14) at (\Ang+\AngSmall+4*72:\distn) {};
  \coordinate (15) at (\Ang-\AngSmall+4*72:\distn) {};

  \draw (1) -- (2)
    (2) -- (3)
    (3) -- (4)
    (4) -- (5)
    (5) -- (1)

    (1) -- (6)
    (1) -- (7)
    (6) -- (7)

    (2) -- (8)
    (2) -- (9)
    (8) -- (9)

    (3) -- (10)
    (3) -- (11)
    (10) -- (11)

    (4) -- (12)
    (4) -- (13)
    (12) -- (13)

    (5) -- (14)
    (5) -- (15)
    (14) -- (15);
\end{tikzpicture}
}
\caption{$S$}
\end{subfigure}
\begin{subfigure}{0.45\textwidth}
\centering
{
  \newcommand{\Ang}{120}
\newcommand{\AngSmall}{8}
\newcommand{\distN}{2}
\newcommand{\distn}{3.3}
\begin{tikzpicture}[scale=0.6]
  \node[bndvertex] (1) at (\Ang:\distN) [label=left:$x_1$]{};
  \node[vertex] (2) at (\Ang+60:\distN) [label=right:$x_2$]{};
  \node[vertex] (3) at (\Ang+2*60:\distN) [label=left:$x_3$]{};
  \node[vertex] (4) at (\Ang+3*60:\distN) [label=right:$x_4$]{};
  \node[vertex] (5) at (\Ang+4*60:\distN) [label=left:$x_5$]{};
  \node[bndvertex] (0) at (\Ang+5*60:\distN) [label=right:$x_1'$]{};

  \coordinate (6) at (\Ang+\AngSmall:\distn) {};
  \coordinate (7) at (\Ang-\AngSmall:\distn) {};
  \coordinate (8) at (\Ang+\AngSmall+60:\distn) {};
  \coordinate (9) at (\Ang-\AngSmall+60:\distn) {};
  \coordinate (10) at (\Ang+\AngSmall+2*60:\distn) {};
  \coordinate (11) at (\Ang-\AngSmall+2*60:\distn) {};
  \coordinate (12) at (\Ang+\AngSmall+3*60:\distn) {};
  \coordinate (13) at (\Ang-\AngSmall+3*60:\distn) {};
  \coordinate (14) at (\Ang+\AngSmall+4*60:\distn) {};
  \coordinate (15) at (\Ang-\AngSmall+4*60:\distn) {};

  \draw (1) -- (2)
    (2) -- (3)
    (3) -- (4)
    (4) -- (5)
    (5) -- (0)

    (1) -- (6)
    (1) -- (7)
    (6) -- (7)

    (2) -- (8)
    (2) -- (9)
    (8) -- (9)

    (3) -- (10)
    (3) -- (11)
    (10) -- (11)

    (4) -- (12)
    (4) -- (13)
    (12) -- (13)

    (5) -- (14)
    (5) -- (15)
    (14) -- (15);
\end{tikzpicture}
}
\caption{$T$}
\end{subfigure}
\caption{A~cyclic fern~$S$ with empty boundary and a~boundaried tree~$T$ representing~$S$. Here, we have $\bnd T = \{ x_1, x_1' \}$ and we set $\name(x_i) = u_i$ and $\name(x_1') = u_1$ (and the corresponding subtrees of $S$ and $T$ are mapped to each other). Note that, by definition, at least one of the vertices $x_1$ and $x_1'$ must be a~leaf of~$T$.}
\Description{Transforming a cyclic fern into a boundaried tree.}
\label{fig:unicyclic_rep}
\end{figure}

We say that a~forest $\Tc$ of boundaried trees together with a~mapping $\name\colon V(\Tc) \to V(H)$ \emph{represents the fern decomposition $H[\Fc]$} of a~multigraph $H$ if every fern from $H[\Fc]$ is represented by a~different boundaried tree $(T, \bnd T)$ in~$\Tc$ together with a~mapping $\name|_{V(T)}$.
Again, if the mapping $\name$ is clear from the context, we may say that $\Tc$ represents $H[\Fc]$.

Given a~dynamic multigraph~$H$, we are going to maintain a~forest $\Delta_\Tc$ of top trees and a~dynamic dictionary $\name \colon V(\Tc) \to V(H)$ such that $\Tc$ with $\name$ represents the fern decomposition of~$H$.
Our data structure will perform the updates on the pair $(\Delta_\Tc, \name)$, where whenever it adds or removes a~vertex~$x$ of~$\Tc$, it immediately updates the value of $\name(x)$.
We assign to each top tree $\Delta_T \in \Delta_\Tc$ one of the two auxiliary states: either $\Delta_T$ is \emph{attached} (which intuitively means that it already represents some~fern $S \in H[\Fc]$), or it is \emph{detached} (which intuitively means that it is being modified).

Now, we describe how we will represent the quotient multigraph~$\quo{H}{\Fc}$.
Recall from Lemma~\ref{lem:fern-decomposition} that the vertices of $\quo{H}{\Fc}$ correspond to the essential vertices of $H$, and each edge of $\quo{H}{\Fc}$ comes from some~fern in the decomposition $H[\Fc]$.
Therefore, we can represent the vertices of $\quo{H}{\Fc}$ by a~subset of $V(H)$, and to each edge $f \in E(\quo{H}{\Fc})$ we can assign a~top tree $\Delta_{T_f} \in \Delta_\Tc$ representing the corresponding fern in $H$.
We store vertices and edges of~$\quo{H}{\Fc}$ in dynamic dictionaries so that we can access and modify them in time $O(\log |H|)$.
Furthermore, to each attached top tree~$\Delta_T \in \Delta_\Tc$ we assign a~tuple $\tup{a}_T$ of vertices of $\quo{H}{\Fc}$ ($|\tup{a}_T| \leq 2$), called hereinafter an~\emph{attachment tuple of $T$}, so that:
\begin{itemize}
  \item if $\tup{a}_T = (u, v)$, then $T$ represents a~fern which corresponds to an~edge $uv$ in $E(\quo{H}{\Fc})$;
  \item if $\tup{a}_T = (v)$, then $T$ represents a~tree fern with boundary $\{ v \}$;
  \item if $\tup{a}_T = \varepsilon$, then $T$ represents a~fern with empty boundary.
\end{itemize}
We store this assignment in a~dynamic dictionary $\attachment \colon \Trees(\Delta_\Tc) \to V(\quo{H}{\Fc})^{\leq 2}$.

We consider a~restricted model of computation, where only the following procedures can be called to modify the multigraph $\quo{H}{\Fc}$:
\begin{itemize}
  \item $\addquo(u)$: adds an~isolated vertex~$u$ to $V(\quo{H}{\Fc})$;
  \item $\delquo(u)$: removes an~isolated vertex~$u$ from $V(\quo{H}{\Fc})$;
  \item $\attach(\tup{a}, \Delta_T)$: takes a~tuple $\tup{a}$ of vertices of $\quo{H}{\Fc}$ ($|\tup{a}| \leq 2$) and a~detached top tree $\Delta_T \in \Delta_\Tc$ and marks $\Delta_T$ as attached by setting $\attachment(\Delta_T) \coloneqq \tup{a}$. 
    
    Moreover, if $\tup{a} = (u, v)$, adds an~edge $uv$ to $E(\quo{H}{\Fc})$ (if $u = v$, it adds a~self-loop $uu$), otherwise it does not modify $\quo{H}{\Fc}$.
  \item $\detach(\tup{a}, \Delta_T)$: analogous to $\attach(\tup{a}, \Delta_T)$, but takes an~attached tree $\Delta_T$ instead, and marks $\Delta_T$ as detached by temporarily removing $\Delta_T$ from the domain of $\attachment$.

    Moreover, instead of adding an~edge to $E(\quo{H}{\Fc})$, we remove it.
\end{itemize}
The purpose of introducing such restrictions will become more clear in Section~\ref{ssec:lemmaA_dynamic_contractions}, where we will augment the structure of top trees with information about \msotwo types.
Having a~compact description of possible modifications of $\quo{H}{\Fc}$ will make the~arguments there simpler, because we will only  need to argue how to maintain types under these modifications.

We are ready to formulate the~main result of this section.

\begin{lemma}
  \label{lem:dynamic_ferns}
  Let $H$ be a~dynamic multigraph, initially empty, where we are allowed to modify (i.e., add or remove) edges and isolated vertices.
  Then, there exists a~data structure $\F$ which maintains:
  \begin{itemize}
    \item a~forest~$\Delta_\Tc$ of top trees of $\Tc$, where $\Tc$ is a~forest of boundaried trees, and a~dynamic dictionary $\name\colon V(\Tc) \to V(H)$ such that $\Tc$ together with $\name$ represents the fern decomposition of $H$;
    \item the~quotient multigraph $\quo{H}{\Fc}$ with a~dynamic dictionary $\attachment: \Trees(\Delta_\Tc) \to V(\quo{H}{\Fc})^{\leq 2}$ defined as described above;
    \item dynamic dictionaries $\name^{-1} : V(H) \to 2^{V(\Tc)}$ and $\attachment^{-1} : V(\quo{H}{\Fc})^{\leq 2} \to 2^{\Trees(\Delta_\Tc)}$, storing inverse functions of~$\name$ and $\attachment$, respectively (the values of $\name^{-1}(u)$ and $\attachment^{-1}(\tup{a})$ are stored in dynamic sets);
    \item a~dynamic dictionary $\edgerep : V(H)^2 \to 2^{V(\Tc)^2}$, that given a~pair $(u, v)$ of vertices of $H$, returns a~dynamic set comprising all pairs $(x, y) \in V(\Tc)^2$ such that $xy \in E(\Tc)$, $\name(x) = u$, and $\name(y) = v$.
  \end{itemize}
  Moreover, the following additional invariants hold.
  \begin{enumerate}[label=(\Alph*)]
    \item \label{inv:all_attached} After performing an~update of~$H$, every top tree $\Delta_T \in \Delta_\Tc$ must be attached. 
    \item \label{inv:modify_detached_tree} Whenever an~operation on a~top tree $\Delta_T \in \Delta_\Tc$ is performed by $\F$, $\Delta_T$ must be detached.
    \item \label{inv:modify_detached_vertex} Whenever $\addquo(u)$ or $\delquo(u)$ is called by $\F$, every top tree $\Delta_T \in \Delta_\Tc$ containing a~vertex $x$ such that $\name(x) = u$ must be detached.
  \end{enumerate}
  $\F$ handles each update of $H$ in time $\Oh{\log |H|}$.
  Additionally, each update of~$H$ requires $\Oh{1}$ operations on $(\Delta_\Tc, \name)$, $\Oh{1}$ calls to $\addquo$, $\delquo$, $\attach$, $\detach$, and $\Oh{1}$ modifications of $\name^{-1}$, $\attachment^{-1}$, and $\edgerep$.
\end{lemma}

\begin{proof}
  We did not mention mappings $\name^{-1}$, $\attachment^{-1}$ and $\edgerep$ before, since they are introduced mainly for technical reasons.
  Let us observe that each operation on $(\Delta_\Tc, \name)$ and $\quo{H}{\Fc}$ naturally induces operations on $\name^{-1}$, $\attachment^{-1}$ and $\edgerep$, hence we can omit the updates of those structures in what follows.
  For example, we will assume that whenever we link/cut vertices $x$ and $y$ in $\Delta_\Tc$, we update the values of $\edgerep(u, v)$, where $\{u, v\} = \{\name(x), \name(y)\}$.

  We begin with the following auxiliary facts.

  \begin{claim}
    \label{cl:df_edge_fern}
    Let $u$, $v$ be essential vertices of a~multigraph $H$ (possibly $u = v$).
    Then, each edge $e \in E(H)$ with endpoints $u$ and $v$ forms a~separate fern $S_e$ in the fern decomposition $H[\Fc]$, and $S_e$~contributes an~edge $uv$ to $E(\quo{H}{\Fc})$.
  \end{claim}
  \begin{proof}
    Observe that $e$ is essential in $H$, for it is either a~self-loop (and thus it lies on a~cycle), or $u \neq v$, and then $e$ is essential by Lemma~\ref{lem:vertex-connection-essential}.
    Hence, incidences $(u, e)$ and $(v, e)$ are both critical, and consequently, $e$ cannot be in $\sim$-relation with any of its adjacent edges.
    Therefore, $e$ forms a~separate fern $S_e$ in $H[\Fc]$, and by definition of $\quo{H}{\Fc}$, $S_e$ contributes an~edge $uv$ to~$\quo{H}{\Fc}$ (a self-loop if $u = v$).
  \cqed\end{proof}

  \begin{claim}
    \label{cl:df_essential}
    Let $H'$ be a~multigraph obtained from a~multigraph $H$ by adding or removing an~edge $e$ between two vertices $u,\ v$ which are essential both in~$H$ and in~$H'$.
    Then, each element (i.e.,~vertex or edge) of~$H$ is essential in~$H$ if and only if it is essential in~$H'$.

    Moreover, the quotient multigraph $\quo{H'}{\Fc'}$ can be obtained from $\quo{H}{\Fc}$ by adding (respectively, removing) an~edge between $u$ and $v$ to it.
  \end{claim}
  \begin{proof}
    For the first part, it is enough to show that every edge of~$H$ is non-essential in~$H$ if and only if it is non-essential in~$H'$.
    Indeed, having proved it, we obtain that vertices $u$ and $v$ are essential both in~$H$ and in~$H'$, and for the other vertices of~$H$ their number of incidences to essential edges does not change.
    Recall that an~edge $f$ is non-essential in~$H$ if it is a~bridge, and removing it produces at least one new component $C$ which is acyclic.
    Clearly, all edges and vertices of $C$ are non-essential in~$H$.
    In particular, this means that $u, v \not\in V(C)$.
    Therefore, adding or removing $e$ does not affect $C$, and thus $f$ is a~non-essential edge in~$H'$ as well.
    The proof that if $f$ is non-essential in~$H'$, then it is non-essential in~$H$, can be obtained by swapping $H$ with $H'$ in the argumentation above.

    By combining this result with Claim~\ref{cl:df_edge_fern}, we obtain that $\Fc' = \Fc \cup \{\{ e \}\}$, and thus adding (resp. removing) an~edge $uv$ to $\quo{H}{\Fc}$ yields $\quo{H'}{\Fc'}$.
  \cqed\end{proof}

  We claim that it is enough to show how to implement the updates of~$H$ of the form: $\mathsf{addEdge}(u, u)$ and $\mathsf{delEdge}(u, u)$, which add and remove a~self-loop at~$u$, respectively.
  Given such two operations, we can implement the remaining updates of~$H$ as follows.
  
  To introduce a~new vertex~$u$ to~$H$, add a~new vertex~$x$ to~$\Delta_{\Tc}$ with $\name(x) = u$.
  To remove an~isolated vertex~$u$ of $H$, delete the~unique vertex $x \in \name^{-1}(u)$ from $\Delta_\Tc$ and remove $x$ from the domain of $\name$.

  Now, consider an~update of an~edge $e$ between two different vertices $u$ and~$v$ ($e$ is either inserted or discarded).
  First, we add two self-loops at the vertex~$u$ and two self-loops at the vertex~$v$.
  Let $H_1$ be the~newly created multigraph, and let $H_1'$ be obtained from $H_1$ by updating edge $e$ appropriately.
  Recall that every self-loop counts as two incidences to essential edges, and thus vertices $u$ and $v$ are essential both in $H_1$ and in~$H_1'$.
  Hence, by combining Claim~\ref{cl:df_edge_fern} with Claim~\ref{cl:df_essential}, we conclude that the partition $\Fc_1'$ of edges of $H_1'$ satisfies $\Fc_1' = \Fc_1 \cup \{ \{e\} \}$ (resp. $\Fc_1' = \Fc_1 \setminus \{ \{e\} \}$), and the only difference between multigraphs $\quo{H_1}{\Fc_1}$ and $\quo{H_1'}{\Fc_1'}$ is a~presence of an~additional edge between $u$ and $v$ in one of them.
  
  Therefore, to insert $e$ to $H_1$, add two new vertices $x$ and $y$ to $\Delta_\Tc$ with $\name(x) = u$ and $\name(y) = v$, link $x$ with $y$, and call $\attach( (u,v), \Delta_{T_{uv}})$, where $\Delta_{T_{uv}}$ is a~top tree on the edge $xy$.
  In case $e$ is to be deleted, we proceed in a~similar way, but we detach a~tree $\Delta_{T_{uv}}$, and remove it from $\Delta_\Tc$ (here $\Delta_{T_{uv}}$ is any~top tree on a~single edge $xy$, where $(x, y) \in \edgerep(u, v)$).

  Finally, it remains to delete four loops that we added at the beginning.

  \bigskip

  We proceed to implementations of $\mathsf{addEdge}(u, u)$ and $\mathsf{delEdge}(u, u)$.
  The following fact will be helpful.

  \begin{claim}
    \label{cl:df_path_loop_essential}
    Let $\ell$ be a~self-loop at a~vertex~$u$ in a~multigraph~$H$.
    Consider a~simple path $P$ in~$H$ that starts at $u$ and ends at~$v$, where $v$ is an~essential vertex of~$H$.
    Then every edge traversed by~$P$ is essential in $H$, and $u$ is an~essential vertex of~$H$.
  \end{claim}
  \begin{proof}
    Since $v$ participates in at least three critical incidences, there is an~essential edge $e$ incident to $v$ which does not lie on $P$.
    Extending $P$ with $\ell$ and $e$ yields a~walk $W$ in~$H$ that starts and ends with an~essential edge, and traverses every edge at most once.
    Hence, by Lemma~\ref{lem:connection-essential}, every edge of $W$ (in particular, every edge of $P$) is essential.

    If $u = v$, then $u$ is essential by definition of $v$.
    If $u \neq v$, then $u$ participates in at least 3 critical incidences: two with the loop~$\ell$ and one with the first edge of~$P$.
    Consequently, $u$ is essential in~$H$.
  \cqed\end{proof}

  \paragraph*{Adding a self-loop $\ell$ at a~vertex~$u$.}
    Let $H^{\Add}$ be a~multigraph obtained by adding $\ell$ to~$H$, and let $\Fc^{\Add}$ be the~partition of its edges into equivalence classes of the relation~$\sim$.

    First, suppose that $u$ is an~essential vertex of~$H$
    (we can verify this by checking whether $u$ is a~vertex of $\quo{H}{\Fc}$).
    Then by Claim~\ref{cl:df_edge_fern}, $\ell$ contributes a~self-loop $\ell'$ at~$u$ to $E(\quo{H^{\Add}}{\Fc^{\Add}})$.
    By Claim~\ref{cl:df_essential} this is the only difference between $\quo{H}{\Fc}$ and $\quo{H^{\Add}}{\Fc^{\Add}}$.
    Therefore, it is enough to initialize a~new top tree on a~single new edge that represents~$\ell$, and attach this top tree to the~tuple $\tup{a} = (u, u)$.

    \medskip
    
      Now, we assume that $u$ is non-essential in $H$.
      Then, there is a~unique fern $S_u \in H[\Fc]$ that contains~$u$.
      Clearly, after adding $\ell$ to $H$, any essential edge of $H$ stays essential in $H^{\Add}$, and consequently, any essential vertex of $H$ stays essential in $H^{\Add}$.
      For the non-essential elements of $H$ we use the following fact.

      \begin{claim}
        \label{cl:df_only_fern}
        If $t \in H$ is a feature (i.e.\ an~edge or a~vertex) which is non-essential in~$H$ and essential in~$H^{\Add}$, then $t \in S_u$.
      \end{claim}

      \begin{proof}
        Take an~edge $e$ that becomes essential in $H^{\Add}$.
        Since $e$ is non-essential in~$H$, it is a bridge in~$H$, and removing it from $H$ creates two new components, one of which (call it $C$) is acyclic.
        On the other hand, $e$ is essential in~$H^{\Add}$, which implies that the component $C$ must be affected by adding the loop~$\ell$.
        Hence, $u \in V(C)$.
        Since $C$ is acyclic in~$H$, all edges of $C$ are non-essential in $H$.
        Moreover, $e$ is incident to $C$, and thus $e \sim f$ for every $f \in E(C)$, where $\sim$ is considered in $H$.
        Consequently, $e$ is in the same fern of $H[\Fc]$ as $u$, that is, $e \in E(S_u)$.

        For the case of vertices, observe that a~vertex~$v$ which becomes essential in $H^{\Add}$ either is equal to $u$, or it is incident to some~edge $e$ of~$H$ that becomes essential in $H^{\Add}$ (hence $e \in E(S_u)$).
        In both cases we conclude that $v \in V(S_u)$.
      \cqed\end{proof}

      Hence, we can focus on how $\Fc^{\Add}$ partitions the edges of $S_u$, and how to update $\Delta_\Tc$ and $\quo{H}{\Fc}$ accordingly.

      In what follows, we will use an~auxiliary procedure $\mathsf{\makeess}(x, Y)$, which takes:
      \begin{itemize}
        \item a~vertex $x \in V(\Tc)$ such that the top tree containing $x$ is detached, and
        \item a~subset of vertices $Y \subseteq V(\Tc)$ such that for every $y \in Y$, $y \neq x$, $x$~and~$y$ are in the same boundaried tree~$T$ of~$\Tc$, and all the paths in $T$ of the form $x \Path y$ (for $y \in Y$) are pairwise edge-disjoint.
      \end{itemize}
      Intuitively, $\makeess(x, Y)$ adds a~new essential vertex $\name(x)$ to $\quo{H}{\Fc}$ with edges of the form $\{\name(x)\name(y) \mid y \in Y, \name(y)\in V(\quo{H}{\Fc}) \}$, and splits $T$ accordingly (see Figure~\ref{fig:make_essential}).
      Formally, it performs the following operations:
      \begin{itemize}
        \item Introduce $\name(x)$ to $\quo{H}{\Fc}$ by calling $\addquo(\name(x) )$, and for every $y \in Y$ do as follows:
          \begin{itemize}
            \item find the second vertex $z_y$ on the path $x \Path y$ in $T$ ($z_y = \mathsf{jump}(x, y, 1)$), and apply on $\Delta_\Tc$:
              \begin{itemize}
                \item $\mathsf{cut}(x, z_y)$,
                \item $\mathsf{add}(x_y)$, where $x_y$ is a~new vertex with $\name(x_y) = \name(x)$, and
                \item $\mathsf{link}(x_y, z_y)$;
              \end{itemize}
            \item let $T_y$ be the boundaried tree that now contains $z_y$;
            \item set the boundary of $T_y$ to $\{ y, x_y \} $, and if $\name(y) \in V(\quo{H}{\Fc})$, attach~$\Delta_{T_y}$ to a~tuple $(\name(y), \name(x) )$.
          \end{itemize}
        \item After performing these operations, if $x$ is an~isolated vertex in~$\Tc$, remove it from $\Tc$.
        Otherwise, let $T_x$ be the~unique boundaried tree containing~$x$.
        Set the boundary of $T_x$ to $\{ x \}$ and attach $\Delta_{T_x}$ to a~tuple $(\name(x))$.
      \end{itemize}
    
\begin{figure}[ht]
    \centering
    \begin{subfigure}{0.45\linewidth}
    \centering
    {
      \newcommand{\Ang}{-30}
\newcommand{\AngSmall}{16}
\newcommand{\distN}{2.5}
\newcommand{\TreeSmall}{0.5}
\newcommand{\TreeLarge}{1.1}
\begin{tikzpicture}
  \node [vertex] (x) at (0, 0) [label=above right:$x$] {};
  \node [vertex] (y1) at (\Ang:\distN) [label=right:$y_1$] {};
  \node [vertex] (y2) at (\Ang+120:\distN) [label=right:$y_2$] {};
  \node [vertex] (y3) at (\Ang+240:\distN) [label=left:$y_3$] {};
  
  \foreach \id/\pa in {a/0.333, b/0.667} {
    \foreach \b in {1, 2, 3} {
      \node [vertex] (y\id\b) at ($(x)!\pa!(y\b)$) {};
    };
  };
  \draw (x) -- (ya1) -- (yb1) -- (y1);
  \draw (x) -- (ya2) -- (yb2) -- (y2);
  \draw (x) -- (ya3) -- (yb3) -- (y3);
  
  \foreach \name/\size/\angle in {x/\TreeLarge/-90,
          ya1/\TreeSmall/-90, yb1/\TreeSmall/-90, y1/\TreeLarge/-90,
          ya2/\TreeSmall/150, yb2/\TreeSmall/150, y2/\TreeLarge/90,
          ya3/\TreeSmall/-90, yb3/\TreeSmall/-90, y3/\TreeLarge/-90} {
    \draw (\name) -- ($(\name) + (\angle-\AngSmall:\size)$) -- ($(\name) + (\angle+\AngSmall:\size)$) -- (\name);
  };
\end{tikzpicture}
    }
    \caption{}
    \label{fig:make_essential_1}
    \end{subfigure}
    \hfill
    \begin{subfigure}{0.45\linewidth}
    \centering
    {
      \newcommand{\Ang}{-30}
\newcommand{\AngSmall}{16}
\newcommand{\distN}{3}
\newcommand{\distn}{1.2}
\newcommand{\TreeSmall}{0.5}
\newcommand{\TreeLarge}{1.1}
\begin{tikzpicture}
  \node [vertex, minimum size=4.5pt] (x) at (0, 0) [label=above right:$x$] {};
  \node [vertex, minimum size=4.5pt] (y1) at (\Ang:\distN) [label=right:$y_1$] {};
  \node [vertex, minimum size=4.5pt] (y2) at (\Ang+120:\distN) [label=right:$y_2$] {};
  \node [vertex, minimum size=4.5pt] (y3) at (\Ang+240:\distN) [label=left:$y_3$] {};
  
  \node [vertex, minimum size=4.5pt] (x1) at (\Ang+0:\distn) [label=above:$x_{y_1}$]{};
  \node [vertex, minimum size=4.5pt] (x2) at (\Ang+120:\distn) [label=right:$x_{y_2}$]{};
  \node [vertex, minimum size=4.5pt] (x3) at (\Ang+240:\distn) [label=above:$x_{y_3}$]{};
  
  \foreach \id/\pa in {a/0.333, b/0.667} {
    \foreach \b in {1, 2, 3} {
      \node [vertex] (y\id\b) at ($(x\b)!\pa!(y\b)$) {};
    };
  };
  \draw (x1) -- (ya1) -- (yb1) -- (y1);
  \draw (x2) -- (ya2) -- (yb2) -- (y2);
  \draw (x3) -- (ya3) -- (yb3) -- (y3);
  
  \foreach \name/\size/\angle in {x/\TreeLarge/-90,
          ya1/\TreeSmall/-90, yb1/\TreeSmall/-90, y1/\TreeLarge/-90,
          ya2/\TreeSmall/150, yb2/\TreeSmall/150, y2/\TreeLarge/90,
          ya3/\TreeSmall/-90, yb3/\TreeSmall/-90, y3/\TreeLarge/-90} {
    \draw (\name) -- ($(\name) + (\angle-\AngSmall:\size)$) -- ($(\name) + (\angle+\AngSmall:\size)$) -- (\name);
  };
\end{tikzpicture}
    }
    \caption{}
    \label{fig:make_essential_2}
    \end{subfigure}
    \caption{The effect of applying $\makeess(x, y_1, y_2, y_3)$.}
    \Description{An example application of makeEssential, separating $x$ from the paths connecting $x$ to $y_1$, $y_2$, and $y_3$.}
  \label{fig:make_essential}
\end{figure}

\medskip
    
    We are ready to show how to update our inner structures provided that $u$ is non-essential in~$H$.
      Let $T_u \in \Tc$ be the~boundaried tree representing $S_u$.
      Note that $|\name^{-1}(u)| \in \{1, 2\}$, where $|\name^{-1}(u)| = 2$ if and only if $S_u$ is a~unicyclic fern with empty boundary and $T_u$ is obtained from $S_u$ by splitting $u$.
      Hence let $x_u \in \name^{-1}(u)$, and if $|\name^{-1}(u)| = 2$, assume that $x_u$ is chosen so that the element of $\name^{-1}(u)$ different from $x_u$ is a~leaf of~$T_u$.
%
      Recall that we can find $\Delta_{T_u}$ and $\bnd T_u$ by calling $\mathsf{get}(x_u)$ on $\Delta_\Tc$.
      We start with adding an~edge to $T_u$ representing the loop~$\ell$.
      To do this, we detach $\Delta_{T_u}$, add a~new vertex~$x_u'$ to~$\Delta_\Tc$ with $\name(x_u') = u$, and link $x_u$ with $x_u'$ in $\Delta_{T_u}$.
      However, due to the fact that some vertices of $S_u$ may become essential in~$H^{\Add}$, we will need to fix this representation.

      We consider different cases depending on the size of $\bnd S_u$.
      
      \mpara{$\bnd S_u$ is empty.}
      First, assume that $\bnd S_u = \emptyset$.
      This means that $S_u$ is equal to one of the connected components of~$H$.

      If $S_u$ is a~tree fern (or, an isolated vertex), then, after adding the loop~$\ell$, the component on $V(S_u)$ becomes a~unicyclic graph.
      Hence, all edges of $H^{\Add}$ between the vertices of $S_u$ still form a~single fern in $H^{\Add}[\Fc^{\Add}]$.
      Recall that we have already added the leaf $x_u'$ to $T_u$, so that $x_u x_u'$ represents~$\ell$.
      Set the boundary of $T_u$ to $\{ x_u, x_u' \}$.
      Then, the current forest $\Tc$ indeed represents the fern decomposition of $H^{\Add}$, and it remains to attach $\Delta_{T_u}$ to the empty tuple.
      Note that we can detect this case by verifying whether $\bnd T_u = \emptyset$ holds.

      If $S_u$ is a~unicyclic graph (see Figure~\ref{fig:df_add_loop_1}), let $u \Path w$ be the shortest path in $S_u$ such that $w$~lies on the unique cycle of~$S_u$.
      We detect this case by checking that $\tup{a}_u = \varepsilon$ and $|\bnd T_u| = 2$.
      Then, a~vertex $x_w$ such that $\name(x_w) = w$ can be found by calling $\textsf{meet}(x_u, y, y')$, where $\{ y, y' \} \coloneqq \bnd T_u$.
      
      Recall that to obtain $T_u$ from $S_u$ we can ``split'' $S_u$ with respect to any vertex on its unique cycle.
      Hence, we may assume that $x_w \in \bnd T_u$, say $y = x_w$ and $y' = x_w'$.
      Indeed, if $x_w \not\in \bnd T_u$, then we can reconstruct $T_u$ as follows.
      Without loss of generality, assume that $y'$ is a~leaf of $T_u$. (For each boundaried tree representing cyclic fern of empty boundary we may store information which element of its boundary is a~leaf.)
      Let $z_1$ be the second vertex on the path $y' \Path x_w$ in $T_u$, and let $z_2$ be the second vertex on the path $x_w \Path y$ in $T_u$ ($z_1$ and $z_2$ can be found by calling $\textsf{jump}(y', x_w, 1)$ and $\textsf{jump}(x_w, y, 1)$, respectively).
      Then, we apply the following operations on $\Delta_\Tc$: $\textsf{cut}(z_1, y')$, $\textsf{del}(y')$ (and remove $y'$ from the domain of $\name$), $\textsf{cut}(x_w, z_2)$, $\textsf{link}(z_1, y)$, $\textsf{add}(x_w')$ (where $\name(x_w') = w$), $\textsf{link}(x_w', z_2)$, $\textsf{expose}(\{x_w, x_w'\})$.
      After performing all these modifications, we see that we rearranged the split of the cycle of $S_u$, and now $\bnd T_u = \{ x_w, x_w' \}$, as desired.

      We consider two cases.
      \begin{itemize}
        \item If $w = u$, then we see that the only vertex that becomes essential in $H^{\Add}$ is $w = u$, and we need to add it to $\quo{H}{\Fc}$ with two self-loops of the form $ww$ (one is contributed by $S_u$, and another one by the trivial fern comprising $\ell$).
          One can observe that in order to modify $\quo{H}{\Fc}$ and $\Delta_{\Tc}$ appropriately, it is enough to call $\makeess(x_w, \{x_w', x_u'\})$.
        \item If $w \neq u$, then both vertices $w$ and $u$ become essential in $H^{\Add}$, and we need to introduce them to $\quo{H}{\Fc}$ with edges $uu$, $uw$, $ww$.
          Again, this can be done by calling $\makeess(x_u, \{ x_u', x_w \})$ and $\makeess(x_w, \{ x_u, x_w'\})$.
          Observe that only the second operation adds an~edge $uw$ to $\quo{H}{\Fc}$ as during the first one the vertex $w$ is not essential yet.
      \end{itemize}

    \begin{figure}[ht]
    \centering
    \begin{subfigure}{0.3\linewidth}
    \centering
    {
      \newcommand{\Ang}{-90}
\newcommand{\AngSmall}{8}
\newcommand{\distN}{1}
\newcommand{\distn}{3.3}
\begin{tikzpicture}
  \node[essvertex] (1) at (\Ang:0.8*\distN) [label=above:$w$]{};
  \node[vertex] (2) at (\Ang+72:0.8*\distN) {};
  \node[vertex] (3) at (\Ang+2*72:0.8*\distN) {};
  \node[vertex] (4) at (\Ang+3*72:0.8*\distN) {};
  \node[vertex] (5) at (\Ang+4*72:0.8*\distN) {};
  \node[vertex] (6) [below=\distN of 1] {};
  \node[essvertex] (7) [below=\distN of 6] [label=right:$u$]{};

  \draw (1) -- (2) -- (3) -- (4) -- (5) -- (1);
  \draw[essedge] (1) -- (6) -- (7);
  \path (7) edge[essedge, loop, in=-45, out=-135, looseness=10] (7);
\end{tikzpicture}
    }
    \caption{}
    \label{fig:df_add_loop_1}
    \end{subfigure}
    \hfill
    \begin{subfigure}{0.3\linewidth}
    \centering
    {
      \newcommand{\distN}{1}
\begin{tikzpicture}
  \node[bndvertex] (1) at (0, 0) [label=above:$v$]{};
  \node[vertex] (2) [below=\distN of 1] {};
  \node[vertex] (3) [below=\distN of 2] {};
  \node[essvertex] (4) [below=\distN of 3] [label=right:$u$]{};

  \draw[essedge] (1) -- (2) -- (3) -- (4);
  \path (4) edge[essedge, loop, in=-45, out=-135, looseness=10] (4);
\end{tikzpicture}
    }
    \caption{}
    \label{fig:df_add_loop_2}
    \end{subfigure}
    \hfill
    \begin{subfigure}{0.3\linewidth}
    \centering
    {
      \newcommand{\distn}{1}
\begin{tikzpicture}
  \node[bndvertex] (1) at (0, 0) [label=above:$v_1$]{};
  \node[essvertex] (2) [right=\distn of 1] [label=above:$w$]{};
  \node[vertex] (3) [right=\distn of 2] {};
  \node[bndvertex] (4) [right=\distn of 3] [label=above:$v_2$]{};
  \node[vertex] (5) [below=\distn of 2] {};
  \node[vertex] (6) [below=\distn of 5] {};
  \node[essvertex] (7) [below=\distn of 6] [label=right:$u$]{};

  \draw[essedge] (2) -- (5) -- (6) -- (7);
  \draw (1) -- (2) -- (3) -- (4);
  \path (7) edge[essedge, loop, in=-45, out=-135, looseness=10] (7);
\end{tikzpicture}
    }
    \caption{}
    \label{fig:df_add_loop_3}
    \end{subfigure}
    \caption[]{Adding a~loop at a~vertex $u$ to a~fern $S_u$ which is: (a) a~unicyclic graph of empty boundary, (b) a~tree fern of boundary $\bnd S_u = \{ v \}$, (c) a~tree fern of boundary $\bnd S_u = \{ v_1, v_2 \}$.
    Empty vertices \begin{tikzpicture}\node[bndvertex] (1) at (0,0) {};\end{tikzpicture} denote essential vertices of~$H$, and red elements: \{ \begin{tikzpicture}\node[essvertex] (1) at (0,0) {};\end{tikzpicture}, \begin{tikzpicture}\draw[essedge] (0,0) -- (0.3,0.2);\end{tikzpicture} \} denote vertices and edges of $H$ which become essential in~$H^{\Add}$.}
    \Description{Process of adding a loop at a vertex.}
    \label{fig:df_add_loop}
    \end{figure}

    \mpara{$\bnd S_u$ is not empty.} Now, assume $\bnd S_u \neq \emptyset$.
    We consider cases based on the shape of~$S_u$.
    \begin{itemize}
      \item $S_u$ is a~tree fern with $\bnd S_u = \{ v \}$ (see Figure~\ref{fig:df_add_loop_2}).
        This case can be detected by checking whether $\tup{a}_u = (v)$.
        Then $v$ is an~essential vertex of~$H$, and thus $v \neq u$.
        By Claim~\ref{cl:df_path_loop_essential}, all edges on the path $u \Path v$ in~$S_u$ are essential in $H^{\Add}$.
        Observe that other edges of $S_u$ remain non-essential in~$H^{\Add}$ as they still isolate a~subtree of $S_u$.
        Hence, $u$ becomes an~essential vertex in $H^{\Add}$, and we need to add $u$ to $V(\quo{H}{\Fc})$ with edges $uu$ and $uv$, and split $T_u$ accordingly.
        This can be done by calling $\makeess(x_u, \{ x_u', y \} )$, where $y$ is the unique vertex of $\bnd T_u$ ($\name(y) = v$).

      \item $S_u$ is a~tree fern with $\bnd S_u = \{ v_1, v_2 \}$ (see Figure~\ref{fig:df_add_loop_3}).
        We can detect this case by checking whether $\tup{a}_u = (v_1, v_2)$, where $v_1 \neq v_2$.
        Then $v_1$ and $v_2$ are two different essential vertices of~$H$, and thus $u \notin \{v_1, v_2\}$.
        By Claim~\ref{cl:df_path_loop_essential}, all edges on the paths $u \Path v_1$ and $u \Path v_2$ are essential in $H^{\Add}$.
        Again, the other edges of $S_u$ remain non-essential in~$H^{\Add}$ as they still isolate a~subtree of $S_u$.
        Let $w$ be the~intersection point of paths $u \Path v_1$, $u \Path v_2$, and $v_1 \Path v_2$.
        Recall that the corresponding vertex $x_w \in V(T_u)$ can be found be calling $\mathsf{meet}(x_u, y_1, y_2)$, where $\{ y_1, y_2\} \coloneqq \bnd T_u$.
        By definition of a~tree fern with boundary of size~$2$, $v_1$ and $v_2$ are leaves of~$S_u$, and thus $w\notin \{v_1, v_2\}$.

        If $w = u$, then the only vertex of $S_u$ that becomes essential in~$H^{\Add}$ is~$w$.
        Then, we need to add it to $\quo{H}{\Fc}$ with edges $ww$, $wv_1$, and $wv_2$.
        This can be done by calling $\makeess(x_w, \{ x_u', y_1, y_2 \})$.
        
        If $w \neq u$, then both $w$ and $u$ become essential in~$H^{\Add}$, and we need to introduce them to $\quo{H}{\Fc}$ with edges $uu$, $uw$, $wv_1$, and $wv_2$.
        This can be done by calling first $\makeess(x_u, \{ x_u', x_w \})$, and then $\makeess(x_w, \{ x_u, y_1, y_2 \})$.
        Observe that only the second operation will add an~edge $uw$ to $\quo{H}{\Fc}$ as during the first one, $w$ is not yet a~vertex of $\quo{H}{\Fc}$.
      \item $S_u$ is a~unicyclic graph, and $\bnd S_u = \{ v \}$.
        We can detect this case by checking whether $\tup{a}_u = (v, v)$.
        By the definition of a cyclic fern with non-empty boundary, $v$ lies on the cycle of $S_u$, and $v$ has degree~$2$ in~$S_u$, and thus the corresponding vertices $x_{v_1}, x_{v_2} \in \bnd T_u$ are leaves of $T_u$.
        Hence, this case is in fact analogous to the previous one, in the~sense that we may perform the same operations on~$T_u$ as if $T_u$ represented a~tree fern with boundary of size~$2$.
     \end{itemize}

  \paragraph*{Deleting a~self-loop $\ell$ at a~vertex~$u$.}
    Let $H^{\Del}$ be the~multigraph obtained by removing $\ell$ from~$H$, and let $\Fc^{\Del}$ be the~partition of its edges into equivalence classes of the relation~$\sim$.

    If $u$ is a~non-essential vertex of~$H$, then the~unique fern $S_u \in H[\Fc]$ that contains $u$ must be a~unicyclic graph (with cycle $(\ell)$) of empty boundary.
    Indeed, $S_u$ cannot be a~tree, for it contains the~loop $\ell$.
    Moreover, if $\bnd S_u$ was non-empty, say $v \in \bnd S_u$, then $v$ would be essential in $H$ and there would be a~path $u \Path v$ in~$S_u$, hence by Claim~\ref{cl:df_path_loop_essential}, $u$ would be essential in~$H$, a~contradiction.

    Therefore, $S_u$ is a~single connected component of~$H$, and after deleting $\ell$ this component becomes a~tree.
    Let $T_u$ be a~tree representing $S_u$ with $\bnd T_u = \{ x_u, x_u' \}$, where $x_u'$ is a~leaf of~$T_u$.
    Then, it is enough to detach $\Delta_{T_u}$, cut the~edge $x_u x_u'$, remove $x_u'$ from both $\Tc$ and the domain of $\name$, and attach $\Delta_{T_u}$ to the~empty tuple.

    \medskip
    
      From now on, assume that $u$ is essential in $H$, that is, $u$ is a~vertex of $\quo{H}{\Fc}$.
      Clearly, removing an~edge (in particular, a~loop) from $H$ cannot make any non-essential edge essential.
      For the other direction, we will use the claims below, but first let us introduce some additional notation.

      For a~multigraph~$H'$ and a~vertex $w \in V(H')$, we denote by $\dcrit^{H'}(w)$ the number of critical incidences in $H'$ that $w$ participates in.
      In particular, we have $\dcrit^{H'}(w) \geq 3$ if $w$ is essential in $H'$, and $\dcrit^{H'}(w) = 0$ otherwise.
      Depending on details of inner representation of~$\quo{H}{\Fc}$, the non-zero values of $\dcrit^H(\cdot)$ can either be obtained directly in time $\Oh{\log |H|}$, or we can maintain $\dcrit$ as an~additional dictionary $V(\quo{H}{\Fc}) \to \N$, which we update on every modification of~$\quo{H}{\Fc}$.

    \begin{claim}
      \label{cl:df_joint_fern}
      Let $e \in E(H^{\Del})$ be an~edge which is essential in $H$ and non-essential in $H^{\Del}$.
      Let $S_e \in H[\Fc]$ be the fern containing $e$.
      Then $S_e$ is a~tree fern with $|\bnd S_e| = 2$ and $u \in \bnd S_e$.
    \end{claim}
    \begin{proof}
      Suppose that $S_e$ is a~cyclic fern.
      Since $e$ is essential in~$H$, $e$ must lie on the unique cycle of $S_e$.
      Clearly, removing $\ell$ does not affect this cycle, hence $e$ lies on a~cycle in~$H^{\Del}$ as well.
      This means that $e$ is essential in $H^{\Del}$, a~contradiction.
      We have $|\bnd S_e| = 2$, for otherwise $e$ would isolate a~subtree of~$S_e$, and thus it would be non-essential in~$H$.

      Since removing $\ell$ makes $e$ non-essential, we conclude that removing $e$ from $H^{\Del}$ must yield a~new component $C$ of $H^{\Del}$ which is a~tree in~$H^{\Del}$ and contains~$u$.
      Hence, there is a~simple path $P$ of consecutive edges $e_1 = e, e_2, \ldots, e_r$ that starts with $e$ and ends at $u$.
      Since $C$ is a~tree with additional loop $\ell$, we see that the internal vertices of~$P$ are non-essential in $H$ and $e_i \sim e_{i+1}$ (for $i = 1, \ldots,r-1$), and thus $e_r \in E(S_e)$.
      Moreover, by Lemma~\ref{lem:connection-essential} applied to loop at $u$ and $e$, $e_r$ is an~essential edge, and thus $u$ is an~essential vertex,  so $u \in \bnd S_e$.
    \cqed\end{proof}

    \begin{claim}
      \label{cl:df_two_tree_ferns}
      Suppose that there exist two different tree ferns $S_1, S_2 \in H[\Fc]$ such that $|\bnd S_i| = 2$ and $u \in \bnd S_i$ for $i = 1, 2$.
      Then, each edge of $H^{\Del}$ is essential in $H^{\Del}$ if and only if it is essential in $H$.
    \end{claim}
    \begin{proof}
      We know that no non-essential edge of~$H$ can become essential in~$H^{\Del}$.
      Suppose that there is an essential edge~$e$ of~$H$ which becomes non-essential in $H^{\Del}$.
      By Claim~\ref{cl:df_joint_fern}, the tree fern $S_e \in H[\Fc]$ that contains $e$ satisfies $|\bnd S_e| = 2$ and $u \in \bnd S_e$.
      Without loss of generality, assume that $S_1 \neq S_e$.
      Let $e_1 \in E(S_1)$ be an~edge incident to $u$.
      Since $e$ is essential in~$H$ but non-essential in~$H^{\Del}$, removing $e$ from $H^{\Del}$ must yield a~new component~$C$ which is a~tree in $H^{\Del}$ and contains~$u$.
      As ferns $S_e$ and $S_1$ are edge-disjoint, we have $S_1 \subseteq C$.
      Hence, removing $e_1$ from $H$ creates a~component $C_1$ such that $S_1 - \{ u \} \subseteq C_1 \subseteq C$, and thus $C_1$ is acyclic in~$H$.
      However, this implies that $e_1$ is a~non-essential edge in~$H$ which is a~contradiction with Lemma~\ref{lem:vertex-connection-essential} as $e_1$ lies on the path within $S_1$ connecting the~vertices of $\bnd S_1$.
    \cqed\end{proof}
    
    \begin{claim}
      \label{cl:df_large_dcrit}
      Suppose that $\dcrit^H(u) \geq 4$.
      Then, each feature (i.e.\ vertex or edge) of~$H^{\Del}$ other than $u$ is essential in~$H^{\Del}$ if and only if it is essential in~$H$.
    \end{claim}
    \begin{proof}
      From previous observations we know that is enough to show that an~essential edge~$e$ of~$H$, $e\neq  \ell$, remains essential in~$H^{\Del}$.
      Recall that the loop $\ell$ is counted as two critical incidences of~$u$.
      Since $\dcrit^H(u) \geq 4$, we obtain that $u$ must belong either to two different tree ferns $S_1, S_2 \in H[\Fc]$ with $|\bnd S_i| = 2$, or a cyclic fern $S_c \in H[\Fc]$ with boundary~$\{ u \}$.
      In the first case the assertion follows from Claim~\ref{cl:df_two_tree_ferns}.
      In the second case if $e$ lies on the cycle of $S_c$, then it is essential both in $H$ and in $H^{\Del}$.
      Otherwise, if removing $e$ creates a~new component $C$ that contains $u$, $C$ contains the cycle of $S_c$ as well.
      Hence, removing the loop $\ell$ does not affect whether $e$ is an~essential edge.
    \cqed\end{proof}
    
    \medskip

      \newcommand{\badd}{\textsf{bad}}
      After deleting the~loop $\ell$, some of the essential vertices may become non-essential in $H^{\Del}$.
      From definition of the~relation~$\sim$, whenever a~vertex $w$ becomes non-essential in $H^{\Del}$ all ferns $S_i \in H[\Fc]$ such that $w \in \bnd S_i$ become a~single fern in $H^{\Del}[\Fc^{\Del}]$.
      Hence, in such a~case we need to join the corresponding boundaried trees $T_i \in \Tc$, and remove $w$ from $\quo{H}{\Fc}$.
      Analogously to the case of adding a~loop $\ell$, we introduce an~auxiliary procedure $\makenon(w)$ to handle such a~situation.
      In fact, $\makenon$ can be seen as a~sort of an ``inverse procedure'' to \makeess.
      This procedure takes a~vertex $w$ which is essential in $H$ but non-essential in $H^{\Del}$, and works as follows:
      \begin{itemize}
        \item Start from a~base boundaried tree $T_w$:
        \begin{itemize}
          \item if there exists a~tree $T \in \Tc$ with boundary $\{x\}$ such that $\name(x) = w$, then we take $T_w \coloneqq T$.
          The uniqueness of $T$ follows immediately from the properties of $\sim$.
          The tree $T$, if it exists, can be found in logarithmic time by querying the only element in $\attachment^{-1}((w))$;
          \item otherwise, we spawn a~new tree $T$ in $\Tc$ containing a~fresh vertex $x$ with $\name(x) = w$, and set $T_w \coloneqq T$.
        \end{itemize}
      \item For each edge $e$ of $E(\quo{H}{\Fc})$ with one endpoint $w$:
          \begin{itemize}
            \item find the unique tree $T_e \in \Tc$ corresponding to~$e$ (again, the corresponding top tree can be found by querying $\attachment^{-1}$);
            \item detach the top tree $\Delta_{T_e}$;
            \item recall from the definitions of ferns and their representations that $T_e$ must contain an~edge $x' y$ such that $x'$ is a~leaf of $T_e$, and $\name(x') = \name(x) = w$;
            \item apply the following operations on $\Delta_{\Tc}$:
            \begin{itemize}
              \item $\textsf{cut}(x', y)$,
              \item $\textsf{del}(x')$,
              \item $\textsf{link}(x, y)$.
            \end{itemize}
          \end{itemize}
        \item Call $\delquo(w)$.
        \item Based on the set of edges we considered in the previous step, we may deduce the appropriate values of $\bnd T_x$ and the attachment tuple for it.
      \end{itemize}

      We are ready to move on to the description of modifications we need to make in order to obtain a valid data structure for $H^{\Del}$.
      Recall that we assumed that $u$ is an~essential vertex of~$H$ (recall that $\dcrit^H(u) \geq 3$).
      Then, by Claim~\ref{cl:df_edge_fern}, $\ell$ forms a~separate fern in $H[\Fc]$.
      We start with removing any~boundaried tree $T_{uu} \in \Tc$ representing a~loop at $u$.
    We do this in the~usual way: find $\Delta_{T_{uu}}$ by calling $\textsf{get}(x)$, where $(x,y) \in \edgerep(u, u)$ for some~$y$, detach $\Delta_{T_{uu}}$, and remove $\Delta_{T_{uu}}$ from $\Delta_{\Tc}$.

      Now, we consider three cases:
      \begin{itemize}
        \item $\dcrit^H(u) \geq 5$.
          Then, by Claim~\ref{cl:df_large_dcrit}, all essential/non-essential features of~$H$ other than $u$ remain essential/non-essential in~$H^{\Del}$.
          Hence, $u$ participates in
          \[
            \dcrit^H(u) - 2 \geq 3
          \]
          incidences to essential edges of~$H^{\Del}$, which means that $u$ remains essential in~$H^{\Del}$ as well.
          Therefore, we have $\Fc^{\Del} = \Fc \setminus \{\{ \ell \}\}$, so no further modifications of $\Delta_{\Tc}$ and $\quo{H}{\Fc}$ are required.
        \item $\dcrit^H(u) = 4$.
          Again, by Claim~\ref{cl:df_large_dcrit}, all essential/non-essential elements of~$H$ other than $u$ remain essential/non-essential in~$H^{\Del}$.
          However, this time $u$ participates in exactly
          \[
            \dcrit^H(u) - 2 = 2.
          \]
          incidences to essential edges of~$H^{\Del}$, which means that $u$ becomes non-essential in~$H^{\Del}$.
          In such a~case we should call $\makenon(u)$ in order to update $\quo{H}{\Fc}$ and join trees of $\Tc$ accordingly.
        \item $\dcrit^H(u) = 3$.
          In this case $u$ has a~unique neighbor $v$ in $\quo{H}{\Fc}$ such that $v \neq u$.
          Let $S_{uv}$ be the fern which contributes the~edge $uv$ to $E(\quo{H}{\Fc})$ (the corresponding tree $\Delta_{T_{uv}}$ is the only~element of the set $\attachment^{-1}( (u, v) )$).
          Similarly to the previous cases, we compute that now $u$ is incident to at most one essential edge.
          This means that $u$ becomes non-essential in~$H^{\Del}$, and we need to call $\makenon(u)$.
          Furthermore, one can observe that every edge on the path $u \Path v$ in $S_{uv}$ becomes non-essential in $H^{\Del}$, as it now separates a~tree containing $u$ from the rest of the graph.
          On the other hand, by Claim~\ref{cl:df_joint_fern}, all other essential edges of $H$ remain essential in~$H^{\Del}$, and thus all essential vertices in $H$ (except $u$ and potentially $v$) remain essential in~$H^{\Del}$.
          Hence, we can conclude that the value of $\dcrit^{H^{\Del}}(v)$ equals either $\dcrit^{H}(v) - 1$ if $\dcrit^{H}(v) \geq 4$, or $0$ if $\dcrit^{H}(v) = 3$.
          In the latter case, $v$ becomes non-essential as well, and we need to call $\makenon(v)$.
          Observe that this operation only joins some trees in $\Delta_{\Tc}$.
          In particular, it does not make any other essential vertex non-essential, and thus this procedure terminates after this call.
      \end{itemize}

  Summing up, we see that each update of $H$ requires a constant number of modifications of $(\Delta_\Tc, \allowbreak \name, \allowbreak \quo{H}{\Fc}, \allowbreak \attachment, \allowbreak \name^{-1}, \allowbreak \attachment^{-1}, \allowbreak \edgerep)$. Each of these structures is of size $\Oh{|H|}$ and supports operations in worst-case time logarithmic in its size, hence the running time of $\F$ is $\Oh{\log |H|}$ per update of~$H$.
\end{proof}

We finish this part with a~remark that for the sole purpose of this section, instead of top trees, we could have used slightly simpler data structure on dynamic forest such as link/cut trees~\cite{SleatorT83}, as in the work of Alman et al.~\cite{AlmanMW20}.
In the next section, we will see why it is convenient to choose top trees as the~underlying data structure.

\subsection{Dynamic maintenance of ensemble contractions}
\label{ssec:lemmaA_dynamic_contractions}


\newcommand{\ColoredEdge}{\mathsf{colored\_edge}}
\newcommand{\ColoredVertex}{\mathsf{colored\_vertex}}
\newcommand{\RemovedList}{\mathsf{removed}}
\newcommand{\RecolorList}{\mathsf{recolor}}
\newcommand{\CountDict}{\mathsf{count}}

We will now combine the results of Sections~\ref{ssec:lemmaA_static} and~\ref{ssec:lemmaA_dynamic_fern}.
Namely, given a~dynamic augmented structure $(\Af, H)$, we will prove that we can efficiently maintain the rank-$p$ contraction of the ensemble $\Xx$, constructed in the proof of the static variant of the \lemmaA in Section~\ref{ssec:lemmaA_static}.
The data structure for dynamic ensembles will extend the dynamic data structure maintaining $\Fc$ from Section~\ref{ssec:lemmaA_dynamic_fern}.

This subsection is devoted to the proof of the following proposition:

\begin{lemma}
  \label{lem:dynamic_contraction}
  Let $(\Af, H)$ be a~dynamic augmented structure over a~binary relational signature $\Sigma$, which is initially empty, in which we are allowed to add and remove isolated vertices, edges or tuples to relations.
  After each update, $\Af$ must be guarded by $H$.

  There exists a~data structure $\C$ which, when initialized with an~integer $p \in \N$, maintains $\Contract^p(\Xx)$, where $\Xx$ is the ensemble constructed from $\Ff$ given in Section~\ref{ssec:lemmaA_static}.
  Each update to $(\Af, H)$ can be processed by $\C$ in worst-case $\Oh[p,\Sigma]{\log |H|}$ time and requires $\Oh[p, \Sigma]{1}$ updates to $\Contract^p(\Xx)$, where each update adds or removes a single element or a single tuple to a relation.
\end{lemma}

  Recall from Lemma~\ref{lem:dynamic_ferns} that there exists a~dynamic data structure $\F$ maintaining a~forest of top trees $\Delta_{\Tc}$ that, together with a~dynamic mapping $\name\colon V(\Tc) \to V(H)$, represents the fern decomposition of $H$.
   In the proof, we will gradually extend $\F$ by new functionality, which will eventually allow us to conclude with a~data structure $\C$ claimed in the statement of the lemma.
    
    \paragraph*{Augmenting $\Tc$ with a~relational structure.}
    Let $(\Af, H)$ be an~augmented $\Sigma$-structure.
    Then, let $\Fc$ be the fern decomposition of $H$ constructed in Subsection~\ref{ssec:lemmaA_ferns}; let $\Xx$ be the ensemble constructed from $\Af$ and $\Fc$ in Subsection~\ref{ssec:lemmaA_static}; and let $\Tc$ be the forest maintained by $\Delta_{\Tc}$ in Subsection~\ref{ssec:lemmaA_dynamic_fern} which, together with $\name$, represents $\Fc$.
    For every fern $S \in H[\Fc]$, let $\Af_S$ be the fern element of $\Xx$ corresponding to $S$, and let $\Tc_S$ be the component of $\Tc$ which, together with $\name$, represents $S$.

    Recall from Section~\ref{ssec:top_trees} that $\Tc$ can be extended with auxiliary information by assigning it a~$\Sigma$-structure $\Bf$ guarded by $\Tc$.
    Then, the interface of $\Delta_\Tc$ is extended by two new methods: $\AddRel$ and $\DelRel$, defined in Section~\ref{ssec:top_trees}.
    The structure $\Bf$ will be defined in a moment, intuitively it corresponds to $\Af$ split into individual elements of the ensemble $\Xx$, which in turn are guarded by the trees of forest $\Tc$. 
    
    For every fern $S \in H[\Fc]$, let $\Bf_S \coloneqq \Bf[\Tc_S]$ be the boundaried substructure of $\Bf$ induced by $\Tc_S$.
    Here, $\Bf_S$ will be a~substructure of $\Bf$ guarded by $\Tc_S$, which is a~tree representation of a~single fern element $\Af_S \in \Xx$. 
    The boundary $\bnd \Bf_S$ will be equal to the set $\bnd \Tc_S$ of external boundary vertices of $\Tc_S \in \Tc$, which in turn represents the set essential vertices of $H$ to which $\Tc_S$ is attached.
    Thus, in the language of relational structures, the boundary of $\Bf_S$ corresponds naturally to $\bnd \Af_S$.

    Note that $\Bf$ is the disjoint sum over $\Bf_S$ for all ferns $S \in H[\Fc]$, hence in order to describe $\Bf$, we only need to describe $\Bf_S$ for each $S$.
    For the ease of exposition, we will often say that if $\Tc_S$ represents $S$, then each of $\Tc_S$ and $\Bf_S$ represents both $S$ and $\Af_S$.

    We now construct a structure $\Bf \coloneqq \Bf(\Xx, \Tc)$ in such a~way that for each fern $S \in H[\Fc]$, the rank-$p$ type of $\Af_S$ can be deduced uniquely from $\bnd \Bf_S$, $\tp^p(\Bf_S)$ and the attachment tuple of $\Tc_S$, given access to $\name$ as an~oracle.
    Fix $S \in H[\Fc]$.
    Then:
          
    \begin{itemize}
      \item If $S$ is a~tree fern, then we construct $\Bf_S$ as a~structure isomorphic to $\Af_S$, with the isomorphism given by $\name$.
      Note that this isomorphism exists since $\Bf_S$ is a~relational structure built on $\Tc_S$, $\Af_S$ is a~relational structure built on $S$, and $\name(\Tc_S) = S$.
      Observe that we have $\name(\Bf_S) = \Af_S$.
      
      \item If $S$ is an~unicyclic fern, then we need to tweak the construction.
        Assume that the cycle in $S$ is split at vertex $w \in V(S)$, and that its two copies in $\Tc_S$ are $x_1$ and $x_2$; the remaining vertices of $\Tc_S$ are in a~bijection with $V(S) \setminus \{w\}$.
        Let $\Bf_S$ be an~initially empty relational structure with $V(\Bf_S) = V(\Tc_S)$.
        We shall now describe how tuples are added to the relations of $\Bf_S$.
        
        First, consider a~vertex $v$ of $S$.
        Then, let $y$ be an arbitrarily chosen element of $\Bf_S$ with $\name(y) = v$.
        (This choice is unique if $v \neq w$.)
        The element $y$ inherits the interpretations of all unary and binary relations on $v$ in $\Af$.
        That is, for each relation $R \in \Sigma$, if $v \in R^{\Af_S}$ (respectively, $(v, v) \in R^{\Af_S}$), then we add $y$ (resp. $(y, y)$) to $R^{\Bf_S}$.
        
        Similarly, consider a~pair $(u, v)$ such that $u \neq v$ and $uv \in E(S)$.
        Then, let $(y, z)$ be a~pair of elements of $\Bf_S$, chosen arbitrarily, so that $yz \in E(\Tc_S)$, $\name(y) = u$ and $\name(z) = v$.
        (This choice is usually unique, apart from the case where the cycle of $S$ has length exactly $2$.)
        Then, for each binary relation $R \in \Sigma^{(2)}$, if $(u, v) \in R^{\Af_S}$, then add $(y, z)$ to $R^{\Bf_S}$.
        
        It can be now easily checked that $\Af_S = \name(\glue_\psi(\Bf_s))$, where $\psi : \{x_1, x_2\} \to \{x_2\}$ is the function given by $\psi(x_1) = \psi(x_2) = x_2$.
    \end{itemize}
    
    Since all fern elements of $\Xx$ have empty flags, the same also holds for $\B$.
    Moreover, (almost) no boundary element $d \in \bnd \B$ satisfies any unary predicates, and the interpretations of binary predicates of $\Sigma$ in $\B$ (usually) do not contain $(d, d)$.
    The only exception is given by unicyclic ferns of $\Fc$ with empty boundary: recall that the component $\Bf_S$ of $\Bf$ representing such a~fern is formed by splitting the cycle of the fern along a~non-deterministically chosen vertex $v$ of the cycle---which is non-essential by the properties of ferns.
    Then, exactly one copy $x$ of $v$ in $\Bf_S$ inherits the interpretations of unary and binary predicates from $\Af$, even though $x \in \bnd \Bf_S$.
    
    As promised, we have:
    
    \begin{lemma}
      \label{lem:tree_type_to_fern_type}
      There exists a~function $\mathsf{treeTypeToFernType(\cdot, \cdot, \cdot)}$ which, given $\bnd \Bf_S$, $\tp^p(\Bf_S)$ and the attachment tuple $\bar{a}$ of $\Tc_S$ as its three arguments, and given access to $\name$ as a~dynamic dictionary, computes $\tp^p(\Af_S)$ in $\Oh[p, \Sigma]{\log |H|}$ time.
      
      \begin{proof}
        We consider all different shapes of the fern $S$.
        We will show that each of them can be distinguished by the number of different elements of $\bar{a}$ and the size of $|\bnd \Bf_S|$, and that $\tp^p(\Af_S)$ can be computed efficiently in each of the cases.
        Let $A \coloneqq \{u \mid u \in \tup{a}\}$.
        By the definition of $\tup{a}$, we have $A = \bnd \Af_S$.
        Then:
        
        \begin{itemize}
          \item If $S$ is a~tree fern with $\ell \in \{0, 1, 2\}$ elements in the boundary, then $|A| = |\bnd \Bf_S| = \ell$.
          Since $\Af_S$ is isomorphic to $\Bf_S$, the type $\tp^p(\Af_S)$ can be constructed from $\tp^p(\Bf_S)$ by replacing each occurrence of a~boundary element $d \in \bnd \Bf_S$ with $\name(d)$.
          
          \item If $S$ is a~unicyclic fern with $\bnd S = \{v\}$, then $|A| = 1$, but $|\bnd \Bf_S| = 2$ (i.e., the boundary of $\Bf_S$ consists of two copies, say $x_1$, $x_2$, of $v$).
          Let $\bnd \Bf_S = \{x_1, x_2\}$.
          Then, $\tp^p(\Af_S)$ is given by
          \[ \tp^p(\Af_S) = \left(\iota_v^{-1} \circ \iota_{x_2}\right)\left( \glue_{\psi}^{p, \Sigma} \left( \tp^p(\Bf_S) \right) \right), \]
          where $\psi\colon\{x_1, x_2\} \to \{x_2\}$ is defined as $\psi(x_1) = \psi(x_2) = x_2$.
          Intuitively, given a~structure $\Bf_S$, we first glue both copies of $v$ in $\Bf_S$ into one vertex $x_2$.
          The resulting structure is isomorphic to $\Af_S$, with an isomorphism $\name$ sending $x_2$ to $v$.
          Hence, the rank-$p$ type of $\Af_S$ can be retrieved from the rank-$p$ type of $\Bf_S$.
          
          \item If $S$ is a~unicyclic fern with $\bnd S = \emptyset$, then $|A| = 0$, but $|\bnd \Bf_S| = 2$ (i.e., the boundary of $\Bf_S$ comprises two copies, say $x_1$, $x_2$, of some vertex on the cycle of $S$).
          Let $\bnd \Bf_S = \{x_1, x_2\}$.
          Then, $\tp^p(\Af_S)$ is given by
          \[ \tp^p(\Af_S) = \forget_{\{x_2\}}^{p, \Sigma}\left( \glue_{\psi}^{p, \Sigma} \left( \tp^p(\Bf_S) \right) \right), \]
          where $\psi$ is defined as above.
          Intuitively, given a~structure $\Bf_S$, we first glue both copies of $v$ in $\Bf_S$ into one vertex, which is then removed from the boundary.
          The resulting boundaryless structure is isomorphic to $\Af_S$, thus its type is exactly $\tp^p(\Af_S)$.
        \end{itemize}
        Hence, all cases can be distinguished by the sizes of $A$ and $\bnd \Bf_S$, and in each case, we can compute $\tp^p(\Af_S)$ in $\Oh[p, \Sigma]{\log |H|}$ time, which is dominated by the queries to $\name$ in the first case.
      \end{proof}
    \end{lemma}

    \paragraph*{Deducing $\msotwo{}$ types of the clusters.}
    We will now show that the clusters of $\Delta_{\Tc}$ can be augmented with information related to the $\msotwo{}$ types of substructures of $\Bf$.
    Recall that top trees can be $\mu$-augmented by assigning each cluster $(C, \bnd C)$ of $\Delta_{\Tc}$ an~abstract piece of information $\mu_{\bnd C}(\Bf\{C\})$ about the substructure of $\Bf$ almost induced by $C$.
    Here, for each finite $D \subseteq \Omega$, $\mu_D$ is a~mapping from stripped boundaried structures with boundary $D$ to some space $I_D$ of possible pieces of information.
    
    We now define $\mu_D$ and $I_D$.
    For every finite set $D \subseteq \Omega$, let $\mu_D(\X) \coloneqq \tp^p(\X)$ be the function that assigns each stripped boundaried structure $\X$ over $\Sigma$ with $\bnd \X = D$ its rank-$p$ type.
      Let also $I_D \coloneqq \Types^{p, \Sigma}(D)$ be the set of different rank-$p$ types of structures with boundary $D$.
      We will prove the following:
    
\begin{lemma}
    \label{lem:storing_stripped_types}
      The top trees data structure $\Delta_{\Tc}$ can be $\mu$-augmented.
      Moreover, each update and query on $\mu$-augmented $\Delta_{\Tc}$ can be performed in worst-case $\Oh[p, \Sigma]{\log n}$ time, where $n = |V(\Tc)|$.
    \begin{proof}
    
     We now prove a~series of claims about the properties from $\mu$.
     From these and Lemma~\ref{lem:top_trees_additional_information}, the statement of the lemma will be immediate.

    \begin{claim}
      \label{cl:contraction_efficient_mu}
      For each $D \subseteq \Omega$, $|D| \leq 2$, the mapping $\mu_D$ can be computed in $\Oh[p, \Sigma]{1}$ worst-case time from any stripped boundaried structure with at most $2$ vertices.
      \begin{proof}
        Trivial as here, $\mu_D$ is only evaluated on structures of constant size.
      \cqed\end{proof}
    \end{claim}
    
    \begin{claim}[Efficient compositionality of $\mu$ under joins]
      \label{cl:contraction_mu_join}
      For every finite $D_1, D_2 \subseteq \Omega$, there exists a~function $\oplus_{D_1, D_2}\colon I_{D_1} \times I_{D_2} \to I_{D_1 \cup D_2}$ such that for every pair $\Xf_1$, $\Xf_2$ of stripped boundaried structures with $\bnd \Xf_1 = D_1$, $\bnd \Xf_2 = D_2$, we have that:
\[ \mu_{D_1 \cup D_2}\left(\Xf_1 \oplus \Xf_2\right) = \mu_{D_1}(\Xf_1)\, \oplus_{D_1, D_2} \, \mu_{D_2}(\Xf_2). \]
      Moreover, if $|D_1|, |D_2| \leq 2$ and $|D_1 \cap D_2| = 1$, then $\oplus_{D_1, D_2}$ can be evaluated on any pair of arguments in worst-case $\Oh[p, \Sigma]{1}$ time.
      \begin{proof}
        Such a~function $\oplus_{D_1, D_2}$ is simply given by $\oplus^{p, \Sigma}_{D_1, D_2}\colon \Types^{p,\Sigma}(D_1) \times \Types^{p, \Sigma}(D_2) \to \Types^{p, \Sigma}(D_1 \cup D_2)$ defined in Lemma~\ref{lem:compositionality}.
        That $\oplus_{D_1, D_2}$ commutes with $\mu$ and that $\oplus_{D_1, D_2}$ can be efficiently evaluated for $|D_1|, |D_2| \leq 2$ follows from Lemma~\ref{lem:compositionality}.
      \cqed\end{proof}
    \end{claim}
    
    \begin{claim}[Efficient compositionality of $\mu$ under forgets]
    \label{cl:contraction_mu_forget}
    For every finite $D \subseteq \Omega$, and $S \subseteq D$, there exists a~function $\forget_{S,D}\,:\, I_D \times \left(2^S\right)^{\Sigma^{(1)} \cup \Sigma^{(2)}} \to I_{D \setminus S}$ so that for every stripped boundaried structure $\Xf$ with $\bnd \Xf = D$, and $P \in \left(2^S\right)^{\Sigma^{(1)} \cup \Sigma^{(2)}}$, we have:
\[ \mu_{D \setminus S}\left(\forget_S(\Xf, P)\right) = \forget_{S,D}\left( \mu_D(\Xf), P \right). \]
    Moreover, if $|D| \leq 3$, then $\forget_{S, D}$ can be evaluated on any pair of arguments in worst-case $\Oh[p, \Sigma]{1}$ time.
    \begin{proof}
    Fix $S$, $D$, and $P$ as above.
    We define the relational structure $\Yf \coloneqq \Yf(S, P)$ with $V(\Yf) = \bnd\Yf = S$ as a~structure with an edgeless Gaifman graph $G(\Yf)$, whose interpretations of unary relations on $S$ and binary relations on self-loops on $S$ are given by $P$.
    Formally,
    \[
    \begin{split}
    R^{\Yf} &= P(R)\qquad\text{for }R \in \Sigma^{(1)}, \\
    R^{\Yf} &= \{(x, x)\,\mid\,x \in P(R)\}\qquad\text{for }R \in \Sigma^{(2)}.
    \end{split}
    \]
    Then, for any stripped boundaried structure $\Xf$, we have that $\forget_S(\Xf, P) = \forget_S(\Xf \oplus \Yf)$.
    In particular, by Lemma~\ref{lem:compositionality},
    \[
    \tp^p\left(\forget_S(\Xf, P)\right) = \tp^p\left(\forget_S(\Xf \oplus \Yf)\right) =
      \forget_{S,D}^{p,\Sigma}\left(\tp^p(\Xf) \oplus_{D,S}^{p,\Sigma} \tp^p(\Yf)\right),
    \]
    where $\tp^p(\Yf)$ can be computed in $\Oh[p, \Sigma]{1}$ time from $S$, $P$, $p$, and $\Sigma$, as long as $|S| = \Oh{1}$.
    Since $\mu_D$ is defined as the rank-$p$ type of a~given stripped boundaried structure, we conclude that the function
    \[ \forget_{S,D}\left(\alpha, P\right) \coloneqq
      \forget_{S,D}^{p,\Sigma}\left(\alpha \oplus_{D,S}^{p,\Sigma} \tp^p\left(\Yf(S, P)\right)\right)
    \]
    satisfies the compositionality of $\mu$.
    Naturally, for $|D| \leq 3$, the function is computable in $\Oh[p, \Sigma]{1}$ time.
    \cqed\end{proof}
    \end{claim}
    
    \begin{claim}[Efficient isomorphism invariance of $\mu$]
    \label{cl:contraction_mu_iota}
    For every finite $D_1, D_2 \subseteq \Omega$ of equal cardinality, and for every bijection $\phi\colon D_1 \to D_2$, there exists a~function $\iota_{\phi}\colon I_{D_1} \to I_{D_2}$ such that for every pair $\Xf_1$, $\Xf_2$ of \emph{isomorphic} boundaried structures with $\bnd \Xf_1 = D_1$, $\bnd \Xf_2 = D_2$, with an~isomorphism $\wh{\phi}\colon V(\Xf_1) \to V(\Xf_2)$ extending $\phi$, we have:
  \[ \mu_{D_2}(\Xf_2) = \iota_{\phi}\left( \mu_{D_1}(\Xf_1) \right). \]
  Moreover, if $|D_1| \leq 2$, then $\iota_{\phi}$ can be evaluated on any argument in worst-case $\Oh[p, \Sigma]{1}$ time.
  \begin{proof}
  Define $\iota_{\phi}$ as a~function taking a~type $\alpha \in \Types^{p,\Sigma}(D_1)$ as an~argument, and replacing every occurrence of a~boundary element $d \in D_1$ with $\phi(d) \in D_2$ within every sentence of $\alpha$.
  Naturally, given two isomorphic boundaried structures $\Xf_1$, $\Xf_2$ with $\bnd \Xf_1 = D_1$, $\bnd \Xf_2 = D_2$, with an isomorphism $\wh{\phi}\colon V(\Xf_1) \to V(\Xf_2)$ extending $\phi$, we have $\tp^p(\Xf_2) = \iota_{\phi}(\tp^p(\Xf_1))$.
  Thus, $\iota_{\phi}$ witnesses the isomorphism invariance of $\mu$.
  
  Naturally, for every $D_1$ with $|D_1| \leq 2$, $\iota_{\phi}$ can be evaluated in $\Oh[p, \Sigma]{1}$ time.
  \cqed\end{proof}
    \end{claim}
    
    By Claims~\ref{cl:contraction_efficient_mu}, \ref{cl:contraction_mu_join}, \ref{cl:contraction_mu_forget}, \ref{cl:contraction_mu_iota}, and Lemma~\ref{lem:top_trees_additional_information}, $\Delta_{\Tc}$ can be $\mu$-augmented.
    Moreover, each update and query on $\Delta_{\Tc}$ can be performed in worst-case $\Oh[p, \Sigma]{\log n}$ time.
    This concludes the proof.
  \end{proof}
\end{lemma}

From now on, assume that $\Delta_{\Tc}$ is $\mu$-augmented.
    Thus, by Lemma~\ref{lem:storing_stripped_types}, for each fern $S \in H[\Fc]$, we can read from $\Delta_{\Tc}$ the rank-$p$ type of $\Strip(\Bf_S)$, stored in the root cluster of the top tree corresponding to $\Bf_S$.
    It still remains to show that we can recover from $\Delta_{\Tc}$ the rank-$p$ type of $\Bf_S$.
\begin{lemma}
  \label{lem:top_trees_gettype}
  The interface of $\Delta_{\Tc}$ can be extended by the following method:
  \begin{itemize}
    \item $\mathsf{getType}(\Delta_T)$: given a~reference to a~top tree $\Delta_T$, returns the rank-$p$ type of the relational structure $\Bf_T$ described by $\Delta_T$.
  \end{itemize}
  This method runs in worst-case $\Oh[p, \Sigma]{\log n}$ time, where $n = |V(\Tc)|$.
  \begin{proof}
    Let $\bnd_{\mathrm{old}}$ be the current boundary of $\Delta_T$.
    We call $\mathsf{clearBoundary}(\Delta_T)$.
    Since $\Bf_T$ satisfies no nullary predicates by the construction of $\Bf$, we immediately infer that the sought rank-$p$ type of $\Bf_T$ is stored in the root cluster of $\Delta_T$.
    Before returning from the method, we restore the boundary of $\Delta_T$ by calling $\mathsf{expose}(\bnd_{\mathrm{old}})$.
  \end{proof}
\end{lemma}

  We remark that Lemma~\ref{lem:top_trees_gettype} can be chained with Lemma~\ref{lem:tree_type_to_fern_type} in order to recover the rank-$p$ type of the original fern element $\Af_S$ in $\Xx$, as long as we have access to $\attachment$ dictionary, mapping top trees to their attachment tuples.
\begin{corollary}
  \label{cor:top_tree_to_fern_type}
  Given a~fern $S \in H[\Fc]$, let $\Af_S \in \Xx$ be the fern element associated with $S$, let $\Tc_S$ be the tree component of $\Tc$ representing $S$, and let $\Delta_S$ be the top tree maintaining $\Tc_S$.
  Then, there exists a~function $\mathsf{topTreeToFernType}(\cdot)$, which, given a~reference to a~top tree $\Delta_T$ as its only argument, and given access to $\name$ and $\attachment$ as relational dictionaries, computes $\tp^p(\Af_S)$ in $\Oh[p, \Sigma]{\log |H|}$ time.
  \begin{proof}
    We have
    \[ \mathsf{topTreeToFernType}(\Delta_S) = \mathsf{treeTypeToFernType}(\bnd \Delta_S, \mathsf{getType}(\Delta_S), \attachment(\Delta_S)), \]
    where $\bnd \Delta_S$ is the set of external boundary vertices of $\Delta_S$, and $\attachment(\Delta_S)$ is the attachment tuple of $\Tc_S$.
    Since $|V(\Tc)| = \Oh{|H|}$, the call to $\mathsf{getType}$ takes worst-case $\Oh[p, \Sigma]{\log |H|}$ time, as well as the call to $\mathsf{TreeTypeToFernType}$.
  \end{proof}
\end{corollary}
    
    \paragraph*{Dynamic maintenance of $\Xx$ and related information.}
    We now show how to maintain the ensemble $\Xx$ and its representation in $\Delta_{\Tc}$ dynamically under the updates of vertices, edges, and tuples of the augmented structure $(\Af, H)$.
    Under each of these updates, $\Xx$ and its representation will be modified in worst-case $\Oh[p, \Sigma]{\log |H|}$ time.

    More importantly, we will show that we can dynamically deduce type information about ferns stored in $\Xx$.
    For each $j \in \{0, 1, 2\}$, define a~function $\CountDict_j: V(H)^j \times \Types^{p,\Sigma}([j]) \to \N$.
    For any tuple $\tup{a} \in V(H)^j$ sorted by $\leq$ and with no two equal elements, and a~type $\alpha \in \Types^{p,\Sigma}([j])$, we define $\CountDict_j(\tup{a}, \alpha)$ as the number of boundaried structures $\Xf \in \Xx$ for which $\bnd \Xf = \{u \mid u \in \tup{a}\}$ and $\iota_{\bnd \Xf}\left(\tp^p(\Xf) \right) = \alpha$.

    Interestingly, it will turn out later in the proof that the rank-$p$ contraction of $\Xx$, whose maintenance is the main objective of Lemma~\ref{lem:dynamic_contraction}, can be uniquely inferred from the functions $\CountDict_j$.
    Therefore, we propose a~data structure that maintains $\Xx$ dynamically and which allows one to examine the properties of $\Xx$ through $\CountDict_0$, $\CountDict_1$, and $\CountDict_2$ only.
  \begin{lemma}
    \label{lem:dynamic_ensemble}
    There exists a~data structure $\C_0$ which, given an~initially empty dynamic augmented structure $(\Af, H)$, updated by adding or removing vertices, edges, or tuples to or from $(\Af, H)$, maintains the functions $\CountDict_0$, $\CountDict_1$, and $\CountDict_2$ deduced from $\Xx$---the ensemble constructed from $\Af$ and $H$.

    Each update to $\C_0$ (addition or removal of a~vertex, edge, or tuple to or from $(\Af, H)$) is processed by $\C_0$ in $\Oh[p, \Sigma]{\log n}$ time and causes $\Oh[p, \Sigma]{1}$ changes to values of $\CountDict_j$.
  \begin{proof}
  We start with an~extra definition.
  Call a~boundaried structure $\Xf \in \Xx$ \emph{detached} if $\Xf$ is a~fern element of $\Xx$ represented by a~tree detached from $\quo{H}{\Fc}$.
  Otherwise, call $\Xf$ attached.
  In particular, singleton elements and the flag element of $\Xx$ are deemed attached.  
  
  We begin with the description of auxiliary data structures.
    We keep:
    \begin{itemize}
      \item An~instance of $\F$, defined in Lemma~\ref{lem:dynamic_ferns}, maintaining the fern decomposition $\Fc$ of $H$ dynamically under the updates of $H$.
      $\F$ implements a~$\mu$-augmented top trees data structure $\Delta_{\Tc}$ for a~forest $\Tc$ guarding a~$\Sigma$-structure $\Bf$.
      \item For every $i \in \{0, 1, 2\}$ and for every relation $R \in \Sigma^{(i)}$, a~dynamic set $R^{\Af}$ storing the interpretation of $R$ in $\Af$.
      \item A~dynamic dictionary $\ColoredEdge \colon V(H)^2 \rightharpoonup V(\Tc)^2$, mapping each pair $(u, v)$, $u, v \in V(H)$, $u \neq v$, such that $uv \in E(H)$, to an~arbitrary pair $(x, y)$ with $x, y \in V(\Tc)$ such that $xy$ is an~edge of $\Tc$ and $\name(x) = u$, $\name(y) = v$.
      \item A~dynamic dictionary $\ColoredVertex \colon V(H) \to V(\Tc)$, mapping every vertex $v \in V(H)$ to an~arbitrary vertex $x \in V(\Tc)$ such that $\name(x) = v$.
      \item For $j \in \{0, 1, 2\}$, dynamic dictionaries $\CountDict^{\star}_j: V(H)^j \times \Types^{p,\Sigma}([j]) \to \N$ defined analogously as $\CountDict_j$, only that for $\tup{a} \in V(H)^j$ and $\alpha \in \Types^{p,\Sigma}([j])$, $\CountDict^{\star}_j(\tup{a}, \alpha)$ is defined as the number of \emph{attached} boundaried structures $\Xf \in \Xx$ for which $\bnd \Xf = \{u \mid u \in \tup{a}\}$ and $\iota_{\bnd \Xf}(\tp^p(\Xf)) = \alpha$.
      A~key $(\tup{a}, \alpha)$ is stored in $\CountDict_j$ if and only if $\CountDict_j(\tup{a}, \alpha) \geq 1$; this way, the dictionaries $\CountDict_j$ contain at most $|\Af|$ elements in total.
    \end{itemize}
    
    All these data structures can be trivially initialized on an~empty augmented structure $(\Af, H)$ and the ensemble $\Xx$ consisting of exactly one element---an~empty fern element.
  
  We remark that every structure $\Xf \in \Xx$ contributes $1$ to exactly one value $\CountDict^{\star}_j(\cdot, \cdot)$, as long as $\Xf$ is attached; otherwise, $\Xf$ contributes to no values of $\CountDict^{\star}$.
    This distinction between attached and detached structures will be vital in our data structure: since no updates can be performed by $\F$ on any boundaried structure unless it is detached (Lemma~\ref{lem:dynamic_ferns}, invariant~\ref{inv:modify_detached_tree}), we guarantee that the types of boundaried structures contributing to $\CountDict^{\star}$ cannot change.
    When $\F$ updates a~component of $\Tc$ corresponding to $\Xf$, it needs to detach the component first by calling $\detach$; the call will be intercepted by us and used to exclude $\Xf$ from contributing to $\CountDict^{\star}$.
    Then, when $\F$ finishes updating tree components, it will reattach all unattached components (Lemma~\ref{lem:dynamic_ferns}, invariant~\ref{inv:all_attached}), causing us to include the resulting boundaried structures in $\CountDict^{\star}$.
    Thus, after each query, we are guaranteed to have $\CountDict_j = \CountDict^{\star}_j$.

    We shall now describe how $\Xx$ is tracked using the auxiliary data structures.
    Firstly, for each $S \in H[\Fc]$, the fern element $\Af_S \in \Xx$ is represented by a~single connected component of the forest $\Tc$ underlying $\Delta_{\Tc}$.

    Next, for each non-essential vertex $v$ of $H$, we store one copy of $v$ in $\Tc$ in $\ColoredVertex(v)$.
    Then, $x \coloneqq \ColoredVertex(v)$ inherits in $\Bf$ the interpretations on $v$ of all unary and binary of relations in $\Af$.
    Formally, for every $R \in \Sigma^{(1)}$, we have $x \in R^{\Bf}$ if and only if $v \in R^{\Af}$; and similarly, for every $R \in \Sigma^{(2)}$, we have $(x, x) \in R^{\Bf}$ if and only if $(v, v) \in R^{\Af}$.
    We remark that in most cases, non-essential vertices $v$ have exactly one copy in $\Tc$.
    However, if the cycle of some unicyclic fern with empty boundary is split along $v$ in $\Tc$, then $v$ has exactly two copies in $\Tc$.
    In this case, only one copy inherits the interpretations of relations from $v$.
    
    Similarly, for every pair $(u, v)$ of distinct vertices of $H$ such that the edge $uv$ is in $H$, we store the endpoints $(x, y) \coloneqq \ColoredEdge((u, v))$ of some copy of this edge in $\Tc$.
    As previously, for every $R \in \Sigma^{(2)}$, we set $(x, y) \in R^{\Bf}$ if and only if $(u, v) \in R^{\Af}$; and the remaining copies of the edge in $\Tc$ do not retain any binary relations from $\Af$.
    Thanks to this choice, an~addition or removal of $(u, v)$ from the interpretation of some binary predicate $R$ in $\Af$ requires the update of the interpretation of $R$ in $\Bf$ only on one pair of vertices.
    
    Finally, singleton elements and the flag element of $\Xx$ are stored in our data structure implicitly.
    This is possible since all these elements can be uniquely reconstructed from the dynamic sets $R^{\Af}$.
    
    In this setting, fern elements of $\Xx$ can be inferred from $\Bf$ and $\name$, and singleton elements and the flag element can be deduced from the dictionaries $R^{\Af}$ for $R \in \Sigma^{(0)} \cup \Sigma^{(1)}$.
    Hence, if a~change of the interpretation of some fern element $\Af_S$ of $\Xx$ is required, we perform it by modifying the  the component $\Bf_S$ of $\Bf$ representing $\Af_S$.
    This in turn is done by calling $\AddRel$ or $\DelRel$ on $\Delta_{\Tc}$ with appropriate arguments.
    
    \medskip
    
    We now show how queries to our data structure are processed.
    First, we will prove a~helpful observation:
    \begin{claim}
      \label{cl:addrel_delrel_update}
      Suppose that all top trees of $\Delta_T$ are attached.
      When $\Xx$ is updated by an~$\AddRel$ or $\DelRel$ call to $\Delta_{\Tc}$, the dictionaries $\CountDict^\star_j$ can be updated under this modification in $\Oh[p, \Sigma]{\log |H|}$ worst-case time, requiring at most $2$ updates to all $\CountDict^\star_j$ in total.
      \begin{proof}
        Without loss of generality, assume that $\AddRel(R, \tup{b})$ is to be called for $R \in \Sigma^{|\tup{b}|}$.
        We can also assume that $\tup{b}$ is nonempty as in our case, $\Bf$ does not store any flags.

        Let $x$ be any element of $\tup{b}$.
        We locate, in logarithmic time, the top tree $\Delta_T \in \Delta_{\Tc}$ containing $x$ as a~vertex.
        Then, $\Delta_T$ represents a~tree component $T$ of $\Tc$; let $\Bf_T$ be the substructure of $\Bf$ induced by $T$, and assume that $T$ corresponds to a~fern element $\Af_T \in \Xx$.

        By Corollary~\ref{cor:top_tree_to_fern_type}, we can get the current type $\alpha_{\mathrm{old}}$ of $\Af_T$ in $\Oh[p, \Sigma]{\log |H|}$ time by calling $\mathsf{topTreeToFernType}(\Delta_T)$.
        Since the type of $\Af_T$ might change under the prescribed modification, we temporarily exclude $\Af_T$ from participating in $\CountDict^\star$ by subtracting $1$ from $\CountDict^\star_{|\bnd \Af_T|}\left( \iota_{\bnd \Af_T}(\alpha_{\mathrm{old}}) \right)$; note that $\bnd \Af_T$ can be uniquely recovered from $\tup{a}$ as $\bnd \Af_T = \{u \colon u \in \tup{a}\}$.
        
        Now, we run the prescribed call to $\Delta_{\Tc}$, causing a~change to $\Bf_T$ (and, hence, to $\Af_T$).
        Afterwards, we retrieve the new type $\alpha_{\mathrm{new}}$ of $\Af_T$ by calling $\mathsf{topTreeToFernType}(\Delta_T)$ again.
        Then, we include the type of $\Af_T$ back to $\CountDict^\star$ by increasing $\CountDict^\star_{|\bnd \Af_T|}\left( \iota_{\bnd \Af_T}(\alpha_{\mathrm{new}}) \right)$ by $1$.

        This concludes the update.
        Naturally, the operation took logarithmic time, and $\CountDict^\star$ was updated twice.
      \cqed\end{proof}
    \end{claim}    
    
    Now, assume that a~query to $\C_0$ requests adding or removing a~tuple $\tup{b}$ ($|\tup{b}| \leq 2$) from the interpretation of some predicate $R \in \Sigma^{|\tup{b}|}$ in $\Af$.
    The query is not relayed to $\F$, so all trees maintained by $\Delta_{\Tc}$ stay attached throughout the query.
    We now consider several cases, depending on the contents of $\tup{b}$:
    \begin{itemize}
      \item If $\tup{b} = \varepsilon$, then the changed flag element $\Xf_{\emptyset}$ is maintained implicitly in $R^{\Af}$.
      Note, however, that the change of $\Xf_{\emptyset}$ will cause a~modification to $\CountDict^\star_0$ since the rank-$p$ type of $\Xf_{\emptyset}$ might change.
      Thus, in order to process the query, we:
      \begin{enumerate}
        \item Construct $\Xf_{\emptyset}$ from the dynamic sets $R^{\Af}$ for each $R \in \Sigma^{(0)}$.
        \item Compute the rank-$p$ type $\alpha_{\mathrm{old}} \coloneqq \tp^p(\Xf_{\emptyset})$.
        \item Decrease $\CountDict^{\star}_0(\varepsilon, \alpha_{\mathrm{old}})$ by $1$.
        \item Perform an~update to $R^{\Af}$ prescribed by the query.
        \item Construct $\Xf'_{\emptyset}$ from the modified family of dynamic sets $R^{\Af}$.
        \item For $\alpha_{\mathrm{new}} \coloneqq \tp^p(\Xf'_{\emptyset})$, increase $\CountDict^{\star}_0(\varepsilon, \alpha_{\mathrm{new}})$ by $1$.
      \end{enumerate}
      
      \item If $\tup{b} = v$ or $\tup{b} = (v, v)$ for some $v \in V\left(\quo{H}{\Fc}\right)$, then the state of $\tup{b}$ is stored in the implicitly stored singleton element, so this case can be processed analogously to the previous case; only that $\CountDict^{\star}_1(v, \cdot)$ is affected by the query instead of $\CountDict^{\star}_0(\varepsilon, \cdot)$.
      
      \item If $\tup{b} = v$ or $\tup{b} = (v, v)$ for some $v \in V(H) \setminus V\left(\quo{H}{\Fc}\right)$, then let $x \coloneqq \ColoredVertex(v)$.
      Since the interpretations of unary and binary relations on $v$ in $\Af$ are stored in vertex $x$ of $\Bf$, we need to call $\AddRel$ or $\DelRel$ on $\Delta_{\Tc}$ to reflect that change.
      This updates the prescribed relation $R$ on $x$ or $(x, x)$ in $\Bf$.
      By Claim~\ref{cl:addrel_delrel_update}, this can be done in logarithmic time, with at most $2$ updates to all $\CountDict^\star_j$.
      \item If $\tup{b} = (u, v)$ with $u \neq v$, then let $(x, y) \coloneqq \ColoredEdge(u, v)$.
      The query is resolved analogously to the previous case, only that we call $\AddRel$ or $\DelRel$ on the tuple $(x, y)$ instead.
    \end{itemize}
    It is apparent that these updates can be performed in each case in worst-case logarithmic time, causing at most $2$ updates to $\CountDict^\star$.

    Now, consider a~query that adds or removes a~single vertex or edge.
    This query is immediately forwarded by us to $\F$, which updates $\Delta_{\Tc}$, $\name$, $\name^{-1}$, and $\edgerep$ in logarithmic time, performing a~constant number of updates to those dynamic structures, and updating the graph $\quo{H}{\Fc}$ by calling $\addquo$, $\delquo$, $\attach$, and $\detach$ a~constant number of times.
    We create two initially empty auxiliary lists $\RemovedList$ and $\RecolorList$.
    Then:
    
    \begin{itemize}
      \item Whenever any vertex $x$ is removed from $\Tc$ by $\F$, we verify if $x = \ColoredVertex(\name(x))$; that is, if $x$ is the colored representative of $\name(x)$ in $\Tc$.
      If not, then we are done.
      Otherwise, under this update, the interpretations of unary and binary predicates on $\name(x)$ are purged from $\Bf$ by $\F$, which requires that some other representative of $\name(x)$ inherit these interpretations instead.
      This fix is deferred to the end of the query; for now, we append $\name(x)$ at the end of $\RemovedList$.
      
      \item Similarly, whenever any edge $(x, y)$ is removed from $\Tc$ by $\F$, if the condition $(x, y) = \ColoredEdge((\name(x), \name(y)))$ holds, then we append $(\name(x), \name(y))$ at the end of $\RemovedList$.
      
      \item Whenever $\addquo(u)$ is called by $\F$, the information on the interpretations of unary and binary relations on $u$ need to be removed by us from $\Bf$.
      This, too, is delayed until the end of the query; for now, we append pair $(u, -)$ at the end of $\RecolorList$.
      
      We remark that at the moment of call to $\addquo(u)$, all trees containing a~representative of $u$ as a~vertex must have been detached by $\F$ (Lemma~\ref{lem:dynamic_ferns}, invariant~\ref{inv:modify_detached_vertex}).
      On the other hand, the types of components of $\Bf$ not containing a~representative remain unchanged; hence, no update in $\CountDict^\star_j$ is needed for the attached fern elements.
      
      However, $\addquo(u)$ spawns a~new, implicitly stored, singleton element $\Af_u \in \Xx$.
      Hence, it needs to be included in $\CountDict^\star_1$.
      Thus, observe that $\Af_u$ can be uniquely constructed by querying $\Oh[\Sigma]{1}$ dynamic sets $R^{\Af}$.
      Then, let $\alpha \coloneqq \tp^p(\Af_u)$ be the rank-$p$ type of $\Af_u$, and increase $\CountDict^\star_1(u, \iota_{\{u\}}(\alpha))$ by $1$.      
      
      \item Whenever $\delquo(u)$ is called by $\F$, the interpretations of unary and binary relations on $v$ must be restored to $\Bf$.
      To this end, we append the pair $(u, +)$ at the end of $\RecolorList$.
      Moreover, $\delquo(u)$ causes the singleton element $\Af_u$ to disappear from $\Xx$.
      This is accounted for in $\CountDict^\star_1$ analogously to $\addquo$.
      Again, by invariant~\ref{inv:modify_detached_vertex} of Lemma~\ref{lem:dynamic_ferns}, no other types of attached ferns change in $\Xx$.
      
      \item Whenever $\detach(\tup{a}, \Delta_T)$ is called by $\F$, we need to exclude the type of $\Delta_T$ from $\CountDict^\star$.
      By calling $\mathsf{topTreeToFernType}(\Delta_T)$ (Corollary~\ref{cor:top_tree_to_fern_type}), we get the type $\alpha$ of the boundaried structure $\Xf \in \Xx$ corresponding to $\Delta_T$, where $\alpha \in \Types^{p, \Sigma}(\bnd \Xf)$.
      Thus, in order to process the detachment of $\Delta_T$, we subtract $1$ from $\CountDict^\star_{|\bnd \Xf|}\left( \iota_{\bnd \Xf}(\alpha) \right)$.
      \item Whenever $\attach(\tup{a}, \Delta_T)$ is called by $\F$, we proceed analogously to $\detach$, only that instead of subtracting $1$ from some value in $\CountDict^\star$, we increase this value by $1$.
    \end{itemize}
    
    After the call to $\F$ concludes, the top trees data structure $\Delta_{\Tc}$ is updated, however $\F$ may have removed from $\Tc$ a~constant number of vertices and edges retaining the interpretations of relations in $\Af$, thus causing the need to fix some values of $\ColoredVertex$ and $\ColoredEdge$.
    Thus, for every vertex $v \in \RemovedList$, we find a~new representative of $v$ by querying any element $x$ of $\name^{-1}(v)$.
    If such an~element exists, we set $\ColoredVertex(v) \gets x$, and call $\AddRel(R, x)$ (resp. $\AddRel(R, (x, x))$) for each relation $R \in \Sigma^{(1)}$ (resp. $R \in \Sigma^{(2)}$) such that $v \in R^{\Af}$ (resp. $(v, v) \in R^{\Af}$).
    Note that at each call to $\AddRel$, Claim~\ref{cl:addrel_delrel_update} must be invoked so that the dictionaries $\CountDict^\star_j$ remain correct.
    Otherwise, if $\name^{-1}(v)$ is empty, we remove $v$ as a~key from $\ColoredVertex$.
    
    Then, $\ColoredEdge$ is repaired analogously: for each $(u, v) \in \RemovedList$, we find a~new representative of $(u, v)$ in $\edgerep$ and restore the binary relations as above.

    Finally, we process $\RecolorList$.
    For each pair of the form $(u, -)$, we take $x \coloneqq \ColoredVertex(u)$.
    Then, we call $\DelRel(R, x)$ for each predicate $R \in \Sigma^{(1)}$ such that $u \in R^{\Af}$; we then repeat the same procedure for binary predicates $R$ and pairs $(x, x)$.
    The pairs of the form $(u, +)$ are processed analogously, only that $\AddRel$ is called instead of $\DelRel$.
    Again, at each call to $\AddRel$ or $\DelRel$, Claim~\ref{cl:addrel_delrel_update} is invoked in order to preserve the correctness of $\CountDict^\star_j$.

    This concludes the description of the dynamic maintenance of $\Xx$ and $\CountDict^\star_j$.
    It can be easily verified that each operation to our data structure is performed in worst-case $\Oh[p, \Sigma]{\log n}$ time, and causes $\Oh[p, \Sigma]{1}$ updates to $\CountDict^\star_j$.
  \end{proof}
  \end{lemma}

    \paragraph*{Dynamic maintenance of $\Contract^p(\Xx)$.}
    As a~final part of the proof of Lemma~\ref{lem:dynamic_contraction}, we show how to extend the data structure $\C_0$ proposed in Lemma~\ref{lem:dynamic_ensemble} to maintain the rank-$p$ contraction of $\Xx$ dynamically in logarithmic time.
    At each query, our data structure will only change $\Oh[p, \Sigma]{1}$ vertices, edges, and relations in $\Contract^p(\Xx)$.
    This construction will yield the proof of Lemma~\ref{lem:dynamic_contraction}.
    
    \begin{proof}[Proof of Lemma~\ref{lem:dynamic_contraction}]
      Let $\C_0$ be the dynamic data structure shown in Lemma~\ref{lem:dynamic_ensemble} which maintains the functions $\CountDict_j$ dynamically.
    We now construct the prescribed data structure~$\C$.
      Each query on $\C$ will consist of three steps:
      \begin{itemize}
        \item Relay the query to $\C_0$.
        \item Intercept each modification of $\quo{H}{\Fc}$ of the form $\addquo(\cdot)$ or $\delquo(\cdot)$ made by $\F$; under each such modification, change the universe of $\Contract^p(\Xx)$.
        \item Intercept each change made to any values of the functions $\CountDict_j$ for $j \in \{0, 1, 2\}$; each such change will correspond to an~update to some relation in $\Contract^p(\Xx)$.
      \end{itemize}
      
      Thus, we begin handling each query by immediately forwarding the query to $\C_0$.

      Recall that vertices of $\Contract^p(\Xx)$ and vertices of $\quo{H}{\Fc}$ stay in a~natural bijection.
      Hence, whenever $\quo{H}{\Fc}$ is modified by $\addquo(u)$, we add a~new element $u$ to $\Contract^p(\Xx)$; similarly, under $\delquo(u)$, we remove $u$ from $\Contract^p(\Xx)$.

      It only remains to show how to update the relations efficiently under the modifications of $\CountDict_j$.
      To this end, we exploit the idempotence of types.
      For $j \in \{0, 1, 2\}$, let $m_j$ be a~constant (depending on $p$ and $\Sigma$) such that for all $a, b \in \N$ with $a, b \geq m_j$ and $a \equiv b \mod{m_j}$ and every type $\alpha \in \Types^p([j])$, we have
      \[ \underbrace{\alpha \oplus^{p,\Sigma}_{[j],[j]} \alpha  \oplus^{p,\Sigma}_{[j],[j]} \cdots \oplus^{p,\Sigma}_{[j],[j]} \alpha}_{a\textrm{ times}}=\underbrace{\alpha \oplus^{p,\Sigma}_{[j],[j]} \alpha  \oplus^{p,\Sigma}_{[j],[j]} \cdots \oplus^{p,\Sigma}_{[j],[j]} \alpha}_{b\textrm{ times}}. \]
      The existence and computability of $m_0$, $m_1$, $m_2$ is asserted by Lemma~\ref{lem:idempotence_of_types}.
      For each $j \in \{0, 1, 2\}$, let also $f_j : \N \to \N$ be a~function defined as follows:
      \[
        f_j(a) = \begin{cases}
          a & \text{if } a < m_j, \\
          (a\ \mathrm{mod}\ m_j) + m_j & \text{if } a \geq m_j.
        \end{cases}
      \]
      Each $f_j$ is constructed so that its range is $\{0, 1, 2, \dots, 2m_j - 1\}$, and so that for every integer $a \in \N$ and every type $\alpha \in \Types^p([j])$, the $a$-fold join of $\alpha$ with itself is equal to the $f_j(a)$-fold join of $\alpha$ with itself.
      Naturally, each $f_j$ can be evaluated on any argument in $\Oh[p,\Sigma]{1}$ time.
      
      Let $D$ be the universe of $\Contract^p(\Xx)$.
      We have:
      \begin{claim}
        \label{cl:efficient_contraction_check}
        Let $j \in \{0, 1, 2\}$ and $\alpha \in \Types^{p, \Sigma}([j])$.
        The interpretation of $\alpha$ in $\Contract^p(\Xx)$ contains a~tuple $\tup{a} \in D^j$ if and only if:
        \begin{itemize}
          \item $\tup{a}$ is ordered by $\leq$ and its elements are pairwise different;
          \item there exists at least one $\beta \in \Types^{p, \Sigma}([j])$ such that $\CountDict_j(\tup{a}, \beta) > 0$; and
          \item the following is satisfied:
          \begin{equation}
          \label{eq:contract_type_by_count}
      \alpha =
      \bigoplus{}^{p,\Sigma}_{[j],[j]} \left\{
        \underbrace{\beta \oplus^{p,\Sigma}_{[j],[j]} \beta \oplus^{p,\Sigma}_{[j],[j]} \cdots \oplus^{p,\Sigma}_{[j],[j]} \beta}_{f_j(\CountDict_j(\bar{a}, \beta))\text{ times}}
        \ \mid \ \beta \in \Types^{p, \Sigma}([j]),\, \CountDict_j(\bar{a}, \beta) > 0
      \right\}.
          \end{equation}
        \end{itemize}

        \begin{proof}
          Recall from the definition of $\Contract^p(\Xx)$ that the interpretation of $\alpha$ contains $\tup{a}$ if and only if:
          \begin{itemize}
             \item $\tup a$ is ordered by $\leq$ and its elements are pairwise different;
             \item there exists at least one $\Gf\in \Xx$ such that $\bnd \Gf$ is equal to the set of entries of $\tup a$; and
             \item the rank-$p$ type of the join of all the $\Gf\in \Xx$ as above is equal to $\iota_{\tup a}^{-1}(\alpha)$.
          \end{itemize}
          The first conditions in the statement of the claim and the definition of the contraction are identical.
          Then, the second conditions in these are equivalent: each $\Gf \in \Xx$ such that $\bnd \Gf = \{u \colon u \in \tup{a}\}$ contributes exactly $1$ to $\CountDict_j(\bar{a}, \tp^p(\Gf))$.
          For the third condition, observe that (\ref{eq:contract_type_by_count}) is equivalent to
          \[
          \alpha =
      \bigoplus{}^{p,\Sigma}_{[j],[j]} \left\{
        \underbrace{\beta \oplus^{p,\Sigma}_{[j],[j]} \beta \oplus^{p,\Sigma}_{[j],[j]} \cdots \oplus^{p,\Sigma}_{[j],[j]} \beta}_{\CountDict_j(\bar{a}, \beta)\text{ times}}
        \ \mid \ \beta \in \Types^{p, \Sigma}([j]),\, \CountDict_j(\bar{a}, \beta) > 0
      \right\},
          \]
          which, by associativity and commutativity of $\oplus$, is equivalent to
          \[
          \alpha = \bigoplus{}^{p,\Sigma}_{[j],[j]} \left\{ \iota_{\tup{a}}\left(\tp^p(\Gf)\right)\ \mid\ \Gf \in \Xx,\, \bnd \Gf = \{u \mid u \in \tup{a}\} \right\}.
          \]
          Applying the commutativity of $\oplus$ with $\tp^p$ and $\iota_{\tup{a}}$, we get equivalently
          \[
          \alpha = \iota_{\tup{a}} \left( \tp^p \left( \bigoplus{} \{ \Gf\, \mid\, \Gf \in \Xx,\, \bnd \Gf = \{u \mid u \in \tup{a}\} \} \right) \right).
          \]
          The equivalence with the third condition of the definition of the contraction follows immediately.
        \cqed\end{proof}
      \end{claim}

      From Claim~\ref{cl:efficient_contraction_check}, we conclude that:
      \begin{claim}
        \label{cl:efficient_contraction_get}
        Given $j \in \{0, 1, 2\}$ and a~tuple $\tup{a} \in D^j$, assume that $\CountDict_j(\tup{a}, \beta) > 0$ for some $\beta \in \Types^{p, \Sigma}([j])$.
        Then, there exists exactly one predicate $\alpha \in \Types^{p, \Sigma}([j])$ of $\Gamma^p$ that contains $\tup{a}$ in its interpretation.
        Moreover, it can be found in $\Oh[p, \Sigma]{\log |H|}$ worst-case time.
        \begin{proof}
          The uniqueness of $\alpha$ follows immediately from (\ref{eq:contract_type_by_count}).
          In order to compute $\alpha$ efficiently, we first query the value of $\CountDict_j(\tup{a}, \beta)$ for each $\beta \in \Types^{p, \Sigma}([j])$.
          This requires $\Oh[p, \Sigma]{\log |H|}$ time in total.

          Then, we compute $\alpha$ from (\ref{eq:contract_type_by_count}).
          The time complexity of this operation is dominated by the evaluations of $\oplus^{p,\Sigma}_{[j],[j]}$.
          Since $f_j$ only returns values smaller than $2m_j$, the join is evaluated no more than $2m_j \cdot |\Types^{p, \Sigma}([j])| = \Oh[p, \Sigma]{1}$ times; thus, (\ref{eq:contract_type_by_count}) can be computed in worst-case $\Oh[p, \Sigma]{1}$ time.
        \cqed\end{proof}
      \end{claim}
      
      Now, we show how to maintain $\Contract^p(\Xx)$ dynamically under the changes of $\CountDict_j$.
      Assume that for some $j \in \{0, 1, 2\}$, $\tup{a} \in D^j$, and $\beta \in \Types^{p, \Sigma}([j])$, the value of $\CountDict_j(\tup{a}, \beta)$ changed.
      We start processing this change by removing $\tup{a}$ from the interpretation of every predicate in $\Contract^p(\Xx)$.
      Then:
      \begin{itemize}
        \item If $\CountDict_j(\tup{a}, \gamma) = 0$ for each $\gamma \in \Types^{p, \Sigma}([j])$, then after the change, $\tup{a}$ will belong to the interpretations of no predicates in the contraction, and we are finished.
        Note that the condition can be verified in worst-case $\Oh[p, \Sigma]{\log |H|}$ time by querying $\CountDict_j$ for each $\gamma$, and checking if any value comes out positive.
        \item Otherwise, we compute the predicate $\alpha$ in $\Oh[p, \Sigma]{\log |H|}$ time using Claim~\ref{cl:efficient_contraction_get}, and we add $\tup{a}$ to the interpretation of $\alpha$ in $\Contract^p(\Xx)$.
      \end{itemize}
      Also, observe that the interpretations of tuples other than $\tup{a}$ will not change under this update.
      Thus, we processed a~single change of $\CountDict_j$ in $\Oh[p, \Sigma]{\log |H|}$ time, performing $\Oh[p, \Sigma]{1}$ updates to $\Contract^p(\Xx)$.
      
      By Lemma~\ref{lem:dynamic_ensemble}, any query to $\C$ (and hence to $\C_0$) causes at most $\Oh[p, \Sigma]{1}$ recalculations of any value of $\CountDict_j$.
      Summing up, we conclude that a~query to $\C$ can be processed in $\Oh[p, \Sigma]{\log |H|}$ worst-case time and causes $\Oh[p, \Sigma]{1}$ updates to the rank-$p$ contraction of $\Xx$.
    \end{proof}

As an end note, we remark that in Lemma~\ref{lem:dynamic_contraction}, the bound on the number of updates made to $\Contract^p(\Xx)$ under each query to $\C$ can be improved to a~universal constant, independent on $p$ or $\Sigma$.
This, however, requires a~more involved analysis of the data structures presented in the proof of the lemma, and does not improve any time complexity bounds presented in this work.
Thus, for the ease of exposition, we have chosen to present a~slightly looser $\Oh[p, \Sigma]{1}$ bound.

\subsection{Conclusion of the proof}
\label{ssec:lemmaA_conclusion}
We now have all necessary tools to finish the proof of the \lemmaA{}.

\begin{proof}[Proof of the \lemmaA{} (Lemma~\ref{lem:lemma_a})]
  We are given an~integer $k \in \N$, a~binary relational signature $\Sigma$, and a~sentence $\varphi \in \msotwo{}[\Sigma]$.
  In the proof of the static variant of Lemma~\ref{lem:lemma_a}, we computed a~new binary relational structure $\Gamma$, a~mapping $\Contract$ from augmented $\Sigma$-structures to augmented $\Gamma$-structures, and a~sentence $\psi \in \msotwo{}[\Gamma]$, with the properties prescribed by the statement of the lemma.
  Recall also that $\Contract((\Af, H))$ was defined as $(\Contract^p(\Xx), \quo{H}{\Fc})$, where $\Fc$ is the fern decomposition of $H$ constructed in Lemma~\ref{lem:fern-decomposition}, and $\Xx$ is the ensemble constructed from $\Af$ and $\Fc$ in Section~\ref{ssec:lemmaA_static}.
  
  Assume that we are given an~efficient dynamic $(\Cc^\star_k, \Gamma, \psi)$-structure $\D^\star$.
  Our aim is to construct an~efficient dynamic $(\Cc_k, \Sigma, \varphi)$-structure.
  To that end, take the dynamic data structure $\C$ constructed in Lemma~\ref{lem:dynamic_contraction}.
  Recall that $\C$, given a~dynamic augmented $\Sigma$-structure $(\Af, H)$, changing under the additions or removals of vertices, edges, and relations, maintains $\Contract^p(\Xx)$.
  Each query to $\C$ is processed in worst-case $\Oh[\varphi, \Sigma]{\log |H|}$ time, producing $\Oh[\varphi, \Sigma]{1}$ updates to $\Contract^p(\Xx)$.
  Let also $\F$ be the dynamic data structure constructed in Lemma~\ref{lem:dynamic_ferns}, which is also given $(\Af, H)$ dynamically, but produces $\quo{H}{\Fc}$ instead.
  We recall that $\Contract^p(\Xx)$ is guarded by $\quo{H}{\Fc}$.
  
  Now, we construct an~efficient dynamic $(\Cc_k, \Sigma, \varphi)$-structure $\D$ by spawning one instance of each of the structures: $\D^\star$, $\C$, and $\F$.
  Assume that a~batch $X$ of operations arrives to our structure.
  This batch will be converted into a~single batch $X'$ of operations supplied to $\D^\star$.
  We process the updates in $X$ one by one.
  Each update is immediately forwarded to both $\C$ and $\F$; then, after $\Oh[\varphi, \Sigma]{\log |H|}$ time, $\C$  produces a~sequence $L_\C$ of $\Oh[\varphi, \Sigma]{1}$ updates to $\Contract^p(\Xx)$, and $\F$ produces a~sequence $L_\D$ of $\Oh{1}$ updates to $\quo{H}{\Fc}$.
  We then filter $L_\C$ and $L_\D$ so that:
  \begin{itemize}
    \item in $L_\C$, everything apart from changes of the interpretations of relations is filtered out, and no tuple is both added to and removed from the same relation; and
    \item in $L_\D$, no edge or vertex is both added to and removed from the quotient graph $\quo{H}{\Fc}$.
  \end{itemize}
  
  We now add the updates from $L_\C$ and $L_\D$ to the batch $X'$ in the following order:
  \begin{itemize}[nosep]
    \item All relation removals in $L_\D$.
    \item All edge removals in $L_\C$.
    \item All vertex removals in $L_\C$.
    \item All vertex additions in $L_\C$.
    \item All edge additions in $L_\C$.
    \item All relation additions in $L_\D$.
  \end{itemize}
  By applying the updates in this specific order, we ensure that at each point of time, if two  distinct vertices $u$, $v$ are bound by a~relation, then $uv$ is an~edge of the guarding multigraph; and that if there exists an~edge incident to a~vertex $v$ of the multigraph, then $v$ is a~vertex of the multigraph.
  
  Let $(\Af, H)$ be the augmented $\Sigma$-structure stored in $\D$ after processing $X$.
  By definition, the multigraph $H$ belongs to $\Cc_k$, i.e., $\fvs{H} \leq k$.
  Let now $(\Af^\star, H^\star)$ be the augmented $\Gamma$-structure stored in $\D^\star$ after processing the batch $X'$.
  Since $(\Af^\star, H^\star) = \Contract((\Af, H))$, we have that $H \in \Cc^\star_k$.
  Thus, $\Af \models \varphi$ if and only if $\Af^\star \models \psi$, so the verification whether $\varphi$ is satisfied in $\Af$ is reduced to the verification whether $\psi$ is satisfied in $\Af^\star$.
  The data structure is thus correct.
  
  It remains to show that $\D$ is efficient.
  Indeed, processing a~single query from a~single batch takes $\Oh[\varphi, \Sigma]{\log |H|}$ worst-case time; and creates $\Oh[\varphi, \Sigma]{1}$ queries that are to be forwarded to~$\D^\star$.
  Since $\D^\star$ processes each query in worst-case $\Oh[\varphi, \Sigma]{\log |H^\star|}$ time, and $|H^\star| \leq |H|$, we conclude that $\D$ processes each query in worst-case $\Oh[\varphi, \Sigma]{\log |H|}$ time.
  Also, $\D$ can be straightforwardly initialized with $(\A, H)$ in $\Oh[\varphi, \Sigma]{|H| \log |H|}$ time by starting with an~empty augmented $\Sigma$-structure and adding the required features (vertices, edges and relations) to the $\Sigma$-structure one by one.
\end{proof}

\section{\lemmaB}%
\label{sec:lemmaB}

\newcommand{\mparacopy}[1]{\medskip\noindent\underline{#1}\hspace{0.2cm}}
\newcommand{\batch}{\mathcal{T}}

The aim of this section is to establish the~\lemmaB (Lemma~\ref{lem:lemma_b}).

Let $k \in \N_+$ be an~integer.
Throughout this section, we say that a~subset $B \subseteq V(H)$ of vertices of a~multigraph~$H$ is \emph{$k$-significant} if $|B| \leq 12k$, and $B$ contains all vertices of $H$ of degree at least $|E(H)|/(3k)$.
We start with showing two properties of $k$-significant subsets.
The first one (Corollary~\ref{cor:significant}) says that any feedback vertex of size~$k$ in~$H$ must contain a~vertex of a~given $k$-significant subset~$B$.
Similar statements were commonly used to design branching algorithms for $\probFVS$, see Lemma~\ref{lem:high} in the Overview for relevant discussion.
The second fact (Lemma~\ref{lem:remains_significant}) says that if $B$ is a~sufficiently large $k$-significant subset of $H$, then $B$ remains $k$-significant after applying $\Theta(|E(H)| / k)$ edge updates to $H$.
This will be crucial in the correctness proof of the weak efficient data structure claimed by Lemma~\ref{lem:lemma_b}.
We note that the same technique was used in the work of Alman et al.~\cite{AlmanMW20}.

\begin{lemma}[\cite{Iwata17, platypus}]
  \label{lem:fvs_degree}
  Let $H \in \Cc^\star_k$ be a~multigraph with $m$~edges, and let $S \subseteq V(H)$ be a~feedback vertex set of $H$ of size at most~$k$.
  Then $S$ must contain a~vertex of degree at least $m / (3k)$.
\end{lemma}

\begin{proof}
  Denote by $d$ the maximum degree of a~vertex in~$S$.
  Then, there are at most $d|S|$~edges incident to~$S$.
  Since the graph $H - S$ is a~forest, we have
  \[
    m \leq (|V(H)| - |S| - 1) + d|S|  < |V(H)| + d|S|.
  \]
  Moreover, $H$ is of minimum degree~$3$, and thus we obtain
  \[
    2m \geq 3|V(H)| > 3(m - d|S|),
  \]
  which implies that $d > m/(3|S|) \geq m / (3k)$.
\end{proof}

\begin{corollary}
  \label{cor:significant}
  If $H \in \Cc^\star_k$ is a~multigraph, and $B \subseteq V(H)$ is a~$k$-significant subset of vertices, then $\fvs{H - B} \leq k - 1$, that is, $H-B \in \Cc_{k-1}$.
\end{corollary}

\begin{lemma}[\cite{AlmanMW20}]
  \label{lem:remains_significant}
  Let $H$ be a~multigraph with $m$~edges, and let~$B$ be the~set of all vertices of~$H$ with degree at least $m / (6k)$.
  Let~$H'$ be a~multigraph obtained after applying at most $\Delta(m) := \floor{m / (6k+2)}$ edge modifications to~$H$ (i.e. edge insertions and deletions, with possible introduction of new vertices).
  Then $B$ is a~$k$-significant subset of vertices of~$H'$.
\end{lemma}

\begin{proof}
  First, by the handshaking lemma there are at most $12k$ vertices of~$H$ with degree at least~$m / (6k)$.
  Let us take a~vertex $v \in V(H') \setminus B$.
  It remains to show that the degree of $v$ in $H'$ is smaller than $|E(H')|/(3k)$.
  If $v \in V(H)$, then by the choice of $B$ ($v \not\in B$), the degree of $v$ in~$H$ must be smaller than $m / (6k)$.
  Hence, the~degree of~$v$ in~$H'$ is smaller than
  \[
    \frac{m}{6k} + \Delta(m) \leq \frac{m - \Delta(m)}{3k} \leq \frac{|E(H')|}{3k},
  \]
  where the first inequality is equivalent to $\Delta(m) \leq m / (6k+2)$.
  If $v \not\in V(H)$, then its degree in $H'$ is at most $\Delta(m) < |E(H')|/(3k)$.
  We obtained that $|B| \leq 12k$, and $B$ must contain all vertices of~$H'$ of degree at least $|E(H')| / (3k)$, and therefore $B$ is a~$k$-significant subset of vertices of~$H'$.
\end{proof}

Now, we may proceed with the proof of the~\lemmaB.
Let us fix an~integer $k \in \N_+$, a~binary signature~$\Sigma$, and a~sentence $\varphi \in \msotwo[\Sigma]$.
First, we provide static definitions of a~signature~$\Gamma$, a~mapping~$\Downgrade$, and a~formula~$\psi$ satisfying the requirements of \lemmaB.
Then, we show how to maintain the mapping $\Downgrade$ when the input augmented structure $(\A, H)$ is modified dynamically.

\paragraph*{Signature $\Gamma$.}
  Let $k_B := 12k$.
  Given a~$\Sigma$-structure~$\A$, we are going to find a~set $B \subseteq V(\A)$ of at most $k_B$~vertices and remove it from~$\A$.
  To be able to do so, we need to extend the~signature~$\Sigma$ with predicates which allow us to encode relations satisfied by vertices from~$B$.
  We will label $B$ with a~subset of~$[k_B]$.
  We use a~predicate $\textsf{vertex\_exists}_b$ to indicate whether a~vertex of identifier $b$ exists in $\A$, and predicates $\textsf{vertex\_color}_{R,b}$ to save all unary relations that vertices of~$B$ satisfy.
  To encode missing binary relations we introduce a~nullary predicate $\textsf{inner\_arc}_{R,b,c}$ (for tuples from $B \times B$), and two unary predicates: $\textsf{incoming\_arc}_{R,b}$ (for tuples from $(V(\A)\setminus B)\times B$) and $\textsf{outgoing\_arc}_{R,b}$ (for tuples from $B \times (V(\A)\setminus B)$).
  
  Summarizing, we set
  \begin{alignat*}{2}
    \Gamma := \Sigma & \cup \{ \textsf{vertex\_exists}_b && \mid b \in [k_B]\} \\
                    & \cup \{ \textsf{vertex\_color}_{R,b} && \mid R \in \Sigma^{(1)}, b \in [k_B]\} \\
                    & \cup \{ \textsf{inner\_arc}_{R,b,c} && \mid R \in \Sigma^{(2)}, b,c \in [k_B] \} \\
                    & \cup \{ \textsf{incoming\_arc}_{R,b}(\cdot) && \mid R \in \Sigma^{(2)}, b \in [k_B] \} \\
                    & \cup \{ \textsf{outgoing\_arc}_{R,b}(\cdot) && \mid R \in \Sigma^{(2)}, b \in [k_B] \}.
  \end{alignat*}
Abusing the notation slightly, we will refer to~$B$ as a~subset of~$[k_B]$, when it is convenient.

\paragraph*{Mapping $\Downgrade$.}
  Let $\A$ be a~$\Sigma$-structure, and let $H \in \Cc^\star_k$ be a~multigraph guarding $\A$.
  We define the augmented structure $(\widetilde{\A}, \widetilde{H}) = \Downgrade(\A, H)$ as follows.

  Let $B \subseteq V(H)$ be a~$k$-significant subset of vertices of~$H$.
  We set $\widetilde{H} := H - B$.
  Clearly, $|\widetilde{H}| \leq |H|$, and by Corollary~\ref{cor:significant} we obtain that $\widetilde{H} \in \Cc_{k-1}$.

  Now, we need to define a~$\Gamma$-structure $\widetilde{\A}$ on the universe $V(\A) \setminus B$ so that $\widetilde{H}$ guards $\widetilde{\A}$.
  For every predicate $R \in \Sigma$, the~structure~$\widetilde{\A}$ inherits its interpretation from $\A$, that is, we set $R^{\widetilde{\A}} := R^{\A}|_{V(\widetilde{\A})^{\mathrm{ar}(R)}}$.
  For every predicate from $\Gamma \setminus \Sigma$ we define its interpretation in a~natural way.
  For example, for a~vertex $b \in B$ and a~binary predicate $R \in \Sigma^{(2)}$, we set
  \[
    \textsf{incoming\_arc}_{R,b}^{\widetilde{\A}} \coloneqq \{ x \in V(\A) \setminus B \mid (x, b) \in R^{\A}\}.
    \]

    We remark that the set~$B$ is chosen non-deterministically in the definition above.
    In case we need to use the above transformations for a~fixed $k$-significant subset~$B$, we will denote their results by~$\widetilde{H}^B$ and $\widetilde{\A}^B$.

\paragraph*{Sentence $\psi$.}
  We now construct the promised sentence $\psi \in \msotwo[\Gamma]$.
  \begin{claim}
    \label{claim:sentence-psi}
    There is a~sentence $\psi \in \msotwo[\Gamma]$ such that for every $\Sigma$-structure $\A$,
    \[
      \A \models \varphi \quad \text{if and only if}\quad \widetilde{\A} \models \psi.
    \]
    Moreover, $\psi$ is computable from $k$ and $\varphi$.
  \end{claim}

\begin{proof}
  We give only a~brief description of $\psi$ as this sentence is just a~syntactic modification of~$\varphi$.
  First, we can assume without loss of generality that the quantification over individual arcs takes the form $\exists_{f \in R}$ or $\forall_{f \in R}$ for some binary relation $R$, and the quantification over arc subsets takes the form $\exists_{F \subseteq R}$ or $\forall_{F \subseteq R}$.
  To obtain $\psi$ we transform sentence $\varphi$ recursively (in other words, we apply transformations by structural induction on $\varphi$).
  If we encounter a~boolean connective or a~negation in $\varphi$, we leave it unchanged.

  Assume that we encounter a~quantifier in $\varphi$.
  Without loss of generality, let it be a~universal quantifier $(\forall)$ as transforming existential quantifiers $(\exists)$ is analogous.
  
  \mparacopy{Quantifications over vertices.}
  \begin{itemize}
    \item Let $\varphi = \forall_x\ \varphi_1$, where $x$ is a~single vertex.
      We need to distinguish two cases: either $x \not\in B$ or $x \in B$.
      First condition is equivalent to $x \in V(\widetilde{\A})$, and thus we can quantify over single vertices $x'$ of~$\widetilde{\A}$ as we did in $\varphi$.
      In the second case, we observe that there are only $k_B = \Oh{k}$ candidates for the valuation of $x$.
      Hence, we can simply check all possibilities for $x$.
      Formally, we transform $\varphi$ to:
      \[
        \forall_{x'}\ \varphi_1\ \wedge\ \bigwedge_{b \in [k_B]} \left( \textsf{vertex\_exists}_b \implies \varphi_1[x \mapsto b]\right),
      \]
      where $\varphi_1[x \mapsto b]$ is defined as the formula $\varphi_1$ with all occurrences of $x$ substituted with $b$.
      Then, to obtain $\psi$, we recursively process $\varphi_1$ and $\varphi_1[x \mapsto b]$ in the formula above.
    \item Let $\varphi = \forall_X\ \varphi_1$, where $X$ is a~subset of vertices.
      For every $X$ we can consider its partition $X = X' \cup X_B$, where $X' \subseteq V(\widetilde{\A})$ and $X_B \subseteq B$.
      Then, we can list all possibilities for $X_B$ and quantify over all subsets $X'$ normally.
      Formally, we write
      \[
        \bigwedge_{X_B \subseteq [k_B]} \forall_{X'}\ \left( \left( \bigwedge_{b \in X_B} \textsf{vertex\_exists}_b \right) \implies \varphi_1[X \mapsto X' \cup X_B] \right),
      \]
      and transform formula $\varphi_1[X \mapsto X' \cup X_B]$ recursively.
  \end{itemize}
 
  \newcommand{\fouter}{F_\text{outer}}
  \newcommand{\finner}{F_\text{inner}}
  \newcommand{\fbound}{F_\text{cross}}
  \newcommand{\neighbors}{\textsf{neighborhood}^B}
  \mparacopy{Quantifications over arcs.}
  \begin{itemize}
    \item Let $\varphi = \forall_{f \in R}\ \varphi_1$, where $f$ is a~single arc.
      There are three possibilities for $f$: either $f$ is an \emph{outer} arc on $V(\widetilde{\A})^{2}$ (and we can quantify over such arcs as in $\varphi$), or it is an \emph{inner} arc on $B^2$ (and again, we can then list all possible valuations of $f$ and make a finite conjunction), or it is a~\emph{crossing} arc from the set $(B \times V(\widetilde{\A})) \cup (V(\widetilde{\A}) \times B)$.
      To consider all crossing arcs of the form $(x', b)$, where $x' \in V(\widetilde{\A})$ and $b$ is a~fixed vertex of~$B$, we quantify over all vertices of $V(\A)$ for which $\textsf{incoming\_arc}_{R,b}$ holds.
      Arcs of the form $(b, x')$ are considered analogously.
    \item Let $\varphi = \forall_{F \subseteq R}\ \varphi_1$, where $F$ is a~subset of arcs.
      We partition $F$ into $\fouter \cup \finner \cup \fbound$ (sets of outer, inner and crossing arcs of $F$, respectively).
      We can deal with sets $\fouter$ and $\finner$ similarly as in the case of subsets of vertices. 
      Recall that to encode elements from $\fbound$ we use unary predicates $\textsf{incoming\_arc}_{R,b}$ and $\textsf{outgoing\_arc}_{R,b}$.
      Hence, instead of quantifying over all subsets $\fbound$ of crossing arcs, we can quantify over all sequences $(X_1, \ldots, X_s)$ of subsets of vertices, where the vertices are grouped by the sets of their in-neighbors and out-neighbors in $B$ (thus, $s = 2^{2k_B}$).
      Formally we assign to each vertex $x \in V(\widetilde{\A})$ a~pair:
      \[
        \neighbors(x) = (\{b \in B \mid (x,b) \in R^\A\},\ \{b \in B \mid (b,x) \in R^\A\}),
      \]
      (there are $2^{2k_B}$ possible values of such a~pair), and we group the elements of $V(\widetilde{\A})$ by the same value of $\neighbors(\cdot)$.
  \end{itemize}

  Finally, let us see how we transform atomic formulas in our procedure.

  \mparacopy{Modular counting checks.}
  \begin{itemize}
    \item If $X$ is a~subset of vertices, and we need to check in $\varphi$ whether $|X| \equiv a \mod p$, we transform it to $|X'| \equiv a - |X_B| \bmod p$, assuming that $X$ was introduced to the transformed formula as $X' \cup X_B$.
    \item If $F$ is a~subset of edges, and we need to check whether $|F| \equiv a \bmod p$, we transform it as follows.
      Recall that after transformation $F$ was introduced as $\fouter \cup \finner \cup \fbound$.
      Denote $a' = |\fouter|$ and observe that the value $a_B = |\finner|$ is known.
      The edges of $\fbound$ are represented by subsets $X_i \subseteq V(\widetilde{\A})$ having the same neighborhood on $B$.
      Define
  \[
    \mathcal{A} \coloneqq \left\{(a', a_1, a_2, \ldots, a_s) \in \{0, 1, \dots, p - 1\}^{s+1}\, \mid \, a' + a_B + \sum_{i=1}^l a_i n_i \equiv a \bmod p \right\},
  \]
    where $n_i$ is the number of arcs between vertices of $X_i$ and $B$.
    Observe that $|\mathcal{A}| = f(k, \varphi)$, for some function $f$.
    We can verify the equality $|F| \equiv a \mod p$ by making a disjunction over all tuples $(a', a_1, \ldots, a_s) \in \mathcal{A}$, and for each checking whether   $|\fouter| \equiv a' \bmod p$ and for all $i = 1, \ldots, s$, $|X_i| \equiv a_i \bmod p$.
  \end{itemize}
  \mparacopy{Equality, membership, incidence tests and relational checks.}

  These are straightforward to transform. Each such an~atomic formula is either unchanged, or it is transformed to an~expression on $\Gamma \setminus \Sigma$, or it is evaluated directly, e.g. if we check whether some outer edge is equal to an~inner edge.
\end{proof}

\paragraph*{Weak efficient dynamic structure.}
  Let $\widetilde{\D}$ be an~efficient dynamic $(\Cc_{k-1}, \Gamma, \psi)$-structure.
  Our goal is to construct a~weak efficient dynamic $(\Cc^\star_k, \Sigma, \varphi)$-structure $\D$.  
  During its run $\D$ will maintain:
  \begin{itemize}
    \item the~$\Sigma$-structure $\A$ and the~multigraph $H \in \Cc^\star_k$ modified by the user (the elements of $\A$ and $H$ are stored in dynamic dictionaries so that we can access and modify them in time $O(\log |H|)$),
    \item a~$k$-significant subset of vertices $B \subseteq V(H)$ (labeled with a~subset of~$[k_B]$); and
    \item a~single instance of $\widetilde{\D}$ running on the~structure~$\widetilde{\A}^B$ and the~multigraph~$\widetilde{H}^B$, where $\wt{\A}^B$ and $\wt{H}^B$ are defined as in the construction of the mapping $\Downgrade$.
  \end{itemize}
  
  $\D$ will be a~weak efficient dynamic structure in the following sense: supposing that $\D$ is initialized with $(\A_0, H_0)$, throughout its life, $\D$ will only accept $\Delta(|H_0|) = \left\lfloor |H_0| / (6k + 2)\right\rfloor \in \Omega(|H_0|)$ updates (and hence $\Omega(|H_0|)$ batches of updates).
  In particular, the multigraph $H$ maintained by $\D$ will undergo at most $\Delta(|H_0|)$ edge modifications.
  
  We initialize the data structure with an~augmented $\Sigma$-structure $(\A_0, H_0)$ as follows: first we determine a~$k$-significant subset of vertices $B \subseteq V(H)$ in time $\Oh{|H|}$ by applying Lemma~\ref{lem:remains_significant}.
  Then we compute $(\widetilde{\A_0}, \widetilde{H_0}) = \Downgrade(\A, H)$ and initialize $\widetilde{\D}$ with $(\widetilde{\A_0}, \widetilde{H_0})$.
  This can be done in total time $\Oh[\varphi, k]{|H| \log |H|}$.
  Also, by construction and Lemma~\ref{lem:remains_significant}, it is guaranteed that throughout the entire life of the data structure, $B$ will remain a~$k$-significant subset of vertices of $H$.

  Let now $\batch$ be a~batch of operations that arrives to~$\D$.
  We start with updating $\A$ and $H$ according to the~operations from~$\batch$.
Next, we create a~batch of updates $\widetilde{\batch}$ for the~data structure~$\widetilde{\D}$ with operations from~$\batch$ translated as follows.
\begin{itemize}
  \item $\mathsf{addVertex}(v)$ / $\mathsf{delVertex}(v)$.
    If we want to remove a~vertex $v \in B$, then we append to the batch $\widetilde{\batch}$ an~operation $\mathsf{delRelation}(\textsf{vertex\_exists}_v)$.
    Otherwise, we copy a~given operation to~$\widetilde{\batch}$.
  \item $\mathsf{addEdge}(u, v)$ / $\mathsf{delEdge}(u, v)$.
    We append a~given operation to $\widetilde{\batch}$ provided that~$u, v \in V(\widetilde{H}^B)$, i.e.~$u, v \not\in B$.
  \item $\mathsf{addRelation}(R, \bar{a})$ / $\mathsf{delRelation}(R, \bar{a})$.
    If at least one of the elements in the tuple~$\bar{a}$ belongs to~$B$, we propagate this modification to~$\widetilde{\D}$ by adding or removing an~appropriate relation from $(\Gamma \setminus \Sigma)^{\widetilde{\A}}$.
    Otherwise, we append a~given operation to~$\widetilde{\batch}$.
\end{itemize}
Note that each operation from $\batch$ is translated to at most one operation in~$\widetilde{\batch}$.
After $\widetilde{\D}$ performs all~operations from~$\widetilde{\batch}$, we query~$\widetilde{\D}$ whether $\widetilde{\A}^B \models \psi$ holds which by Claim~\ref{claim:sentence-psi} is equivalent to verifying whether~$\A \models \varphi$ holds.
Of course, each update can be performed in worst-case $\Oh[\varphi, k]{\log |H|}$ time.
This concludes the proof of the Downgrade Lemma.

\begin{acks}
The authors thank the anonymous reviewers for their valuable input, especially for pointing us to the technique of global rebuilding by Overmars and van~Leeuwen~\cite{DBLP:conf/wg/Overmars81, DBLP:journals/ipl/OvermarsL81a} and to the paper of Bodlaender on the dynamic maintenance of treewidth-$2$ graphs~\cite{Bodlaender93a}.
\end{acks}

\bibliographystyle{ACM-Reference-Format}
\bibliography{references}

\appendix
\section{Proof of the Replacement Lemma}\label{sec:replacement-proof}

In this section we provide a proof of the Replacement Lemma (Lemma~\ref{lem:replacement_lemma}).
As mentioned, the proof is based on an application of Ehrenfeucht-Fra\"isse games (EF games, for short). Therefore, we need to first recall those games for the particular variant of the Monadic Second-Order logic that we work with -- $\msotwo$.

\paragraph*{Ehrenfeucht-Fra\"isse games.} Fix a binary signature $\Sigma$. The EF game is played on a pair of boundaried $\Sigma$-structures, $\Af$ and $\Bf$, both with the same boundary $D$. There are two players: Spoiler and Duplicator. Also, there is a parameter $q\in \N$, which is the length of the game. The game proceeds in $q$ rounds, where every round is as follows.

First, the Spoiler chooses either structure $\Af$ or $\Bf$ and makes a move in the chosen structure. There are four possible types of moves, corresponding to four possible types of quantification. These are:
\begin{itemize}
 \item Choose a vertex $u$.
 \item Choose an arc $f$.
 \item Choose a set of vertices $U$.
 \item Choose a set of arcs $F$.
\end{itemize}
In the third type, this can be any subset of vertices of the structure in which the Spoiler plays. In the fourth type, we require that the played set is a subset of all the arcs present in this structure. The first two types of moves are {\em{individual moves}}, and the last two types are {\em{monadic moves}}.

Then, the Duplicator needs to reply in the other structure: $\Bf$ or $\Af$, depending on whether Spoiler played in~$\Af$ or in $\Bf$. The Duplicator replies with a move of the same type as the Spoiler chose. 

Thus, every round results in selecting a {\em{matched pair}} of features, one from $\Af$ and second in~$\Bf$, and one chosen by Spoiler and second chosen by Duplicator. When denoting such matched pairs, we use the convention that the first coordinate is selected in $\Af$ and the second is selected in $\Bf$. The pairs come in four types, named naturally.

The game proceeds in this way for $q$ rounds and at the end, the play is evaluated to determine whether the Duplicator won. The winning condition for the Duplicator is that every atomic formula (with modulus at most $q$, if present) that involves the elements of $D$ and features introduced throughout the play holds in $\Af$ if and only if it holds in $\Bf$. More precisely,  it is the conjunction of the following checks:
\begin{itemize}
 \item $\Af$ and $\Bf$ satisfy the same nullary predicates.
 \item Every element of $D$ satisfies the same unary predicates in $\Af$ and in $\Bf$.
 \item For every matched pair of individual moves $(x,y)$, $x$ and $y$ satisfy the same predicates (unary or binary, depending on whether $x$ is a vertex or an arc). 
 \item For every matched pair of monadic moves $(X,Y)$, $X\cap D=Y\cap D$ and $|X|\equiv |Y|\bmod q'$ for every $1\leq q'\leq q$.
 \item For every matched pair of individual moves $(x,y)$ and a matched pair of monadic moves $(X,Y)$, we have $x\in X$ if and only if $y\in Y$.
 \item For every matched pair of individual vertex moves $(u,v)$ and a matched pair of individual arc moves $(e,f)$, $u$ is the head/tail of $e$ if and only if $v$ is the head/tail of $f$.
 \item For every two matched pairs of individual moves $(x,y)$ and $(x',y')$, $x=x'$ if and only if $y=y'$.
\end{itemize}
It is straightforward to check that this model of an EF game exactly corresponds to our definition of $\msotwo$ and the rank, in the sense that the following statement binds equality of types with the existence of a winning condition for Duplicator. The standard proof is left to the reader (see e.g.~\cite[Chapter~6]{Immerman99} for an analogous proof for ${\mathsf{FO}}$).

\begin{lemma}\label{lem:EF}
 Fix $q\in \N$.
 For every pair of boundaried $\Sigma$-structures $\Af$ and $\Bf$, where $\bnd \Af=\bnd \Bf$, we have $\tp^q(\Af)=\tp^q(\Bf)$ if and only if the Duplicator has a winning strategy in the $q$-round EF game on $\Af$ and $\Bf$.
\end{lemma}

\paragraph*{Setup.} We now prove the existential part of the Replacement Lemma, that is, the existence of the function $\Infer$ for a large enough constant $p$, depending on $\Sigma$ and~$q$. We will later argue that both $p$ and $\Infer$ can be computed from $\Sigma$ and $q$.

To prove the existence of a suitable mapping $\Infer$, it suffices to argue that for a large enough constant $p$, we have the following assertion: if $\Xx$ and $\Yy$ are $\Sigma$-ensembles then
\begin{equation}\label{eq:types-EF}
\tp^p(\Contract^p(\Xx))=\tp^p(\Contract^p(\Yy))\quad\textrm{implies}\quad\tp^q(\Smash(\Xx))=\tp^q(\Smash(\Yy)).
\end{equation}
Let 
\begin{align*}
& \Af\coloneqq \Smash(\Xx), & & \Bf\coloneqq \Smash(\Yy),\\
& \wh{\Af}\coloneqq \Contract^p(\Xx), & & \wh{\Bf}\coloneqq \Contract^p(\Yy),
\end{align*}
where $p$ will be chosen later. By Lemma~\ref{lem:EF}, to prove~\eqref{eq:types-EF} it suffices to argue that if the Duplicator has a winning strategy in the $p$-round EF game on $\wh{\Af}$ and $\wh{\Bf}$, then she also has a winning strategy in the $q$-round EF game on $\Af$ and $\Bf$, provided we choose $p$ large enough.

\newcommand{\Kk}{{\cal K}}
\newcommand{\Ll}{{\cal L}}
\newcommand{\Gm}{{\frak G}}

Let us introduce some notation and simplifying assumptions. By $A(\Af)$ we denote the set of arcs in $\Af$, that is, all pairs $(u,v)\in V(\Af)^2$ that appear in any relation in $\Af$; similarly for other structures.
Denote
\begin{align*}
& C\coloneqq \bigcup_{\Gf\in \Xx} \bnd \Gf=V(\wh{\Af}), & & D\coloneqq \bigcup_{\Hf\in \Yy} \bnd \Hf=V(\wh{\Bf}).
\end{align*}
Further, let $\Kk$ be the set of all subsets $X\subseteq C$ for which there exists $\Gf\in \Xx$ with $X=\bnd \Gf$; note that $\Kk$ consists of sets of size at most $2$. Define $\Ll$ for the ensemble $\Yy$ analogously.

We assume without loss of generality the following assertion: for every $X\in \Kk$ there exists exactly one element $\Gf_X\in \Xx$ with $\bnd \Gf_X=X$. Indeed, if there are multiple such elements, then we can replace them in $\Xx$ with their join; this changes neither $\Smash(\Xx)$ nor $\Contract^p(\Xx)$. We also make the analogous assertion about the elements of $\Ll$ and the ensemble $\Yy$. We extend notation $\Gf_X$ to allow single vertices and arcs in the subscript, treating them in this case as unordered sets of vertices.
Also, by adding to $\Xx$ and $\Yy$ trivial one-element structures with empty relations, we may assume that $\Kk$ contains all singleton sets $\{u\}$ for $u\in C$, and the same for $\Ll$. Similarly, by adding empty structures to $\Xx$ and $\Yy$ if necessary, we may assume that $\emptyset\in \Kk$ and $\emptyset\in \Ll$. As a side note, these assertions actually do hold without any modifications needed in all applications of the Replacement Lemma in this paper.

For convenience of description we define a mapping $\xi$ that maps features (vertices and arcs) in $\Af$ to features in $\wh{\Af}$ as follows. Consider any vertex $u\in V(\Af)$. If $u\in C$, then $\xi(u)\coloneqq u$. Otherwise, there exists a unique $\Gf\in \Xx$ that contains $u$. Let $X=\bnd \Gf$; by the assumption of previous paragraph, $X\in \Kk$ and $\Gf=\Gf_X$. Then we set $\xi(u)$ to be $X$ ordered naturally (by the ordering on $\Omega=\N$). Note that thus, $\xi(u)$ is either an arc, or a vertex, or the empty tuple. Similarly, for any arc $f\in A(\Af)$, we take the unique $\Gf\in \Xx$ that contains $f$ and set $\xi(f)$ to be $\bnd \Gf$ ordered naturally. The mapping is extended to features of $\Bf$ (mapped to features in $\wh{\Bf}$) as expected.


Finally, we define the extended signature $\wt{\Sigma}$ by adding to $\Sigma$ $q$ fresh unary predicates $X_1,\ldots,X_q$ and $q$ fresh binary predicates $F_1,\ldots,F_q$. We set
$$p\coloneqq q\cdot (2|\Types^{q+1,\wt{\Sigma}}([2q+2])|+1)+1.$$
Also, we let $\wt{\Sigma}_i$ be the subset of $\wt{\Sigma}$ where only predicates $X_1,\ldots,X_i$ and $F_1,\ldots,F_i$ are added.

\paragraph*{Designing the strategy: general principles.}
Let $\Gm$ be the $q$-round EF game played on $\Af$ and $\Bf$, and $\wh{\Gm}$ be the $p$-round EF game played on $\wh{\Af}$ and $\wh{\Bf}$. Our goal is to design a strategy for Duplicator in the game $\Gm$, provided she has a strategy in the game $\wh{\Gm}$. 

While playing $\Gm$, the Duplicator simulates in her head a play in the game $\wh{\Gm}$; the choice of moves in the latter game will guide the choice of moves in the former. More precisely, when choosing a move in $\Gm$, the Duplicator will always apply the following general principle.
\begin{itemize}[nosep]
 \item Suppose Spoiler makes some move $m$ in $\Gm$.
 \item The Duplicator translates $m$ into a batch $M$ of moves in the game $\wh{\Gm}$.
 \item In her simulation of $\wh{\Gm}$, the Duplicator executes the batch $M$ and obtains a batch of her responses $N$ (in $\wh{\Gm}$).
 \item The batch of responses $N$ is translated into a single move $n$ in $\Gm$, which is the Duplicator's answer to the Spoiler's move $m$.
\end{itemize}
The size of the batch $M$ will depend on the type of the move $m$, but it will be always the case that 
$$|M|\leq 2|\Types^{q+1,\wt{\Sigma}}([2q+2])|+1.$$
By the choice of $p$, this means that the total length of the simulated game $\wh{\Gm}$ will never exceed~$p$. So by assumption, the Duplicator has a winning strategy in $\wh{\Gm}$.

It will be (almost) always the case that individual moves in $\Gm$ are translated to individual moves in $\wh{\Gm}$, that is, $|M|=1$ whenever $m$ is an individual move. More precisely, the Duplicator will maintain the invariant that if $(m,n)$ is a matched pair of individual moves in $\Gm$, then this pair is simulated by a matched pair of individual moves $(\xi(m),\xi(n))$. (There will be a corner case when $\xi(m)=\xi(n)=\emptyset$, in which case $M=\emptyset$, that is, $m$ is not simulated by any move in~$\wh{\Gm}$.)

Furthermore, during the game, the Duplicator will maintain the following {\em{Invariant $(\star)$}}.
Suppose $i$ moves have already been made in $\Gm$. 
Call a vertex $u\in C$ {\em{similar}} to a vertex $v\in D$ if the following conditions hold:
\begin{itemize}
 \item $u$ and $v$ satisfy the same unary predicates in $\wh{\Af}$ and $\wh{\Bf}$, respectively;
 \item for each matched pair $(L,R)$ of monadic vertex subset moves played before in $\wh{\Gm}$, $u\in L$ if and only if $v\in R$; and
 \item for each matched pair $(\ell,r)$ of individual vertex moves played before in $\wh{\Gm}$, $u=\ell$ if and only if $v=r$.
\end{itemize}
Similarity between arcs of $\wh{\Af}$ and arcs of $\wh{\Bf}$ is defined analogously. Next, for any $\Gf\in \Xx$ we define its {\em{$i$-snapshot}}, which is a $\wt{\Sigma}_i$-structure $\wt{\Gf}$ obtained from $\Gf$ as follows:
 \begin{itemize}
 \item Add the restrictions of previously made monadic moves to the vertex/arc set of $\Gf$, using predicates $X_1,\ldots,X_i$ and $F_1,\ldots,F_i$. (Predicate with subscript $j$ is used for a monadic move from round $j$.)
 \item Add all previously made individual moves to the boundary of $\Gf$, whenever the move was made on a feature present in $\Gf$. In case of individual arc moves, we add both endpoints of the arc.
 \item Reindex the boundary with $[2q+2]$ so that the original boundary of $\Gf$ is assigned indices in $[2]$ in the natural order, while all the vertices added to the boundary in the previous point get indices $3,4,\ldots$ in the order of moves (within every arc move, the head comes before the tail).
 \end{itemize}
Then Invariant $(\star)$ says the following (in all cases, we consider types over signature $\wt{\Sigma}_i$):
\begin{itemize}
 \item Whenever $e\in A(\wh{\Af})$ is similar to $f\in A(\wh{\Bf})$, we have
 $$\tp^{q+1-i}(\wt{\Gf}_e)=\tp^{q+1-i}(\wt{\Hf}_f).$$
 \item Whenever $u\in V(\wh{\Af})$ is similar to $v\in V(\wh{\Bf})$, we have
 $$\tp^{q+1-i}(\wt{\Gf}_u)=\tp^{q+1-i}(\wt{\Hf}_v).$$
 \item We have
 $$\tp^{q+1-i}(\wt{\Gf}_\emptyset)=\tp^{q+1-i}(\wt{\Hf}_\emptyset).$$
\end{itemize}
Observe that when $i=0$, that Invariant~$(\star)$ is satisfied follows directly from the construction of $\wh{\Af}$ and $\wh{\Bf}$ through the $\Contract^p$ operator and the fact that $p\geq q+1$. Further, it is easy to see that if the Duplicator maintains Invariant~($\star$) till the end of $\Gm$, she wins this game.

\paragraph*{Choice of moves.} We now proceed to the description of how the translation of a move $m$ in $\Gm$ to a batch of moves $M$ in $\wh{\Gm}$ and works, and similarly for the translation of $N$ to $n$. This will depend on the type of the Spoiler's move $m$. In every case, $N$ is obtained from $M$ by executing the moves of $M$ in $\wh{\Gm}$ in any order, and obtaining $N$ as the sequence of Duplicator's responses. We also consider the cases when the Spoiler plays in $\Af$, the cases when he plays in $\Bf$ are symmetric. Let $i$ be such that $i$ moves have already been made in $\Gm$, and now we consider move $i+1$.

\mpara{Individual vertex move.} Suppose that $m=u\in V(\wh{\Af})$, that is, the move $m$ is to play a vertex~$u$.

If $u\in C$, then we set $M\coloneqq \{u\}$, that is, the Duplicator simulates $m$ by a single move in $\Gm$ where the same vertex $u$ is played. Then $N\coloneqq \{v\}$ is the batch of Duplicator's responses in $\wh{\Gm}$, and accordingly $n=v$; that is, Duplicator also plays $v$ in $\Gm$. Note that thus $v\in D$.

If $u\notin C$, then there exists a unique $\Gf\in \Xx$ that contains $u$. Let $X\coloneqq \bnd \Gf$; then $X\in \Kk$ and $\Gf=\Gf_X$. We consider cases depending on the cardinality of $X$.

If $|X|=2$, then the elements of $X$ (after ordering w.r.t. the order on $\Omega=\N$) form an arc~$e$. We set $M\coloneqq \{e\}$, and let $N\coloneqq \{f\}$ be the batch of Duplicator's responses in $\wh{\Gm}$. Let $Y$ be the set of (two) endpoints of $f$. Then there is a unique $\Hf\in \Yy$ satisfying $Y=\bnd \Hf$.

As $f$ is Duplicator's response to the Spoiler's move $e$ in $\wh{\Gm}$, it must be the case that $e$ and $f$ are similar. From Invariant~($\star$) it then follows that if $\wt{\Gf}$ and $\wt{\Hf}$ are the $i$-snapshots of $\Gf$ and $\Hf$, respectively, then
\begin{equation}\label{eq:eqtp}
 \tp^{q+1-i}(\wt{\Gf})=
 \tp^{q+1-i}(\wt{\Hf})
\end{equation}
From~\eqref{eq:eqtp} and Lemma~\ref{lem:EF} it now follows that there that Duplicator has a response $v$ in $\Hf$ to Spoiler's move $u$ in $\Gf$ so that after playing $n\coloneqq v$ in game $\Gm$, the Invariant~$(\star)$ remains satisfied. So this is Duplicator's response in $\Gm$. 

The cases when $|X|=1$ and $|X|=0$ are handled analogously, here is a summary of differences.
\begin{itemize}
 \item If $|X|=1$, then $\bnd \Gf=\{c\}$ for a single vertex $c\in C$. We set $M=\{c\}$ and let $N=\{d\}$ be the Duplicator's response in $\wh{\Gm}$. Then $\Hf$ is the unique element of $\Yy$ with $\bnd \Hf=\{d\}$ and the existence of the response $v$ in $\Hf$ to the move $u$ in $\Gf$ follows from the similarity of $c$ and $d$ and Invariant $(\star)$.
 \item If $|X|=0$, then $\Gf$ and $\Hf$ are the unique elements of $\Xx$ and $\Yy$, respectively, with $\bnd \Gf=\bnd \Hf=\emptyset$. Then the existence of the response $v$ in $\Hf$ to the move $u$ in $\Gf$ follows from the last point of Invariant~$(\star)$.
\end{itemize}

\mpara{Individual arc move.} This case is handled completely analogously to the case of an individual vertex move, except that we do not have the case of a move within $C$. That is, if $m=a\in A(\Af)$, then there exists a unique $\Gf\in \Xx$ that contains $a$. Given $\Gf$ and the arc $a$ in it, we proceed exactly as in the case of an individual vertex move.

\mpara{Monadic arc subset move.} We now explain the case of a monadic arc subset move, which will be very similar, but a bit simpler than the case monadic vertex subset move. Suppose then that the Spoiler's move is $m=S\subseteq A(\Af)$.

For every $\Gf\in \Xx$, let $S_\Gf$ be the restriction of $S$ to the arcs of $\Gf$. Then $\{S_\Gf\colon \Gf\in \Xx\}$ is a partition of $S$. For each $\Gf\in \Xx$, let $\alpha_\Gf\in \Types^{\wt{\Sigma},q+1}([2q+2])$ be the rank-$(q+1)$ type of $(i+1)$-snapshot $\wt{\Gf}$ of $\Gf$, that is, a snapshot that {\em{includes}} the move $S$. Now comes the main technical idea of the proof: we set
\begin{eqnarray*}
M & \coloneqq & \{ \{a\in A(\wh{\Af})~|~\alpha_{\Gf_a}=\alpha\}\colon \alpha\in \Types^{\wt{\Sigma}_{i+1},q-i}([2q+2])\}\cup\\
& &\{ \{u\in V(\wh{\Af})~|~\alpha_{\Gf_u}=\alpha\}\colon \alpha\in \Types^{\wt{\Sigma}_{i+1},q-i}([2q+2])\}.
\end{eqnarray*}
Intuitively speaking, the Duplicator breaks $S$ into parts $S_\Gf$ contained in single elements $\Gf\in \Xx$, and observes how playing $S_\Gf$ in each $\Gf$ affects the type of $\wt{\Gf}$, in the sense of what is the resulting type. Then she models the move $S$ in $\Gm$ through a batch $M$ of monadic moves in~$\wh{\Gm}$: one monadic vertex subset move and one monadic arc subset move per each possible type. These moves respectively highlight the sets vertices and arcs of $\wh{\Af}$ where particular resulting types have been observed. Note that $|M|\leq 2|\Types^{\wt{\Sigma},q+1}([2q+2])|$, as promised.

In the simulated game $\wh{\Gm}$, the Duplicator applies the moves from $M$ in any order as Spoiler's moves, and gets in return a sequence $N$ of responses. Denoting the Spoiler's moves in $M$ as $F_\alpha\subseteq A(\wh{\Af})$ and $U_\alpha\subseteq V(\wh{\Af})$ for $\alpha\in \Types^{\wt{\Sigma}_{i+1},q-i}([2q+2])$ naturally, we let $F'_\alpha\subseteq A(\wh{\Bf})$ and $U'_\alpha\subseteq V(\wh{\Bf})$ be the respective responses in $N$. Note that since sets $F_\alpha$ form a partition of $A(\wh{\Af})$ and $U_\alpha$ form a partition of $V(\wh{\Af})$, the analogous must be also true for sets $F'_\alpha$ and $U'_\alpha$ in $\Bf$, for otherwise the Spoiler could win $\wh{\Gm}$ in one move.

We are left with describing how the Duplicator translates the batch of responses $N$ into a single response $n$ in the game $\Gm$ so that Invariant~$(\star)$ is maintained.
Consider any matched pair of moves $(F_\alpha,F'_\alpha)$ in $\wh{\Gm}$ as described above. From Invariant~$(\star)$ and the fact that $F'_\alpha$ was a response to move $F_\alpha$, we observe that for every arc $b\in F'_\alpha$, if $\Hf$ is the unique element of $\Yy$ with $b=\bnd \Hf$, then
$$\tp^{q+1-i}(\wt{\Hf})=\alpha',$$
where $\wt{\Hf}$ is the $i$-snapshot of $\Hf$ and $\alpha'$ is a type over the signature $\wt{\Sigma}_i$ that contains the sentence $\exists_{F_{i+1}} \bigwedge \alpha$. Consequently, within each $\Hf$ with $\bnd \Hf\in F'_\alpha$ the Duplicator may choose a set of arcs $S'_{\alpha,\Hf}$ so that after adding this set to $\wt{\Hf}$ using predicate $F_{i+1}$, the type of $\wt{\Hf}$ becomes $\alpha$.
We may apply an analogous construction to obtain sets $S'_{\alpha,\Hf}$ for all $\alpha\in \Types^{\wt{\Sigma}_{i+1},q+1}([2q+2])$ and $\Hf$ with $\bnd \Hf\in U'_\alpha$. 

Finally, we obtain the Duplicator's response $S'$ by summing the sets $S'_{\alpha,\Hf}$ throughout all $\alpha\in \Types^{\wt{\Sigma}_{i+1},q+1}([2q+2])$ and all $\Hf\in \Yy$ with $\bnd \Hf\in F_\alpha'\cup U_\alpha'$. It is easy to see that in this way, Invariant~$(\star)$ is maintained.

\mpara{Monadic vertex subset move.} Let the Spoiler's move be $S\subseteq V(\Af)$. The strategy for the Duplicator is constructed analogously to the case of a monadic arc subset move, except that we add to the batch $M$ one more move: the set $S\cap C$. The rest of the construction is the same; we leave the details to the reader.

\bigskip

We have shown how that the Duplicator can choose her moves so that we maintains the Invariant~$(\star)$ for $q$ rounds. Consequently, she wins the game $\Gm$ and we are done.

\paragraph*{Computability.} We have already argued the existence of $p$ and function $\Infer$ that is promised in the statement of the Replacement Lemma. We are left with arguing their computability. Clearly, $p$ is computable by Lemma~\ref{lem:formulas-finite}. As for the function $\Infer$, the reasoning follows standard arguments, see e.g.~\cite{Makowsky04} for a reasoning justifying the computability claim in Lemma~\ref{lem:compositionality}. So we only sketch the argument. Later we shall provide also a different argument that might be somewhat more transparent, but relies on the assumption that the Gaifman graphs of all structures involved have treewidth bounded by a given parameter $k$; this assumption is satisfied in our applications.

Let $\beta\in \Types^{\Gamma^p,p}$; we would like to construct $\alpha=\Infer(\beta)$. For this it suffices to decide, for any sentence $\varphi\in \Sentences^{q,\Sigma}$, whether $\varphi\in \alpha$. We translate $\varphi$ to a sentence $\wh{\varphi}$, of rank at most~$p$ and working over $\Gamma^p$-structures, such that $\Smash(\Xx)\models \varphi$ if and only if $\Contract^p(\Xx)\models \wh{\varphi}$. Then to decide whether $\varphi\in \alpha$ it suffices to check whether $\beta$ entails $\wh{\varphi}$.

As $\varphi$ has quantifier rank at most $q$, it is a boolean combination of sentences of the form $\exists_x \psi$ or $\exists_X \psi$, where $\psi$ has quantifier rank at most $q-1$ (here $x$ is an individual vertex/arc variable and $X$ is a monadic vertex/arc subset variable). So it suffices to apply the translation to each such sentence individually, and then take the same boolean combination. To this end, we translate the quantifier ($\exists_x$ or $\exists_X$) to a sequence of quantifiers in the same way as was done in the translation of moves from $\Gm$ to $\wh{\Gm}$:
\begin{itemize}[nosep]
 \item A quantifier $\exists_x$ is translated to a quantifier $\exists_{\wh{x}}$, where $\wh{x}$ is interpreted to be $\xi(x)$. This is followed by a disjunction over possible types that the ensemble element contracted to $\xi(x)$ might get after selecting $x$ in it. 
 \item A quantifier $\exists_X$ is translated to a sequence of existential quantifiers that quantify sets $F_\alpha$ and $U_\alpha$, where $\alpha$ ranges over suitable types as in the translation from $\Gm$ to $\wh{\Gm}$. We verify that these quantifiers select partitions of the vertex set and the arc set, and interpret $F_\alpha$ and $U_\alpha$ as sets of those arcs and vertices of the structure that correspond to those ensemble elements where the choice of $X$ refined the type to $\alpha$.
\end{itemize}
The translation is applied recursively to $\psi$. Once we arrive at atomic formulas, these can be translated to atomic formulas that check unary/binary predicates on variables $\wh{x}$ in question, and their relation to previously quantified sets $F_\alpha$, $U_\alpha$. Then one can argue that $\Smash(\Xx)\models \varphi$ if and only if $\Contract^p(\Xx)\models \wh{\varphi}$, as claimed.

\newcommand{\Rep}{\mathsf{Rep}}

There is another, somewhat simpler argument that can be applied under the assumption that all Gaifman graphs of all structures under consideration have treewidth bounded by a fixed parameter $k$. This is the case in our applications, as even the feedback vertex number is always bounded by $k$. For the remainder of this section we fix $k$.

Let $q\in \N$, $D\subseteq \Omega$, and $\Sigma$ be a binary signature. For a satisfiable type $\alpha\in \Types^{q,\Sigma}(D)$, we fix the {\em{representative}} of $\alpha$ to be any smallest (in terms of the total number of vertices and sizes of relations) boundaried $\Sigma$-structure $\Af$ satisfying the following:
\begin{itemize}[nosep]
 \item $\bnd \Af=D$;
 \item $\tp^q(\Af)=\alpha$;
 \item the treewidth of the Gaifman graph of $\Af$ is at most $k$.
\end{itemize}
The representative of a type $\alpha$ will be denoted by $\Rep(\alpha)$. Clearly, provided $\alpha$ is satisfiable, such a representative always exists. With the restriction on the treewidth present, one can use connections with the notion of recognizability and standard unpumping arguments (see~\cite{CEbook}) to give a computable bound: the following {\em{small model property}} is well-known.

\begin{theorem}\label{thm:small-model}
 For every satisfiable $\alpha\in \Types^{q,\Sigma}$, the size of $\Rep(\alpha)$ is bounded by a computable function of $q$, $k$, and $\Sigma$.
\end{theorem}

Observe that from Theorem~\ref{thm:small-model} it follows that given $\alpha$, one can compute $\Rep(\alpha)$. Indeed, it suffices to enumerate all $\Sigma$-structures of sizes up to the (computable) bound provided by Theorem~\ref{thm:small-model}, compute the type of each of them by brute force, and choose the smallest one that has type $\alpha$ (provided it exists; otherwise $\alpha$ is not satisfiable).
As a side note, this proves that some restriction on the structure of Gaifman graphs is necessary to state a result such as Theorem~\ref{thm:small-model}, for on general relational structures the theory of $\msotwo$ is undecidable~\cite{Seese91} (and decidability would follow from the small model property as described above). 

With this observation, we can show how, in the setting of Replacement Lemma, to compute the function $\Infer$ in time bounded by a computable function of $k$, $q$, and $\Sigma$. First compute~$p$. Then, for each $\beta\in \Types^{p,\Gamma^p}$, we would like to compute $\Infer(\beta)$. For this, compute $\wh{\Af}\coloneqq \Rep(\beta)$ and note that vertices and arcs of $\wh{\Af}$ bear the information on rank-$p$ types of ensemble elements that get contracted to them. For every arc $a\in A(\wh{\Af})$, let $\tau(a)$ be this type, and define $\tau(u)$ for $u\in V(\wh{\Af})$ and $\tau(\emptyset)$ analogously. Then for each $a\in A(\wh{\Af})$ we can compute $\Gf_a\coloneqq \Rep(\tau(a))$, and similarly we compute $\Gf_u\coloneqq \Rep(\tau(u))$ for $u\in V(G)$ and $\Gf_\emptyset\coloneqq \Rep(\tau(\emptyset))$. Let $\Xx$ be the ensemble consisting of all structures $\Gf_{\cdot}$ described above. Then from the Replacement Lemma it follows that $\Infer(\beta)=\tp^q(\Smash(\Xx))$. Noting that the total size of $\Smash(\Xx)$ is bounded by a computable function of $k$, $q$, and $\Sigma$, we may simply construct $\Smash(\Xx)$ and compute its rank-$q$ type by brute force.

\end{document}